\documentclass{CSML}
\pdfoutput=1

\usepackage{lastpage}
\newcommand{\dOi}{14(1:9)2018}
\lmcsheading{}{1--\pageref{LastPage}}{}{}%
{Feb.~01, 2017}{Jan.~23, 2018}{}


\ACMCCS{[{\bf Software and its engineering}]: Software organization and properties---Software functional properties---Formal methods---Software verification}

\usepackage{hyperref}
\usepackage{multirow}
\usepackage{subcaption}
\usepackage{mdframed}
\usepackage{gastex,delarray,wrapfig}
\usepackage{graphicx}
\usepackage{tikz}
\usepackage{listings}
\usetikzlibrary{automata}
\usetikzlibrary{trees}
\usetikzlibrary{shapes}
\usetikzlibrary{petri}
\usetikzlibrary{arrows}
\usetikzlibrary{decorations}
\usetikzlibrary{backgrounds}
\usepackage{bm}
\usepackage{bold-extra}
\usetikzlibrary{positioning}
\usetikzlibrary{calc}
\usetikzlibrary{fit}

\usepackage{amssymb}

\def\by{\xrightarrow}
\newcommand{\minnof}[1]{\minn\left({#1}\right)}
\newcommand{\minn}{\texttt{min}}
\newcommand{\upclosure}[1]{{#1}\uparrow}

\newcommand{\minpreof}[1]{\minpre\left({#1}\right)}

\newcommand{\Wsts}{{\sc Wsts}}
\newcommand{\m}[1]{\mathcal{#1}\hspace{-0.05cm}}

\newcommand{\rop}{{\sf r}}
\newcommand{\wop}{{\sf w}}
\newcommand{\arw}{{\sf arw}}
\newcommand{\nop}{{\sf nop}}
\newcommand{\fenceop}{{\sf fence}}
\newcommand{\updateop}{{\sf update}}
\newcommand{\propagate}{{\sf propagate}}
\newcommand{\delete}{{\sf delete}}

\newcommand{\ra}{\rightarrow}

\newcommand{\op}{{\it op}}

\newcommand{\flh}{\hookleftarrow}

\newcommand{\cprogram}{\m{P}}
\newcommand{\system}{\cprogram}

\newcommand{\matchingof}[1]{{\tt match}\left(#1\right)}
\newcommand{\processof}[1]{{\tt proc}\left(#1\right)}
\newcommand{\posof}[2]{{\tt pos}\left(#1,#2\right)}
\newcommand{\set}[1]{\left\{#1\right\}}
\newcommand{\setcomp}[2]{\left\{{#1}|\;{#2}\right\}}
\newcommand{\tuple}[1]{\left(#1\right)}
\newcommand{\vars}{{\mathbb X}}
\newcommand{\xvar}{x}
\newcommand{\yvar}{y}
\newcommand{\dataset}{{\mathbb V}}
\newcommand{\data}{v}
\newcommand{\valset}{{\mathbb V}}

\newcommand{\proc}{p}
\newcommand{\procs}{{\mathbb P}}
\newcommand{\automaton}{A}

\newcommand{\cstuple}{\tuple{\automaton_1,\automaton_2,\ldots,\automaton_n}}

\newcommand{\states}{Q}

\newcommand{\state}{q}
\newcommand{\initstate}{\state^{\it init}}
\newcommand{\transitions}{\Delta}
\newcommand{\transition}{t}

\newcommand{\targetof}[1]{{\it target}\left(#1\right)}

\newcommand{\confs}{C}
\newcommand{\allconfs}{{\tt C}}
\newcommand{\tsoconfs}{\allconfs_{{\sf TSO}}}
\newcommand{\dtsoconfs}{\allconfs_{\sf DTSO}}

\newcommand{\conf}{c}
\newcommand{\initconf}{c_{\it init}}
\newcommand{\dinitconf}{c^D_{\it init}}
\newcommand{\dconf}{d}
\newcommand{\initconfs}{{\tt Init}}

\newcommand{\tsys}{{\mathcal T}}
\newcommand{\actions}{{\tt Act}}
\newcommand{\tstuple}{\tuple{\allconfs,\initconfs,\actions,\cup_{a \in \actions}\by{a}}}

\newcommand{\mem}{{\bf mem}}

\newcommand{\initmem}{{\bf mem}_{\it init}}
\newcommand{\buffermapping}{{\bf \buffer}}

\newcommand{\initbuffermapping}{{\bf \buffer_{\it init}}}
\newcommand{\buffer}{b}

\newcommand{\statemapping}{{\bf \state}}

\newcommand{\initstatemapping}{{\bf \state_{\it init}}}

\newcommand{\onestatemappingof}[1]{{\tt states}\left(#1\right)}
\newcommand{\twostatemappingof}[2]{{\tt states}\left(#1\right)\left(#2\right)}

\newcommand{\onebuffermappingof}[1]{{\tt buffers}\left(#1\right)}
\newcommand{\twobuffermappingof}[2]{{\tt buffers}\left(#1\right)\left(#2\right)}
\newcommand{\threebuffermappingof}[3]{{\tt buffers}\left(#1\right)\left(#2\right)\left(#3\right)}

\newcommand{\onememoryof}[1]{{\tt mem}\left(#1\right)}
\newcommand{\twomemoryof}[2]{{\tt mem}\left(#1\right)\left(#2\right)}

\newcommand{\update}[2]{\left[#1\flh #2\right]}
\newcommand{\tsomovesto}[1]{\by{#1}_{\sf TSO}}
\newcommand{\dtsomovesto}[1]{\by{#1}_{{\sf DTSO}}}

\newcommand{\app}{\cdot}
\newcommand{\sizeof}[1]{|#1|}
\newcommand{\emptyword}{\epsilon}

\newcommand{\nat}{{\mathbb N}}

\newcommand{\word}{w}

\newcommand{\run}{\pi}
\newcommand{\comp}{\pi}
\newcommand{\apprun}{\bullet}
\newcommand{\tsocomp}{\comp_{\sf TSO}}
\newcommand{\sbcomp}{\comp_{\sf DTSO}}
\newcommand{\sbtotso}[1]{{\tt DTSO2TSO}\left(#1\right)}
\newcommand{\dtsototso}[1]{{\tt DTSO2TSO_+}\left(#1\right)}

\newcommand{\maxproc}{\proc_{\it max}}
\newcommand{\minproc}{\proc_{\it min}}
\newcommand{\procsuccof}[1]{{\it succ}\left(#1\right)}

\newcommand{\procordering}{\prec}

\newcommand{\numberof}[2]{\sharp\left({#1},{#2}\right)}

\newcommand{\preof}[2]{{\tt Pre}_{#1}\left(#2\right)}
\newcommand{\genordering}{\sqsubseteq}

\newcommand{\instance}{\it Inst}

\newcommand{\cordering}{\sqsubseteq}
\newcommand{\pconf}{\alpha}

\newcommand{\ownof}[1]{[{#1}]_{\it own}}
\newcommand{\wordering}{\preceq}

\newcommand{\minpre}{{\sf minpre}}

\newcommand{\parsys}{{\mathcal S}}

\def\dx{1.5cm} 
\def\dy{0.5cm}

\tikzset{
  state/.style={draw}
}

\newcommand{\newState}[4]{\node[state,#3](#1)[#4]{#2};}
\newcommand{\newTransition}[4]{\path[->] (#1) edge [#4] node {#3} (#2);}

\colorlet{bggreen}{green!20}

\tikzstyle{background rectangle}=
[rounded corners,inner frame sep=0pt]

\tikzstyle{autedge}=[->,color=red,line width=1pt]
\tikzstyle{labelnode}=[text=black,align=left,sloped,above=-2pt,font=\tiny]
\tikzstyle{fitnode}=[draw=black, rounded corners,line width=1pt]

\tikzstyle{typenode}=
[circle,draw,fill=yellow!20,inner sep=2pt]

\tikzstyle{equivedge}=[->,line width=1pt]
\tikzstyle{commentedge}=[->,color=red,line width=1pt,decorate, decoration={snake}]
\tikzstyle{trivialedge}=[->,color=red,dotted,line width=1pt]
\tikzstyle{impliesedge}=[->,color=black,line width=1pt]
\tikzstyle{transedge}=[->,color=black,line width=0.5pt,>=latex]

\tikzstyle{midnode}=[midway,circle,minimum size=3pt,inner sep=0mm,fill=black]

\tikzstyle{stabnode}=
[rounded corners,draw,inner sep=4pt]

\tikzstyle{commentnode}=
[ellipse callout,draw,inner sep=0.5pt]

\tikzstyle{blackdiagramnode}=
[rounded corners,fill=white,text=black,inner sep=1pt, inner ysep=2pt,anchor=north]
\tikzstyle{reddiagramnode}=
[rounded corners,fill=white,text=black,inner sep=1pt, inner ysep=2pt,anchor=north]
\tikzstyle{bluediagramnode}=
[rounded corners,fill=white,text=black,inner sep=1pt, inner ysep=2pt,anchor=north]

\definecolor{mcolor}{RGB}{200,255,200}
\tikzstyle{buffernode}=[draw,line width=1pt,minimum height=5mm,outer sep=0mm,inner sep=2pt]

\tikzstyle{statenode}=[draw,circle,inner sep=1pt]

\pgfdeclarelayer{mybackground}
\pgfdeclarelayer{background}
\pgfdeclarelayer{foreground}
\pgfsetlayers{background,mybackground,foreground,main}

\newcommand\xqed[1]{%
  \leavevmode\unskip\penalty9999 \hbox{}\nobreak\hfill
  \quad\hbox{#1}}
\newcommand\myend{\xqed{$\triangle$}}

\usepackage{xcolor}

\hypersetup{hidelinks}

\begin{document}

\title[A Load-Buffer Semantics for Total Store Ordering]{A Load-Buffer Semantics for Total Store Ordering\rsuper*}

\titlecomment{{\lsuper*} A preliminary version of this paper appeared as
at CONCUR'16 \cite{DBLP:conf/concur/AbdullaABN16}.}
\author[P.A.~Abdulla]{Parosh Aziz Abdulla\rsuper a}	
\address{\lsuper{a,b,d}Uppsala University, Sweden}	
\email{parosh@it.uu.se}
\email{mohamed\_faouzi.atig@it.uu.se}
\email{tuan-phong.ngo@it.uu.se}
\thanks{This work was supported in part by the Swedish Research Council and carried out within the Linnaeus centre of excellence UPMARC, Uppsala Programming for Multicore Architectures Research Center.}	
\author[M.F.~Atig]{Mohamed Faouzi Atig\rsuper b}	
\address{\vskip-6pt}
\author[A.~Bouajjani]{Ahmed Bouajjani\rsuper c}	
\address{ \lsuper{c}IRIF Universit\'e Paris Diderot - Paris 7, France}	
\email{abou@liafa.univ-paris-diderot.fr}  
\author[T.P.~Ngo]{Tuan Phong Ngo\rsuper d}
\address{\vskip-6pt}

\begin{abstract}
  \noindent We address the problem of verifying safety properties of concurrent programs running over the Total Store Order (TSO) memory model. Known decision procedures for this model are based on complex encodings of store buffers as lossy channels. These procedures assume that the number of processes is fixed. However, it is important in general to prove the correctness of a system/algorithm in a parametric way with an arbitrarily large number of processes.
  
In this paper, we introduce an alternative (yet equivalent) semantics to the classical one for the TSO semantics that is more amenable to efficient algorithmic verification and for the extension to parametric verification. For that, we adopt a {\em dual} view where {\em load buffers} are used instead of store buffers. The flow of information is now from the memory to load buffers. We show that this new semantics allows (1) to simplify drastically the safety analysis under TSO, (2) to obtain a spectacular gain in efficiency and scalability compared to existing procedures, and (3) to extend easily the decision procedure to the parametric case, which allows obtaining a new decidability result, and more importantly, a verification algorithm that is more general and more efficient in practice than the one for bounded instances.

\end{abstract}

\maketitle


\section{Introduction}
Most modern processor architectures execute instructions in an
out-of-order manner to gain efficiency.
In the context of {\it sequential} programming, this
out-of-order execution is transparent to the programmer
since one can still work
under the Sequential Consistency (SC) model~\cite{lamport-79}.
However, this is not true when
we consider concurrent processes that share the memory.
In fact, it turns out that concurrent algorithms
such as mutual exclusion and producer-consumer
protocols may not behave correctly any more.
Therefore, program verification is a relevant
(and difficult) task in order to prove correctness under the new
semantics.
The out-of-order execution of instructions
has led to the invention of new program semantics, so called
{\it Weak (or relaxed) Memory Models} (WMMs),
by allowing permutations between certain types of memory operations
\cite{adve-gharachorloo-96,DBLP:conf/isca/DuboisSB86,DBLP:conf/isca/AdveH90}.
 Total Store Ordering (TSO) is one of the the most common models, and
it corresponds to the relaxation
adopted by Sun's SPARC multiprocessors
\cite{sparc-v9-manual} and formalizations of the
x86-TSO memory model \cite{OSS2009,SSONM2010}.
These models put an unbounded perfect (non-lossy) {\it store buffer}
between each process and the main memory where
a store buffer carries the pending store operations of the process.
When a process performs a store operation, it appends it to the end of
its buffer.
These operations are propagated to the shared memory non-deterministically
in a FIFO manner.
When a process reads a variable, it searches its buffer
for a pending store operation on that variable.
If no such a store operation exists, it fetches the value
of the variable from the main memory.
Verifying programs running on the TSO memory model poses
a difficult challenge since the  unboundedness
of the buffers implies that
the state space of the system is infinite even in the case
where the input program is finite-state.
Decidability of safety properties
has been obtained by constructing equivalent models
that replace the perfect store buffer by {\em lossy} channels
\cite{ABBM10,ABBM12,DBLP:conf/tacas/AbdullaACLR12}.
However, these constructions are complicated and
involve several ingredients that lead to
inefficient verification procedures.
For instance, they require each message inside a lossy channel to carry
(instead of a single store operation)
a full snapshot of the memory representing a local view of the memory
contents by the process.
Furthermore, the reductions involve non-deterministic guessing the lossy channel contents.
The guessing is then resolved either by consistency checking
\cite{ABBM10} or by using explicit pointer variables
(each corresponding to one process) inside the buffers
\cite{DBLP:conf/tacas/AbdullaACLR12}, causing
a serious state space explosion problem.

In this paper, we introduce a novel semantics which we call
the {\it Dual TSO} semantics. Our aim is to provide an alternative (and equivalent) semantics that is more amenable for efficient algorithmic verification.
The main idea is to have {\it load buffers} that contain
pending load operations (more precisely, values that will potentially be taken by forthcoming load operations) rather than store buffers (that contain store operations).
%
%
The flow of information will now be in the reverse direction, i.e.,
store operations are performed by the processes atomically on the main memory,
while values of variables are propagated non-deterministically  from the memory to the load
buffers of the processes.
When a process performs a load operation, it can fetch the value of the variable
from the head of its load buffer.  
We show that the Dual TSO  semantics is equivalent to the original
one in the sense that any given set of processes 
will reach the same set of local
states under both semantics.
The Dual TSO semantics allows us to understand 
the TSO model
in a totally different way compared to the classical semantics.
Furthermore, the Dual TSO semantics offers several important advantages
from the point of view of formal reasoning and program verification.
First,  the Dual TSO semantics allows transforming the load buffers
to {\it lossy} channels without adding the 
costly overhead that was necessary in the case of  store buffers.   This means that we can assume w.l.o.g. that any message in the load buffers (except a finite number of messages) can be lost in non-deterministic manner. 
Hence, we can apply the theory of 
{\it well-structured systems} \cite{DBLP:journals/bsl/Abdulla10,abdulla-general-96,FinkelS01} in a straightforward manner leading 
to a much simpler proof of decidability of safety properties.
Second, the absence of extra overhead means  that we obtain 
more efficient algorithms and better scalability 
(as shown by our experimental results).
Finally, the Dual TSO semantics allows extending the framework to
perform {\it parameterized verification} which is an
important paradigm in concurrent program verification.
Here, we consider systems, e.g., mutual exclusion protocols,
that consist of an arbitrary number
of processes.
The aim of parameterized verification is
to prove correctness of the system regardless
of the number of processes.
It is not obvious how
to perform parameterized verification under the classical semantics.
For instance, extending the framework of
\cite{DBLP:conf/tacas/AbdullaACLR12}, 
 would involve an unbounded number of pointer variables, thus
leading to channel systems with unbounded message 
alphabets.
In contrast, as we show in this paper, the simple nature
of the Dual TSO semantics allows a straightforward extension
of our verification algorithm to the case of parameterized verification. This is the first time a decidability result is established for the parametrized verification of programs running over WMMs. Notice that this result is taking into account two sources of infinity: the number of processes and the size of the buffers. 

Based on our framework, we have implemented a tool and applied
it to a large set of benchmarks.
The experiments demonstrate the efficiency of the Dual TSO semantics compared to the classical one (by two order of magnitude in average), and the feasibility of parametrized  verification in the former case. In fact, besides its theoretical generality, parametrized verification is practically crucial in this setting: as our experiments show, it is much more efficient than verification of bounded-size instances (starting from a number of components of 3 or 4), especially concerning memory consumption (which also is a critical resource).

\paragraph*{\bf Related Work.}
There have been a lot of works related to the analysis  of programs running under WMMs  (e.g.,~\cite{DBLP:conf/pldi/LiuNPVY12,KVY2010,KVY2011,DBLP:conf/sas/DanMVY13,DBLP:conf/tacas/AbdullaACLR12,BM2008,BSS2011,DBLP:conf/esop/BouajjaniDM13,BAM07,yang-gopalakrishnan-PDPS04,DBLP:conf/netys/AbdullaALN15,DBLP:conf/sas/AbdullaACLR12,DBLP:conf/cav/AbdullaAJL16,Dan201762,DBLP:conf/ictac/TravkinW16,DBLP:conf/fm/LahavV16,DBLP:conf/icalp/LahavV15,DBLP:conf/popl/Vafeiadis15,DBLP:conf/pdp/HeVQF16}). 
Some of these works propose precise analysis techniques for checking safety properties or stability  of finite-state programs under WMMs (e.g.,~\cite{DBLP:conf/tacas/AbdullaACLR12,DBLP:conf/esop/BouajjaniDM13,DM14,DBLP:conf/esop/AbdullaAP15,DBLP:conf/netys/AbdullaALN15}). 
Others propose context-bounded analyzing techniques (e.g., \cite{AtigBP11,fmcad16,eps402285,abdullaABN17}) or stateless model-checking techniques (e.g., \cite{tacas15:tso,Zhang:pldi15,DBLP:conf/oopsla/DemskyL15,oopsla16}) for programs under TSO and PSO. Different other techniques based on monitoring and testing have also been developed during these last years (e.g.,~\cite{BM2008,BSS2011,DBLP:conf/pldi/LiuNPVY12}). There are also a number of efforts to design bounded model checking techniques for programs under  WMMs  (e.g.,~\cite{DBLP:conf/esop/AlglaveKNT13,AlglaveKT13,yang-gopalakrishnan-PDPS04,BAM07}) which encode the verification problem in SAT/SMT.

The  closest works to ours are those presented in~\cite{DBLP:conf/tacas/AbdullaACLR12,ABBM10,DBLP:conf/tacas/AbdullaACLR13,ABBM12} which provide precise and sound techniques  for checking  safety properties for finite-state programs running under TSO. However, as stated in the introduction, these    techniques  are  complicated and can not be extended, in a  straightforward manner, to the verification of parameterized systems  (as it is the case of the developed techniques  for the Dual TSO semantics). 
In Section~\ref{experiments:section}, we experimentally compare our techniques   with {\sf Memorax}~\cite{DBLP:conf/tacas/AbdullaACLR12,DBLP:conf/tacas/AbdullaACLR13} which  is the {only precise and sound tool} for  checking  safety properties for concurrent programs under TSO.



\section{Preliminaries}
\label{sec:prels}

Let $\Sigma$ be a finite alphabet. We use $\Sigma^{*}$ (resp. $\Sigma^{+}$) to denote  the set
of all \textit{words} (resp. non-empty words) over $\Sigma$. Let  $\emptyword$ be the empty word.
The length of a word $\word \in \Sigma^*$ is denoted by $\sizeof \word$ (and in particular $|\epsilon| = 0$).
For every $i:1\leq i\leq\sizeof\word$,
let $\word(i)$  be the symbol at position $i$ in $\word$.
For $a \in \Sigma$, we write $a \in \word$ if $a$ appears in $\word$, i.e.,
$a = \word(i)$ for some $i:1\leq i\leq\sizeof \word$.

Given two words $u$ and $v$ over $\Sigma$, we use  $u \preceq v$ to denote that $u $ is a (not necessarily contiguous) subword of $v$, i.e., if there is an injection $h: \{1,\ldots,|u|\} \mapsto \{1,\ldots, |v|\}$  such that: $(1)$ $h(i) < h(j)$ for all $i, j: 1\leq i <j \leq |u|$ and $(2)$ for every $i: 1\leq i \leq |u|$, we have $u(i)=v(h(i))$.

Given a subset $\Sigma' \subseteq \Sigma$ and a word $\word \in \Sigma^*$, we use $\word|_{\Sigma'}$ to denote the projection of $\word$ over $\Sigma'$, i.e., the word obtained from $\word$ by erasing all the symbols that are not in $\Sigma'$.

Let $A$ and $B$ be two sets and let $f:A\mapsto B$ be a total function from $A$ to $B$.  We use $f[a\flh b]$
to denote the function $g$ such that $g(a)=b$ and $g(x)=f(x)$ for all $x\neq a$.
%

A transition system $\tsys$ is a tuple $\tstuple$ where
$\allconfs$ is a (potentially infinite) set of {\em configurations};
$\initconfs\subseteq\allconfs$ is a set of {\em initial configurations}; $\actions$ is a  set of actions; 
and for every $a \in \actions $, $\by{a}\subseteq\allconfs\times \allconfs$
is a {\em transition relation}.
We use  $\conf\by{a}\conf'$ to denote that
$\tuple{\conf,\conf'}\in\by{a}$. Let  $\by{} := \cup_{a\in \actions} \by{a}$.    %

A {\em run} $\run$ of $\tsys$ is of the form
$\conf_0\by{a_1}\conf_1\by{a_2}\cdots\by{a_n}\conf_n$
where
$\conf_i\by{a_{i+1}}\conf_{i+1}$ for all $i:0\leq i <n$.
Then, we write
$\conf_0\by{\run}\conf_n$.
We use $\targetof{\run}$ to denote the configuration $\conf_n$.
The run $\run$ is said to be a {\em computation}
if $\conf_0\in\initconfs$.
Two runs
$\run_1=\conf_0\by{a_1}\conf_1\by{a_2}\cdots\by{a_m}\conf_m$ and
$\run_2=\conf_{m+1}\by{a_{m+2}}\conf_{m+2}\by{a_{m+3}}\cdots\by{a_n}\conf_n$
 are  {\em compatible} if
$\conf_m=\conf_{m+1}$.
Then, we write
$\run_1\apprun\run_2$ to denote the run
$$\run=\conf_0\by{a_1}\conf_1\by{a_2}\cdots\by{a_m}\conf_m\by{a_{m+2}}\conf_{m+2}\by{a_{m+3}}\cdots\by{a_n}\conf_n.$$
%
For two configurations $\conf$ and $\conf'$, we use 
$\conf\by{*}\conf'$ to denote that 
$\conf\by\run\conf'$ for some run $\run$.
A configuration $\conf$ is said to be
{\it reachable} in $\tsys$ if $\conf_0\by{*}\conf$ for some
$\conf_0\in\initconfs$, and a set $\confs$ of configurations
is said to be 
{\it reachable} in $\tsys$ if some $\conf\in\confs$ is reachable in $\tsys$.



\section{Concurrent Systems}
\label{definitions:section}

In this section, we define the syntax  we use for  {\em concurrent programs}, a model for representing   communication of  concurrent processes. Communication between processes is performed through a shared memory that consists of a finite number of shared variables
(over finite domains) to which all processes can read and write. Then we recall the classical TSO semantics including the transition system it induces and its reachability problem. Next, we  introduce the {\em Dual TSO} semantics and its induced   transition system. Finally, we state the equivalence between the two semantics; i.e., for  a given concurrent  program, we can reduce its reachability problem  under the classical TSO semantics to its reachability problem under Dual TSO semantics and vice-versa.

\subsection{Syntax}

Let $\valset$ be a finite data domain and    $\vars$ be a finite set of variables. We assume w.l.o.g. that $\valset$ contains the value $0$. Let    $\Omega(\vars,\valset)$  be   the smallest set of memory operations that contains with  $\xvar\in\vars$ 
and $\data, \data'\in\dataset$:
\begin{enumerate}
\item
{\em ``no''} operation $\nop$,
\item
{\em read} operation $\rop(\xvar,\data)$,
\item  
{\em write} operation $\wop(\xvar,\data)$,
\item   
{\em fence} operation $\fenceop$,  and
\item
{\em atomic read-write} operation
$\arw(\xvar,\data,\data')$.
\end{enumerate}

A {\em  concurrent system} (or a concurrent program) is a tuple
$\cprogram=\cstuple$ where for every $\proc: 1 \leq \proc \leq n$,  $\automaton_\proc$ is a  finite-state automaton describing the behavior of the process $\proc$.
The automaton $\automaton_\proc$ is defined as a triple 
$\tuple{\states_\proc,\initstate_\proc,\transitions_\proc}$ where
$\states_\proc$ is a finite set of {\em local states},
$\initstate_\proc\in\states_\proc$ is the {\em initial} local state,
and $\transitions_\proc \subseteq \states_\proc \times \Omega(\vars,\valset) \times \states_\proc$ is a finite set of {\em transitions}.
 We define $\procs:=\{1,\ldots,n\}$  to be  the set of process IDs, $\states:=\cup_{\proc\in\procs}\states_\proc$ to be the set of all local states and $\transitions:=\cup_{\proc\in\procs}\transitions_\proc$ to be the set of all transitions.

\begin{exa}
\begin{figure}
\begin{tikzpicture}[node distance=\dy and \dx,
  >=latex,shorten >=2pt,shorten <=2pt,auto,
  semithick,   initial distance=1cm,
  every initial by arrow/.style={*->}
  ]
  \newState{q-init-p1}{$\state_0$}{initial left}{}
  \newState{q11}{$\state_1$}{right=of q-init-p1}{}
  \newState{q21}{$\state_2$}{right=of q11}{accepting}
  \node[right=of q21,](a1)[]{\Large $\automaton_{1}$};
  \newTransition{q-init-p1}{q11}{$\wop(\xvar,2)$}{}
  \newTransition{q11}{q21}{$\rop(\yvar,0)$}{}
 
  \newState{q-init-p2}{$\state'_0$}{initial left, below=of q-init-p1}{}
   \newState{q12}{$\state'_1$}{right=of q-init-p2}{}
  \newState{q22}{$\state'_2$}{right=of q12}{}
   \newState{q32}{$\state'_3$}{right=of q22}{accepting}
   \node[right=of q32](a2)[]{\Large $\automaton_{2}$};
   \newTransition{q-init-p2}{q12}{$\wop(\yvar,1)$}{}
  \newTransition{q12}{q22}{$\wop(\xvar,1)$}{}
   \newTransition{q22}{q32}{$\rop(\xvar,2)$}{}
 \end{tikzpicture}
\caption{An example of a concurrent system $\cprogram=\set{\automaton_{1}, \automaton_2}$.}
\label{example:fig}
\end{figure}

Figure~\ref{example:fig} shows an example of a concurrent system $\cprogram=\set{\automaton_{1}, \automaton_2}$ consisting of  two concurrent processes,
called $\proc_1$ and $\proc_2$. Communication between processes is performed through two shared variables $\xvar$ and $\yvar$
to which the processes can read and write. 
The automaton
$\automaton_{1}$
is defined as a triple
$\tuple{\set{\state_0,\state_1,\state_2}, \set{\state_0}, \set{\tuple{\state_0,\wop(\xvar,2),\state_1}, \tuple{\state_1,\rop(\yvar,0),\state_2}}}$.
Similarly,
$\automaton_{2}=\tuple{\set{\state'_0,\state'_1,\state'_2,\state'_3}, \set{\state'_0}, \set{\tuple{\state'_0,\wop(\yvar,1),\state'_1}, \tuple{\state'_1,\wop(\xvar,1),\state'_2}, \tuple{\state'_2,\rop(\xvar,2),\state'_3}}}$.
\myend
\end{exa}

\subsection{Classical TSO Semantics}

In the following, we recall  the semantics of concurrent systems under the classical TSO model as formalized in ~\cite{OSS2009,SSONM2010}. To do that, we define the set of configurations and the induced  transition relation. Let $\system=\cstuple$ be a concurrent system.

\paragraph*{\bf TSO-configurations.}
A {\it TSO-configuration}
$\conf$ is a triple $\tuple{\statemapping,\buffermapping,\mem}$
where: 
\begin{enumerate}
\item
$\statemapping: \procs \mapsto\states $  is the {\em global state} of $\cprogram$,    mapping  each process $\proc \in \procs$ to  a   local state in $\states_\proc$ (i.e., 
$\statemapping(\proc)\in\states_\proc$). 
\item
$\buffermapping:\procs \mapsto\left(\vars\times\dataset\right)^*$ gives the content of the store buffer of each process.  
\item
$\mem:\vars\mapsto\dataset$ defines the value of each  shared variable. 
\end{enumerate}
Observe that the store buffer  of each process contains a sequence of write operations, where
each write operation is defined by a pair, namely a variable
$\xvar$ and a value $\data$ that is assigned to $\xvar$. 

The  {\it initial} TSO-configuration $\initconf$  is defined by the tuple  $\tuple{\initstatemapping,\initbuffermapping,\initmem}$
where, for all $\proc \in \procs$ and $\xvar \in \vars$, we have that
$\initstatemapping(\proc)=\initstate_\proc$,
$\initbuffermapping(\proc)=\emptyword$ and $\initmem(\xvar)=0$.
In other words, each process is in its initial local state, all
the buffers are empty, and all the variables in the shared memory are initialized to $0$.
We  use
 $\tsoconfs$ to denote the set of  all TSO-configurations.

{
\begin{figure}[t!h]
\begin{mdframed}
  \centering
  
    \begin{tabular}{@{}c@{}}
      \begin{tabular}{@{}c@{}}
     $ \transition=\tuple{\state,\nop,\state'}  \;\;\;\; \; \statemapping(\proc)=\state$\\
              \hline
$\tuple{\statemapping,\buffermapping,\mem} \tsomovesto{\transition}\tuple{\statemapping\update{\proc}{\state'},\buffermapping,\mem}$
            \end{tabular}
 \begin{tabular}{@{}c@{}}
        \\
             \;\;   {\sf Nop}\\
        \\
      \end{tabular}
      \\
      \begin{tabular}{@{}c@{}}
     $ \transition=\tuple{\state,\wop(\xvar,\data),\state'}   \;\;\;\;\;  \statemapping(\proc)=\state$ \\
              \hline
$\tuple{\statemapping,\buffermapping,\mem} \tsomovesto{\transition}\tuple{\statemapping\update{\proc}{\state'},\buffermapping\update{\proc}{(\xvar,\data)\app  \buffermapping(\proc)},\mem}$
            \end{tabular}
 \begin{tabular}{@{}c@{}}
        \\
             \;\;  {\sf Write}\\
        \\
      \end{tabular}
           \\
      \begin{tabular}{@{}c@{}}
      $\transition=\updateop_{\proc} $  \\
              \hline
$\tuple{\statemapping,\buffermapping\update{\proc}{  \buffermapping(\proc) \app(\xvar,\data)},\mem} \tsomovesto{\transition}\tuple{\statemapping,\buffermapping,\mem\update{\xvar}{\data}}$
            \end{tabular}
 \begin{tabular}{@{}c@{}}
        \\
             \;\;   {\sf Update}\\
        \\
      \end{tabular}
            \\
      \begin{tabular}{@{}c@{}}
      $\transition=\tuple{\state,\rop(\xvar,\data),\state'}   \;\;\;\; \; \statemapping(\proc)=\state  \;\;\;\; \;\buffermapping(\proc)|_{\{\xvar\} \times \dataset} =\tuple{\xvar,\data} \app \word$ \\
              \hline
$\tuple{\statemapping,\buffermapping,\mem} \tsomovesto{\transition}\tuple{\statemapping\update{\proc}{\state'},\buffermapping,\mem}$
            \end{tabular}
 \begin{tabular}{@{}c@{}}
        \\
             \;\;   {\sf Read-Own-Write}\\
        \\
      \end{tabular}
                  \\
      \begin{tabular}{@{}c@{}}
    $  \transition=\tuple{\state,\rop(\xvar,\data),\state'}   \;\;\;\; \; \statemapping(\proc)=\state  \;\;\;\; \;\buffermapping(\proc)|_{\{\xvar\} \times \dataset} =\epsilon   \;\;\;\; \; \mem(\xvar)=\data$\\
              \hline
$\tuple{\statemapping,\buffermapping,\mem} \tsomovesto{\transition}\tuple{\statemapping\update{\proc}{\state'},\buffermapping,\mem}$
            \end{tabular}
 \begin{tabular}{@{}c@{}}
        \\
             \;\;   {\sf Read from Memory}\\
        \\
      \end{tabular}
      \\
      \begin{tabular}{@{}c@{}}
    $  \transition=\tuple{\state,\arw(\xvar,\data,\data'),\state'}   \;\;\;\; \; \statemapping(\proc)=\state  \;\;\;\; \;\buffermapping(\proc)=\epsilon   \;\;\;\; \; \mem(\xvar)=\data$\\
              \hline
$\tuple{\statemapping,\buffermapping,\mem} \tsomovesto{\transition}\tuple{\statemapping\update{\proc}{\state'},\buffermapping,\mem\update{\xvar}{\data'}}$
            \end{tabular}
 \begin{tabular}{@{}c@{}}
        \\
             \;\;  {\sf ARW}\\
        \\
      \end{tabular}  
            \\
      \begin{tabular}{@{}c@{}}
    $  \transition=\tuple{\state,\fenceop,\state'}   \;\;\;\; \; \statemapping(\proc)=\state  \;\;\;\; \;\buffermapping(\proc)=\epsilon  $\\
              \hline
$\tuple{\statemapping,\buffermapping,\mem} \tsomovesto{\transition}\tuple{\statemapping\update{\proc}{\state'},\buffermapping,\mem}$
            \end{tabular}
 \begin{tabular}{@{}c@{}}
        \\
             \;\;  {\sf Fence}\\
        \\
      \end{tabular}  
 \end{tabular}
    \end{mdframed}

  \caption{The  transition relation  $\tsomovesto{}$ under  TSO semantics. Here, process $\proc \in \procs$ and transition $\transition \in \transitions_\proc \cup\set{\updateop_{\proc}}$  where
$\updateop_\proc$ is a transition that updates the memory
using the oldest message in the buffer of the process $\proc$.
  }
  
  \label{fig:oper:sem}
  
\end{figure} }

\paragraph*{\bf TSO-transition Relation.}
The {\it transition relation} $\tsomovesto{}$ between TSO-configurations is given by a set of rules, described in Figure~\ref{fig:oper:sem}. Here, we informally explain  these rules.  
A  {\it nop} transition   $\tuple{\state,\nop,\state'} \in \transitions_\proc$  changes  only  the local  state of the  process $\proc$ from $\state$ to $\state'$. 
A {\it write} transition $\tuple{\state,\wop(\xvar,\data),\state'} \in \transitions_\proc$   adds a new message $(\xvar,\data)$ to the tail 
of  the store buffer of the process $\proc$.  
A memory {\it update} transition  $\updateop_{\proc}$ can be performed at any time by removing the (oldest) message 
at the head 
of the store buffer of the process $\proc$ and updating the memory accordingly.  
For a  {\it read} transition $\tuple{\state,\rop(\xvar,\data),\state'} \in \transitions_\proc$, if the store buffer of the process  $\proc$ contains some write operations to $\xvar$, then the read value $\data$ must correspond to the value of the most recent 
such a write operation. Otherwise, the value $\data$ of $\xvar$ is fetched from the memory. 
A {\it fence} transition  $\tuple{\state,\fenceop,\state'}\in \transitions_\proc$ can be performed by the process $\proc$ only if its store buffer is empty. 
Finally, an {\it atomic read-write} transition  $\tuple{\state,\arw(\xvar,\data,\data'),\state'} \in \transitions_\proc$ can be performed by the process $\proc$ only if its store buffer is empty. This transition checks  whether the value of $\xvar$ in the memory is $\data$ and then changes it to $\data'$.

Let $\transitions':=\setcomp{\updateop_\proc}{\proc\in\procs}$, i.e.,\
$\transitions'$ contains all memory update transitions.
We use $\conf\tsomovesto{}\conf'$ to denote that
$\conf\tsomovesto{\transition}\conf'$ for some
$\transition\in\transitions\cup\transitions'$.
The transition system induced by  $\cprogram$ under the classical TSO semantics is then given by
$\tsys_{{\sf TSO}}:=\tuple{\tsoconfs,\{\initconf\},\transitions\cup\transitions',\tsomovesto{}}$.

\paragraph*{\bf The TSO Reachability Problem.} 
  A  global state $\statemapping_{target}$ is said to be {\it reachable} in   $\tsys_{\sf TSO}$  if and only if 
there  is a  TSO-configuration $\conf$ of the form ${\tuple{\statemapping_{target},\buffermapping,\mem}}$, with ${\buffermapping(\proc)=\epsilon  \; \; \text{for all} \; \proc \in \procs}$, such that $\conf$ is reachable  in $\tsys_{\sf TSO}$. 

The {\it TSO reachability problem} for   the concurrent system $\cprogram$ under the TSO semantics asks,
for a given  global state $\statemapping_{target}$, whether $\statemapping_{target}$ is reachable in $\tsys_{\sf TSO}$. 
Observe that, in the definition of the reachability problem,  we require that  the buffers of the configuration $\conf$  must be empty instead of being arbitrary. This is only for the sake of simplicity and does not constitute a restriction. Indeed, we can easily show that the  ``arbitrary buffer'' reachability problem is reducible to the  ``empty buffer'' reachability problem.  

 \begin{figure}[t]
\centering
 \begin{tikzpicture}
[background rectangle/.style={rounded corners,line width=1pt},show background rectangle]
\node [name=dummy]{};

\node[rectangle,rounded corners,align=center,minimum  height=20mm,font=\footnotesize,name=memory,line width=0.5pt,draw] 
at (dummy) {$x=0$\\$y=0$};
\node[font=\small,name=memoryrname] at ($(memory.west)+(5mm,12mm)$){memory};

\node[buffernode,font=\footnotesize,line width=0.5pt,anchor=east,name=b11] at ($(memory.west)+(-19mm,7mm)$){$\tuple{\xvar,2}$};
\node[buffernode,font=\footnotesize,line width=0.5pt,anchor=east,name=b21] at ($(memory.west)+(-15mm,-7mm)$){$\tuple{\yvar,1}$};
\node[buffernode,font=\footnotesize,line width=0.5pt,anchor=east,name=b22] at ($(b21.west)+(0mm,0mm)$){$\tuple{\xvar,1}$};

\node[font=\small,name=buffername] at ($(b11.west)+(5mm,5mm)$){store buffers};
\node[font=\small,name=tailname] at ($(b11.west)+(-2mm,-4mm)$){tail};
\node[font=\small,name=headname] at ($(b11.west)+(10mm,-4mm)$){head};

\node[name=q3,anchor=east,font=\footnotesize] at ($(b11.west)+(-18mm,-1.6mm)$) {$\tuple{\state_1,\rop(\yvar,0),\state_2}$};

\draw[fill=black] ($(q3.west)+(1.0mm,0mm)$) -- ($(q3.west)+(-1.0mm,1mm)$) -- ($(q3.west)+(-1.0mm,-0.7mm)$) -- cycle;
\node[name=p1,anchor=west,font=\footnotesize] at ($(q3.west)+(0mm,5mm)$){$\proc_1: \statemapping(\proc_1)=\state_1$};

\node [rounded corners,fit=(p1) (q3), name=proc1,line width=0.5pt,draw] {};
\node[font=\small,name=memoryrname] at ($(proc1.west)+(12mm,9mm)$){processes};

\node[name=s3,anchor=east,font=\footnotesize] at ($(b21.west)+(-22mm,-1.6mm)$) {$\tuple{\state'_2,\rop(\xvar,2),\state'_3}$};

\draw[fill=black] ($(s3.west)+(1.0mm,0mm)$) -- ($(s3.west)+(-1.0mm,1mm)$) -- ($(s3.west)+(-1.0mm,-0.7mm)$) -- cycle;
\node[name=p2,anchor=west,font=\footnotesize] at ($(s3.west)+(0mm,5mm)$){$\proc_2: \statemapping(\proc_2)=\state'_2$};

\node [rounded corners,fit=(p2)(s3),name=proc2,line width=0.5pt,draw] {};

\draw[transedge] ($(proc1.east)+(0mm,-1mm)$) -- (b11);
\draw[transedge] ($(proc2.east)+(0mm,-1mm)$) -- (b22);
\draw[transedge] (b11) -- ($(b11)+(24mm,0mm)$);
\draw[transedge] (b21) -- ($(b21)+(19.5mm,0mm)$);

\draw[transedge,dashed,line width=1pt,font=\footnotesize] ($(memory.west)+(0mm,-3mm)$)  to [in=320,out=150]
node[above, very near end, sloped]{$\rop(\yvar,0)$}  ($(proc1.east)+(0mm,-6mm)$);

\end{tikzpicture}
\caption{A reachable TSO-configuration  of the concurrent system in Figure~\ref{example:fig}.}
\label{reachable:tso-configuration}
\end{figure}

 \begin{figure}[t]
\centering
 \begin{tikzpicture}
[background rectangle/.style={rounded corners,line width=1pt},show background rectangle]
\node [name=dummy]{};

\node[rectangle,rounded corners,align=center,minimum  height=20mm,font=\footnotesize,name=memory,line width=0.5pt,draw] 
at (dummy) {$x=2$\\$y=1$};
\node[font=\small,name=memoryrname] at ($(memory.west)+(5mm,12mm)$){memory};



\node[name=p1,anchor=west,font=\footnotesize] at ($(q3.west)+(0mm,5mm)$){$\proc_1: \statemapping(\proc_1)=\state_2$};

\node [rounded corners,fit=(p1) (q3), name=proc1,line width=0.5pt,draw] {};
\node[font=\small,name=memoryrname] at ($(proc1.west)+(12mm,9mm)$){processes};


\node[name=p2,anchor=west,font=\footnotesize] at ($(s3.west)+(0mm,5mm)$){$\proc_2: \statemapping(\proc_2)=\state'_3$};

\node [rounded corners,fit=(p2)(s3),name=proc2,line width=0.5pt,draw] {};

\draw[transedge] ($(proc1.east)+(0mm,-1mm)$) -- ($(b11)+(24mm,0mm)$);
\draw[transedge] ($(proc2.east)+(0mm,-1mm)$) -- ($(b21)+(19.5mm,0mm)$);


\end{tikzpicture}
\caption{A reachable ``empty buffer''   TSO-configuration of the concurrent system in Figure~\ref{example:fig}.}
\label{reachable:tso-configuration1}
\end{figure}

\begin{exa}
\label{tso:example}

Figure~\ref{reachable:tso-configuration} illustrates a TSO-configuration $\conf$
that can be reached from the initial configuration $\initconf$
of the concurrent system in Figure~\ref{example:fig}.
To reach  this configuration,
the process $\proc_1$ has executed the  write
transition
$\tuple{\state_0,\wop(\xvar,2),\state_1}$
and
appended  
the message
$\tuple{\xvar,2}$
to its store buffer.
Meanwhile,
the process $\proc_2$ has executed two write
transitions
$\tuple{\state'_0,\wop(\yvar,1),\state'_1}$
and 
$\tuple{\state'_1,\wop(\xvar,1),\state'_2}$.
Hence, the store buffer of $\proc_2$ contains  the sequence  
$\tuple{\xvar,1}\app\tuple{\yvar,1}$.
Now, the process $\proc_1$
can perform
the  read transition
$\tuple{\state_1,\rop(\yvar,0),\state_2}$.
Since the buffer of $\proc_1$ does not contain
any pending write message on $\yvar$, 
the read value 
is fetched from the memory (represented by the dash arrow in Figure~\ref{reachable:tso-configuration}).
Then,
$\proc_1$
and $\proc_2$
perform
the following  sequence
of update transitions
$\updateop_{\proc_2}\app\updateop_{\proc_2}\app\updateop_{\proc_1}$
to empty their buffers
and update the memory to $\xvar=2$ and $\yvar=1$.
Finally,
$\proc_2$ performs the read transition
$\tuple{\state'_2,\rop(\xvar,2),\state'_3}$ (by 
reading from the memory) to reach to the configuration $\conf_{target}$ given in Figure~\ref{reachable:tso-configuration1}.
Observe that the buffers of both  processes are empty in $\conf_{target}$. Let $\statemapping_{target}$ be  the global state
in $\conf_{target}$
 defined as follows:
$\statemapping_{target}(\proc_1)=\state_2$
and 
$\statemapping_{target}(\proc_2)=\state'_3$. Therefore,
we can say that
the global state
$\statemapping_{target}$
is reachable
in 
$\tsys_{\sf TSO}$. 
\myend

\end{exa}


{
\begin{figure}[t]
\begin{mdframed}
  \centering

    \begin{tabular}{@{}c@{}}
      \begin{tabular}{@{}c@{}}
     $ \transition=\tuple{\state,\nop,\state'}  \;\;\;\; \; \statemapping(\proc)=\state$\\
              \hline
$\tuple{\statemapping,\buffermapping,\mem} \dtsomovesto{\transition}\tuple{\statemapping\update{\proc}{\state'},\buffermapping,\mem}$
            \end{tabular}
 \begin{tabular}{@{}c@{}}
        \\
             \;\;   {\sf Nop}\\
        \\
      \end{tabular}
      \\
      \begin{tabular}{@{}c@{}}
     $ \transition=\tuple{\state,\wop(\xvar,\data),\state'}   \;\;\;\;\;  \statemapping(\proc)=\state$ \\
              \hline
$\tuple{\statemapping,\buffermapping,\mem} \dtsomovesto{\transition}\tuple{\statemapping\update{\proc}{\state'},\buffermapping\update{\proc}{(\xvar,\data,{own})\app  \buffermapping(\proc)},\mem\update{\xvar}{\data}}$
            \end{tabular}
 \begin{tabular}{@{}c@{}}
        \\
             \;\;  {\sf Write}\\
        \\
      \end{tabular}
           \\
      \begin{tabular}{@{}c@{}}
      $\transition=\propagate_\proc^\xvar  \;\;\;\; \;  \mem(\xvar)=\data$  \\
              \hline
$\tuple{\statemapping,\buffermapping,\mem} \dtsomovesto{\transition}\tuple{\statemapping,\buffermapping\update{\proc}{(\xvar,\data)\app  \buffermapping(\proc)},\mem}$
            \end{tabular}
 \begin{tabular}{@{}c@{}}
        \\
             \;\;   {\sf Propagate}\\
        \\
      \end{tabular}
            \\
                  \begin{tabular}{@{}c@{}}
      $\transition=\delete_\proc  \;\;\;\; \; |m|=1 $  \\
              \hline
$\tuple{\statemapping,\buffermapping\update{\proc}{  \buffermapping(\proc) \app m},\mem} \dtsomovesto{\transition}\tuple{\statemapping,\buffermapping,\mem}$
            \end{tabular}
 \begin{tabular}{@{}c@{}}
        \\
             \;\;   {\sf Delete}\\
        \\
      \end{tabular}
            \\
                  \begin{tabular}{@{}c@{}}
      $\transition=\tuple{\state,\rop(\xvar,\data),\state'}   \;\;\;\; \; \statemapping(\proc)=\state  \;\;\;\; \;\buffermapping(\proc)|_{\{\xvar\} \times \dataset \times \{{own}\}} =\tuple{\xvar,\data,{own}} \app \word$ \\
              \hline
$\tuple{\statemapping,\buffermapping,\mem} \dtsomovesto{\transition}\tuple{\statemapping\update{\proc}{\state'},\buffermapping,\mem}$
            \end{tabular}
 \begin{tabular}{@{}c@{}}
        \\
             \;\;   {\sf Read-Own-Write}\\
        \\
      \end{tabular}
                  \\
      \begin{tabular}{@{}c@{}}
    $  \transition=\tuple{\state,\rop(\xvar,\data),\state'}   \;\;\; \statemapping(\proc)=\state  \;\;\;\buffermapping(\proc)|_{\{\xvar\} \times \dataset \times \{{own}\}} =\epsilon   \;\;\; \buffermapping(\proc)= \word \app \tuple{\xvar,\data}$\\
              \hline
$\tuple{\statemapping,\buffermapping,\mem} \dtsomovesto{\transition}\tuple{\statemapping\update{\proc}{\state'},\buffermapping,\mem}$
            \end{tabular}
 \begin{tabular}{@{}c@{}}
        \\
             \;   {\sf Read from  Buffer}\\
        \\
      \end{tabular}
      \\
      \begin{tabular}{@{}c@{}}
    $  \transition=\tuple{\state,\arw(\xvar,\data,\data'),\state'}   \;\;\;\; \; \statemapping(\proc)=\state  \;\;\;\; \;\buffermapping(\proc)=\epsilon   \;\;\;\; \; \mem(\xvar)=\data$\\
              \hline
$\tuple{\statemapping,\buffermapping,\mem} \dtsomovesto{\transition}\tuple{\statemapping\update{\proc}{\state'},\buffermapping,\mem\update{\xvar}{\data'}}$
            \end{tabular}
 \begin{tabular}{@{}c@{}}
        \\
             \;\;  {\sf ARW}\\
        \\
      \end{tabular}  
            \\
      \begin{tabular}{@{}c@{}}
    $  \transition=\tuple{\state,\fenceop,\state'}   \;\;\;\; \; \statemapping(\proc)=\state  \;\;\;\; \;\buffermapping(\proc)=\epsilon  $\\
              \hline
$\tuple{\statemapping,\buffermapping,\mem} \dtsomovesto{\transition}\tuple{\statemapping\update{\proc}{\state'},\buffermapping,\mem}$
            \end{tabular}
 \begin{tabular}{@{}c@{}}
        \\
             \;\;  {\sf Fence}\\
        \\
      \end{tabular}  
 \end{tabular}
    \end{mdframed}

  \caption{The induced transition relation  $\dtsomovesto{}$ under the Dual  TSO semantics. Here, process $\proc \in \procs$ and transition $\transition \in \transitions_\proc \cup \transitions'_\proc$  where  $\transitions'_\proc:=\setcomp{\propagate_{\proc}^{\xvar} , \delete_{\proc}}{\xvar\in\vars}$.    }
  
  \label{fig:oper:sem-dtso}
  
\end{figure}
}

\subsection{Dual TSO Semantics}
\label{dual-tso}

In this section, we define  the {\it Dual TSO} semantics.   The model  has a  FIFO {\em load} buffer between the main memory  and each process. 
This load buffer is used to store  {\em potential read} operations
that will be performed by the process. We allow this buffer to  {\it lose} messages at any time by deleting   the messages at its head in non-deterministic manner. 
Each message in the load buffer of a process $\proc$ is either a pair of the form  $(\xvar,\data)$ or a triple of the form $(\xvar,\data,{own})$ where $\xvar \in \vars$ and $\data \in \valset$. A message of the form $(\xvar,\data)$ corresponds to the fact  that  $\xvar$ has had the value  $\data$ in   the shared memory. Meanwhile, a message of the form $(\xvar,\data,{own})$  corresponds to the fact that the process $\proc$ has written the   value  $\data$ to  $\xvar$.
%
%
We say that a message $(\xvar,\data,{own})$ is 
 an {\it own-message}.

A {\it write} operation $\wop(\xvar,\data)$ of the process  $\proc$  immediately  updates the shared memory  and then appends a new own-message   $(\xvar,\data,{own})$   to the tail 
of the load buffer of $\proc$.
Read {\it propagation} is then performed by non-deterministically choosing a variable (let's say $\xvar$ and its value is  $\data$ in the shared memory) and appending the new message $(\xvar,\data)$ to the tail of the load buffer of  $\proc$. This propagation operation speculates on a read operation of $\proc$ on $\xvar$ that will be performed later on.
Moreover, 
{\it delete} operation of the process $\proc$ can be performed at any time
by removing the 
(oldest) message at the head 
of the load buffer of $\proc$.
A {\it read} operation $\rop(\xvar,\data )$  of the process $\proc$  can be executed if the   message at the head of the load buffer  of $\proc$ is of the form $(\xvar,\data)$  and there is no pending  own-message    of the form $(\xvar,\data',{own})$. 
In the case that the load buffer of $\proc$ contains some own-messages 
(i.e., of the form  $(\xvar,\data',{own})$), the  read value must correspond to the value of the most recent 
such an own-message.
Implicitly, this allows to simulate the  Read-Own-Write transitions in the TSO semantics.
A {\it fence} operation means that the load  buffer of  $\proc$
must be empty before $\proc$ can continue.
Finally, an {\it atomic read-write} operation  $\arw(\xvar,\data,\data')$ means that the load  buffer of  $\proc$
must be empty
and the value of the variable $\xvar$ in the memory is $\data$ before $\proc$ can continue. 
%
%
%
%

\paragraph{\bf DTSO-configurations.}
A {\it DTSO-configuration}
$\conf$ is a triple $\tuple{\statemapping,\buffermapping,\mem}$
where:
\begin{enumerate}
\item
 $\statemapping: \procs \mapsto\states $  is the {\em global state} of $\cprogram$.
 \item
$\buffermapping:\procs\mapsto\left((\vars\times\dataset) \cup (\vars\times\dataset\times \{own\})\right)^*$ is the content of the load buffer of each process.
\item
$\mem:\vars\mapsto \dataset$ gives the value of each  shared variable. 
\end{enumerate}

The  {\it initial} DTSO-configuration ${\dinitconf}$  is defined by   $\tuple{\initstatemapping,\initbuffermapping,\initmem}$
where, for all $\proc\in\procs$ and $\xvar \in \vars$, we have that
$\initstatemapping(\proc)=\initstate_\proc$,
$\initbuffermapping(\proc)=\emptyword$ and $\initmem(\xvar)=0$.

We  use
 $\dtsoconfs$ to denote the set of all DTSO-configurations.

\paragraph*{\bf DTSO-transition Relation.}
The   {\it transition relation} $\dtsomovesto{}$ 
between DTSO-configurations is given by a set of rules, described in Figure~\ref{fig:oper:sem-dtso}.
%
This  relation is induced by members of
  $\transitions\cup\transitions''$ where
 $\transitions'':=\setcomp{\propagate_{\proc}^{\xvar} , \delete_{\proc}}{\proc\in\procs ,\; \xvar \in \vars}$.
   
We informally explain the transition relation rules.
The {\it propagate} transition  $\propagate_\proc^\xvar$ speculates on a read operation of $\proc$ over $\xvar$ that will be executed later. This is done by appending a new message  $(\xvar,\data)$ to the tail 
of the load buffer of $\proc$ where $\data$ is the current value of $\xvar$ in the shared memory. 
The {\it delete} transition $\delete_{\proc}$  removes  the (oldest) message at the head 
of the load buffer of the process $\proc$.
A {\it write} transition $\tuple{\state,\wop(\xvar,\data),\state'}\in\transitions_\proc$ updates the memory and appends a new own-message $(\xvar,\data,{own})$ to the tail of the load buffer. 
A {\it read} transition $\tuple{\state,\rop(\xvar,\data),\state'}\in\transitions_\proc$ checks first if the load buffer of $\proc$ contains an own-message of the form  $(\xvar,\data',{own)}$. In that case, the read value $\data$ should correspond to the value of the most recent 
such an own-message. If there is no such message on the variable $\xvar$  in the load buffer of $\proc$,  then the value $\data$ of $\xvar$
is fetched from the (oldest) message at the head of the load buffer of $\proc$.

We use $\conf\dtsomovesto{}\conf'$ to denote that
$\conf\dtsomovesto{\transition}\conf'$ for some
$\transition\in\transitions\cup\transitions''$.
The transition system induced by $\cprogram$ under the Dual TSO semantics is then given by
$\tsys_{\sf DTSO}=\tuple{\dtsoconfs,\{{\dinitconf}\},\transitions\cup\transitions'',\dtsomovesto{}}$.

\paragraph*{\bf The DTSO Reachability Problem.}
The  {\it DTSO reachability problem} for   $\cprogram$ under the Dual TSO semantics is defined in a similar manner to the case of the TSO semantics.  A  global state $\statemapping_{target}$ is said to be {\it reachable} in   $\tsys_{\sf DTSO}$  if and only if 
there  is a DTSO-configuration $\conf$ of the form ${\tuple{\statemapping_{target},\buffermapping,\mem}}$, with ${\buffermapping(\proc)=\epsilon  \; \; \text{for all} \; \proc \in \procs}$, such that $\conf$ is reachable in $\tsys_{\sf DTSO}$. Then, the {\it DTSO reachability problem} consists in checking whether $\statemapping_{target}$ is reachable in  $\tsys_{\sf DTSO}$. 

\begin{exa}
\label{dtso:example}

Figure~\ref{reachable:dtso-configuration}
illustrates a DTSO-configuration $\conf'$ that can be reached from the initial configuration $\dinitconf$ of the concurrent system in Figure~\ref{example:fig}.
To reach  this configuration, a propagation operation 
is performed by appending the  message $\tuple{\yvar,0}$
into the load buffer of $\proc_1$.
Then, the 
process $\proc_2$
 executes two write transitions
$\tuple{\state'_0,\wop(\yvar,1),\state'_1}$
and $\tuple{\state'_1,\wop(\xvar,1),\state'_2}$ that 
update the shared memory to $\xvar=1$ and $\yvar=1$
and add two own-messages 
to the tail of the load buffer
of $\proc_2$. 
Hence, the load buffer of $\proc_2$
contains the sequence $\tuple{\xvar,1,{own}}\app\tuple{\yvar,1,{own}}$.
Then,
the process $\proc_1$ executes
the  write transition
$\tuple{\state_0,\wop(\xvar,2),\state_1}$ which updates 
 the shared memory  and appendes
the own-message $\tuple{\xvar,2,{own}}$
to the tail of the load buffer of $\proc_1$.
After that, 
a propagation operation appending 
the  message $\tuple{\xvar,2}$
into the load buffer of 
$\proc_2$ is performed. Hence, the value of $\xvar$  (resp.  $\yvar$ ) is $2$ (resp. $1$) in the shared memory. Furthermore,  the load buffer of $\proc_1$ (resp. $\proc_2$)
contains the following sequence
$\tuple{\xvar,2,{own}}\app\tuple{\yvar,0}$ (resp, 
$\tuple{\xvar,2}\app\tuple{\xvar,1,{own}}\app\tuple{\yvar,1,{own}}$).
Now from the configuration $\conf'$ (given in Figure~\ref{reachable:dtso-configuration}),
the process
$\proc_1$ can perform
a read transition $\tuple{\state_1,\rop(\yvar,0),\state_2}$.
Since
there is no pending
own-message
of the form $\tuple{\yvar,\data,{own}}$ for some $\data\in\dataset$
in the load buffer of $\proc_1$,
$\proc_1$ reads from the message at the head of its load buffer, i.e. the message $\tuple{\yvar,0}$ (represented by the dash arrow for $\proc_1$).
Then,
$\proc_2$ performs
two delete transitions
$\delete_{\proc_2}$
to remove
two own-messages
at the head of its load buffer. Now, the process $\proc_2$ can perform  the  read transition
$\tuple{\state'_2,\rop(\xvar,2),\state'_3}$ to read from its load buffer.
Finally,
$\proc_1$ and $\proc_2$ performs a sequence of delete transitions
$\delete_{\proc_1}\app\delete_{\proc_1}\app\delete_{\proc_2}$ to empty their load buffers,
reaching to the configuration $\conf'_{target}$ given in  Figure~\ref{reachable:dtso-configuration1}.  Let $\statemapping_{target}$ be  the global state
in $\conf'_{target}$
 defined as follows:
$\statemapping_{target}(\proc_1)=\state_2$
and 
$\statemapping_{target}(\proc_2)=\state'_3$. Therefore,
we can say that
the global state $\statemapping_{target}$
in $\conf'_{target}$
is reachable in 
$\tsys_{\sf DTSO}$.
\myend

\end{exa}

\begin{figure}[t]
\centering
 \begin{tikzpicture}
[background rectangle/.style={rounded corners,line width=1pt},show background rectangle]
\node [name=dummy]{};

\node[rectangle,rounded corners,align=center,minimum  height=20mm,font=\footnotesize,name=memory,line width=0.5pt,draw] 
at (dummy) {$x=2$\\$y=1$};
\node[font=\small,name=memoryname] at ($(memory.west)+(5mm,12mm)$){memory};

\node[buffernode,font=\footnotesize,line width=0.5pt,anchor=east,name=b11] at ($(memory.west)+(-16mm,7mm)$){$\tuple{\xvar,2,{own}}$};
\node[buffernode,font=\footnotesize,line width=0.5pt,anchor=east,name=b12] at ($(b11.west)+(0mm,0mm)$){$\tuple{\yvar,0}$};

\node[font=\small,name=buffername] at ($(b11.west)+(4mm,5mm)$){load buffers};
\node[font=\small,name=headname] at ($(b11.west)+(16mm,-5mm)$){tail};
\node[font=\small,name=headname] at ($(b11.west)+(-8mm,-5mm)$){head};

\node[buffernode,font=\footnotesize,line width=0.5pt,anchor=east,name=b21] at ($(memory.west)+(-7mm,-7mm)$){$\tuple{\xvar,2}$};
\node[buffernode,font=\footnotesize,line width=0.5pt,anchor=east,name=b22] at ($(b21.west)+(0mm,0mm)$){$\tuple{\xvar,1,{own}}$};
\node[buffernode,font=\footnotesize,line width=0.5pt,anchor=east,name=b23] at ($(b22.west)+(0mm,0mm)$){$\tuple{\yvar,1,{own}}$};

\node[name=q3,anchor=east,font=\footnotesize] at ($(b12.west)+(-17.5mm,-1.6mm)$) {$\tuple{\state_1,\rop(\yvar,0),\state_2}$};

\draw[fill=black] ($(q3.west)+(1.0mm,0mm)$) -- ($(q3.west)+(-1.0mm,1mm)$) -- ($(q3.west)+(-1.0mm,-0.7mm)$) -- cycle;
\node[name=p1,anchor=west,font=\footnotesize] at ($(q3.west)+(0mm,5mm)$){$\proc_1: \statemapping(\proc_1)=\state_1$};

\node [rounded corners,fit=(p1) (q3), name=proc1,line width=0.5pt,draw] {};
\node[font=\small,name=processname] at ($(proc1.west)+(13mm,9mm)$){processes};

\node[name=s3,anchor=east,font=\footnotesize] at ($(b23.west)+(-10mm,-1.6mm)$) {$\tuple{\state'_2,\rop(\xvar,2),\state'_3}$};

\draw[fill=black] ($(s3.west)+(1.0mm,0mm)$) -- ($(s3.west)+(-1.0mm,1mm)$) -- ($(s3.west)+(-1.0mm,-0.7mm)$) -- cycle;
\node[name=p2,anchor=west,font=\footnotesize] at ($(s3.west)+(0mm,5mm)$){$\proc_2: \statemapping(\proc_2)=\state'_2$};

\node [rounded corners,fit=(p2)(s3),name=proc2,line width=0.5pt,draw] {};

\draw[transedge] (b12) -- ($(proc1.east)+(0mm,-1mm)$);
\draw[transedge] (b23) -- ($(proc2.east)+(0mm,-1mm)$);
\draw[transedge] ($(b11)+(24mm,0mm)$) -- (b11);
\draw[transedge] ($(b21)+(11.5mm,0mm)$) -- (b21);

\draw[transedge,dashed,line width=1pt,font=\footnotesize] ($(b12.west)+(0mm,3mm)$)  to [in=30,out=150]
node[above,  sloped]{$\rop(\yvar,0)$}  ($(proc1.east)+(0mm,2mm)$);

\draw[transedge,dashed,line width=1pt,font=\footnotesize] ($(b21.west)+(5mm,-3mm)$)  to [in=320,out=200]
node[above, sloped]{$\rop(\xvar,2)$}  ($(proc2.east)+(0mm,-3mm)$);

\end{tikzpicture}
\caption{A reachable DTSO-configuration  of the concurrent system in Figure~\ref{example:fig}.}
\label{reachable:dtso-configuration}
\end{figure}

 \begin{figure}[t]
\centering
 \begin{tikzpicture}
[background rectangle/.style={rounded corners,line width=1pt},show background rectangle]
\node [name=dummy]{};

\node[rectangle,rounded corners,align=center,minimum  height=20mm,font=\footnotesize,name=memory,line width=0.5pt,draw] 
at (dummy) {$x=2$\\$y=1$};
\node[font=\small,name=memoryname] at ($(memory.west)+(5mm,12mm)$){memory};


%

%
\node[name=p1,anchor=west,font=\footnotesize] at ($(q3.west)+(0mm,5mm)$){$\proc_1: \statemapping(\proc_1)=\state_2$};

\node [rounded corners,fit=(p1) (q3), name=proc1,line width=0.5pt,draw] {};
\node[font=\small,name=processname] at ($(proc1.west)+(13mm,9mm)$){processes};

%
%

\node[name=p2,anchor=west,font=\footnotesize] at ($(s3.west)+(0mm,5mm)$){$\proc_2: \statemapping(\proc_2)=\state'_3$};

\node [rounded corners,fit=(p2)(s3),name=proc2,line width=0.5pt,draw] {};

\draw[transedge] ($(b11)+(24mm,0mm)$) -- ($(proc1.east)+(0mm,-1mm)$);
\draw[transedge] ($(b21)+(11.5mm,0mm)$) -- ($(proc2.east)+(0mm,-1mm)$);

\end{tikzpicture}
\caption{A reachable ``empty buffer''   DTSO-configuration of the concurrent system in Figure~\ref{example:fig}.}
\label{reachable:dtso-configuration1}
\end{figure}

\subsection{Relation between TSO and DTSO Reachability Problems}

The following theorem  states  the equivalence of the reachability problems under the TSO and Dual TSO semantics.
\begin{thm}[TSO-DTSO reachability equivalence]
\label{DTSO:TSO:equivalence:theorem}
A global state  $\statemapping_{target}$ is reachable in $\tsys_{\sf TSO}$ iff  $\statemapping_{target}$ is reachable in $\tsys_{\sf DTSO}$.\end{thm}
\begin{proof}
The proof of this theorem can be found   in Appendix~\ref{proofs:DTSO:TSO:equivalence:theorem}.
\end{proof}

%

\begin{exa}
In the Example~\ref{tso:example} and Example~\ref{dtso:example},
we have shown that
the global state
$\statemapping_{target}$ (defined by 
$\statemapping_{target}(\proc_1)=\state_2$
and $\statemapping_{target}(\proc_2)=\state'_3$)
is both reachable in 
$\tsys_{\sf TSO}$
and 
$\tsys_{\sf DTSO}$
for the concurrent system given in Figure~\ref{example:fig}.
\myend
\end{exa}



\section{The DTSO Reachability Problem}
\label{dual-reachability:section}

In this section, we show the {\it decidability} of the DTSO reachability problem by  making use of the framework of \textit{Well-Structured Transition Systems} (\Wsts) \cite{abdulla-general-96,FinkelS01}. First, we briefly recall the framework of {\Wsts}. Then, we instantiate it   to show the decidability  of the  DTSO reachability problem.
Following Theorem~\ref{DTSO:TSO:equivalence:theorem}, we also obtain the decidability of the TSO reachability problem.

\subsection{Well-structured Transition Systems}
\label{general:well-structured-transition-systems}
 
Let $\tsys=\tstuple$ be a transition system. Let $\genordering$ be a   {\em well-quasi-ordering}   on $\allconfs$. 
 Recall that a well-quasi-ordering on $\allconfs$   is a binary relation over $\allconfs$ that is reflexive and transitive; and for every
infinite sequence $(\conf_i)_{i\ge0}$ of elements in $\allconfs$, 
there exist $i,j\in\nat$ such that $i<j$ and $\conf_i\genordering \conf_j$. 

A set ${\tt U} \subseteq \allconfs$ is called
 {\em upward closed} if for every $\conf \in {\tt U}$ and $\conf' \in \allconfs$ with $ \conf \genordering \conf'$,
 we have $\conf' \in {\tt U}$.  It is known that every upward closed set ${\tt U}$
 can be characterised by a finite {\em minor set} ${\tt M} \subseteq {\tt U}$ such that: (i) for every $\conf \in {\tt U}$, there is $\conf'\in {\tt M}$ such that $\conf' \genordering \conf$; and (ii)  if $\conf, \conf'\in {\tt M}$ and $\conf \genordering \conf'$, then $\conf = \conf'$.	
We use $\minn({\tt U})$ to denote 
for a given  upward closed set ${\tt U}$ 
its 
minor set.

Let ${\tt D} \subseteq \allconfs$. The upward closure of  ${\tt D}$ is defined
 as $\upclosure{{\tt D}} := \setcomp{\conf' \in \allconfs }{\exists \conf \in {\tt D} \text{ with } 
 \conf \genordering \conf'}$. We also define the set of {\it predecessors} of ${\tt D}$ as $\preof{\tsys}{{\tt D}}  := \setcomp{\conf}{\exists\conf_1\in {\tt D},a\in\actions, \conf\by{a} \conf_1}$.  For a finite set of configurations ${\tt M} \subseteq \allconfs$, we use $\minpreof{{\tt M}}$ to denote
$\minnof{\preof{\tsys}{ \upclosure{{\tt M}} } \cup \upclosure{{\tt M}}}$.

The transition relation  $\by{}$ is said to be \textit{ monotonic} wrt.\ the ordering $\genordering$ if, given  $\conf_1, \conf_2, \conf_3 \in \allconfs$ where $\conf_1 \by{} \conf_2$ and $\conf_1 \genordering \conf_3$, we can compute a configuration $\conf_4 \in \allconfs$ and a run $\run$ such that  $\conf_3 \by{\run} \conf_4$ and $\conf_2 \genordering \conf_4$.
The pair   $(\tsys,\genordering)$ is called a {\em monotonic transition system}  if    $\by{}$ is \textit{ monotonic} wrt.\ $\genordering$.

 Given  a {\em finite} set of configurations ${\tt M} \subseteq \allconfs$, the {\em coverability problem} of ${\tt M}$ in  
 the monotonic transition system $(\tsys,\genordering)$
asks whether the set $\upclosure{{\tt M}} $ is  {\it reachable} in $\tsys$; i.e.\ there exist two configurations $\conf_1$ and $\conf_2$ such that $\conf_1\in{\tt M}$, $\conf_1\genordering\conf_2$, and $\conf_2$ is reachable in $\tsys$.

For the {\it decidability} of this problem, the following {\bf three conditions} are sufficient:
\begin{enumerate}
\item
For every two configurations $\conf_1$ and $\conf_2$, it is decidable
whether $\conf_1 \genordering \conf_2$.
\item 
For  every $\conf\in \allconfs$, we can check whether 
$\upclosure{\set{\conf}} \cap \initconfs \neq \emptyset$.
\item
 For  every $\conf\in \allconfs$, the set \textup{$\minpreof{\set{\conf}}$} is finite and computable.
\end{enumerate}
The  solution for the coverability problem as suggested in 
\cite{abdulla-general-96,FinkelS01}
is based on a backward analysis approach. It is shown that starting 
from a finite set ${\tt M}_0 \subseteq \allconfs$, the sequence $({\tt M}_i)_{i\ge0}$
with ${\tt M}_{i+1} := \minpreof{{\tt M}_i}$, for $i\ge0$, reaches a 
fixpoint and it is computable.

\subsection{DTSO-transition System is a \bfseries{\scshape{Wsts}} }

In this section, we instantiate the framework  of {\Wsts} to show the following result:

\begin{thm}[Decidability of DTSO reachability problem]
\label{decidability-dtso}
The DTSO reachability problem is decidable.
\end{thm}
\begin{proof}
The rest of this section is devoted to the proof of the above theorem.
 Let $\cprogram=\cstuple$ be a concurrent system (as defined in Section \ref{definitions:section}). Moreover, let  $\tsys_{\sf DTSO}=\tuple{\dtsoconfs,\{{\dinitconf}\},\transitions\cup\transitions'',\dtsomovesto{}}$ be  the transition system induced by $\cprogram$ under the Dual TSO semantics (as defined in Section \ref{dual-tso}).   
 
In the following, 
we will show that the DTSO-transition system $\tsys_{\sf DTSO}$
is monotonic wrt.\ an ordering $\genordering$.
Then, we will  show the three sufficient conditions for the decidability of the coverability problem for  $(\tsys_{\sf DTSO},\genordering)$
(as stated in Section~\ref{general:well-structured-transition-systems}). 
\begin{enumerate}
\item
We   first define the ordering $\genordering$ on the set of DTSO-configurations $\dtsoconfs$ (see Section~\ref{ordering:dtso:transition-system}).
\item
Then, we show that the transition system $\tsys_{\sf DTSO}$ induced under the Dual TSO semantics is  monotonic wrt.\ the ordering $\genordering$ (see Lemma \ref{monotonicity-dtso}). 
\item
For the first sufficient condition, we show that  $\genordering$
is a well-quasi-ordering; and that
for every two configurations $\conf_1$ and $\conf_2$, it is decidable whether $\conf_1 \genordering \conf_2$ (see Lemma \ref{well-quasi-order}). 
\item
The second  sufficient condition  (i.e., checking whether the upward closed set $\upclosure{\set{\conf}}$, with $\conf$ is a DTSO-configuration,  contains the initial configuration $\dinitconf$)  is trivial. This check  boils down to  verifying  whether $\conf$ is the  initial configuration $\dinitconf$.  
\item
For the third sufficient condition, we show that we can calculate the set of minimal DTSO-configurations for the set of predecessors of any upward closed set  (see Lemma~\ref{compute-pre-DTSO}).
\item
Finally, we will  also show that  the DTSO reachability problem for $\cprogram$ can be reduced to the coverability problem in the monotonic transition system $(\tsys_{\sf DTSO},\genordering)$ (see Lemma \ref{eq-reach-upw}). Observe that this reduction is needed since we are requiring that the load buffers are empty when defining the DTSO reachability
problem.
\end{enumerate}

This concludes the proof of Theorem~\ref{decidability-dtso}.
\end{proof}

\subsubsection{\bf Ordering $\bm{\genordering}$}
\label{ordering:dtso:transition-system}
In the following, we  define an ordering $\genordering$ on  the set of DTSO-configurations $\dtsoconfs$. Let us first introduce some notations and definitions.

Consider a word
$\word\in\left((\vars\times\dataset) 
\cup (\vars\times\dataset\times \{ own\})\right)^*$
representing the content of a  load buffer. 
 We  define an operation that divides
$\word$ into a number of fragments
according to the most-recent own-message concerning  each variable.  
 We define 
$$\ownof\word:=
\left( \word_1, (\xvar_1,\data_1,{ own}), \word_2,\ldots,\word_{m}, (\xvar_m,\data_m,{ own}),\word_{m+1}\right)$$
where  the following conditions are satisfied: 
\begin{enumerate}
\item
$\xvar_i\neq\xvar_j$ for all $i, j: i \neq j$ and $1\leq i, j \leq m$. 
\item
If 
$(\xvar,\data,{ own})\in\word_{i}$ for some $i: 1 < i \leq m+1$, then
$\xvar=\xvar_j$ for some $j: 1\leq j < i $, i.e., the most recent own-message on variable $\xvar_j$ occurs
at the $(2j)^{\it th}$ fragment of $\ownof\word$.
\item
$\word=
\word_1\app (\xvar_1,\data_1,{ own})\app
\word_2 \cdots
\word_{m}\app (\xvar_m,\data_m,{ own})\app \word_{m+1}$, i.e., 
the  divided fragments correspond to the given word $\word$.
\end{enumerate}

Let $\word,\word'\in\left((\vars\times\dataset) 
\cup (\vars\times\dataset\times \{ own\})\right)^*$ be two words. Let us assume that: 
$$\ownof\word=
( \word_1, (\xvar_1,\data_1,{ own}),\word_2, \ldots, 
\word_{r}, (\xvar_r,\data_r,{ own}),\word_{r+1})$$
$$
\ownof{\word'}=
( \word'_1, (\xvar'_1,\data'_1,{ own}),
\word'_2, \ldots,
\word'_{m}, (\xvar'_m,\data'_m,{ own}),\word'_{m+1}).$$
We write $\word\genordering\word'$ to denote that
the following conditions are satisfied:
\begin{enumerate}
\item
$r=m$,
\item
$x'_i=x_i$ and $\data'_i=\data_i$ for all $i:1\leq i\leq m$, and
\item
$\word_i\wordering\word'_i$ for 
all $i:1\leq i\leq m+1$.
\end{enumerate}

Consider two DTSO-configurations  $\conf=\tuple{\statemapping,\buffermapping,\mem}$ and $\conf'=\tuple{\statemapping',\buffermapping',\mem'}$, we extend the ordering  $\genordering$ to configurations as follows: 
We write
$\conf \genordering \conf'$ if and only if the following conditions are satisfied: 
\begin{enumerate}
\item
$\statemapping=\statemapping'$, 
\item
$\buffermapping(\proc)\genordering\buffermapping'(\proc)$ for all process $\proc \in \procs$, and 
\item
$\mem'=\mem$. 
\end{enumerate}

\subsubsection{\bf Monotonicity.}
Let $\conf_1=\tuple{\statemapping_1,\buffermapping_1,\mem_1}, \conf_2=\tuple{\statemapping_2,\buffermapping_2,\mem_2}, \conf_3=\tuple{\statemapping_3,\buffermapping_3,\mem_3} \in \dtsoconfs$ such that $\conf_1 \by{\transition}_{\sf DTSO} \conf_2$ for some  $\transition  \in\transitions_\proc \cup \setcomp{\propagate_{\proc}^{\xvar} , \delete_{\proc}}{\xvar\in\vars}$ with   $\proc \in \procs$, and $\conf_1 \genordering \conf_3$. We will show that it is possible to {compute} a configuration $\conf_4 \in \dtsoconfs$ and a run $\run$ such that $\conf_3 \by{\run}_{\sf DTSO} \conf_4$ and $\conf_2 \genordering \conf_4$. 

To that aim, we  first show that it is possible from $\conf_3$ to reach  a configuration $\conf'_3$, by performing a  certain number of  $\delete_\proc$ transitions,  such that the process $\proc$ will have  the same last message in its load buffer  in the configurations $\conf_1$ and $\conf'_3$ while  $\conf_1 \genordering \conf'_3$. Then, from the configuration $\conf'_3$, the process $\proc$ can   perform  the same transition $\transition$ as  $\conf_1$ did (to reach $\conf_2$) in order to reach  the configuration $\conf_4$ such that $\conf_2 \genordering \conf_4$.   Let us assume that $\ownof{\buffermapping_1(\proc)}$ is of the form 
$$\left( \word_1,  (\xvar_1,\data_1,{ own}),
\word_2,\ldots,\word_{m}, (\xvar_m,\data_m,{ own}),\word_{m+1}\right)$$
and
$\ownof{\buffermapping_3(\proc)}$ is of the form $$
\left( \word'_1, (\xvar'_1,\data'_1,{ own}),
\word'_2,\ldots,
\word'_{m}, (\xvar'_m,\data'_m,{ own}),\word'_{m+1}\right).$$ We define the word $\word \in\left((\vars\times\dataset) 
\cup (\vars\times\dataset\times \{ own\})\right)^*$ to be the longest word such that $\word'_{m+1}= \word' \cdot  \word$ with $\word_{m+1} \preceq \word'$. Observe that in this case we have either $\word_{m+1}=\word'=\epsilon$ or $\word'(|\word'|)=\word_{m+1}(|\word_{m+1}|)$. Then, after executing a certain number $|w|$ of  $\delete_\proc$ transitions from the configuration $\conf_3$, one can obtain a configuration $\conf'_3=\tuple{\statemapping_3,\buffermapping'_3,\mem_3} $ such that $\buffermapping_3= \buffermapping'_3 \update{\proc}{\buffermapping'_3(\proc) \cdot \word}$. As a consequence, we have $\conf_1 \genordering \conf'_3$. Furthermore, since $\conf_1$ and $\conf'_3$ have the same global state, the same memory valuation,  the same sequence of most-recent own-messages concerning each variable, and the same last message in the load buffers of $\proc$,  $\conf'_3$ can perform  the  transition $\transition$  and  reaches  a configuration $\conf_4$ such that $\conf_2 \genordering \conf_4$.  

The following lemma shows  that   $(\tsys_{\sf DTSO},\genordering)$ is a monotonic transition system.

\begin{lem}[DTSO monotonic transition system]
\label{monotonicity-dtso}
The transition relation $\dtsomovesto{}$ is  monotonic wrt. the ordering $\genordering$.
\end{lem}
\begin{proof}
The proof of the lemma is given in Appendix~\ref{proofs:monotonicity-dtso}.
\end{proof}


\subsubsection{\bf Conditions of Decidability.}
We show the first and the third conditions of the three conditions for the decidability of the coverability problem for  $(\tsys_{\sf DTSO},\genordering)$ (as stated in Section~\ref{general:well-structured-transition-systems}). The second condition has been shown to be trivial in the main proof of Theorem~\ref{decidability-dtso}.

The following lemma shows that the ordering $\genordering$ is indeed a well-quasi-ordering.

\begin{lem}[Well-quasi-ordering $\genordering$]
\label{well-quasi-order}
The ordering $\genordering$ is  a well-quasi-ordering over  $\dtsoconfs$. Furthermore, for every two DTSO-configurations $\conf_1$ and $\conf_2$, it is decidable
whether $\conf_1 \genordering \conf_2$.
\end{lem}
\begin{proof}
The proof of the lemma is given in Appendix~\ref{proofs:well-quasi-order}.
\end{proof}


The following lemma shows that we can calculate the set of minimal
DTSO-configurations for the set of predecessors of any upward closed set.
\begin{lem}[Computable  minimal predecessor set]
\label{compute-pre-DTSO}
For any  DTSO-configuration $\conf$,
we can compute $\minpre(\set{\conf})$.
\end{lem}
\begin{proof}
The proof of lemma is given in Appendix~\ref{proofs:compute-pre-DTSO}.
\end{proof}

\subsubsection{\bf From DTSO Reachability to Coverability.} 
Let  $\statemapping_{ target}$ be a   global state of a concurrent program $\cprogram$ and let ${\tt M}_{ target}$ be the set of DTSO-configurations of the form ${\tuple{\statemapping_{ target},\buffermapping,\mem}}$ with ${\buffermapping(\proc)=\epsilon  \; \; \text{for all} \; \proc \in \procs}$ where 
$ \procs$ be the set of process IDs in $\cprogram$.
  We recall that $\statemapping_{ target}$ in $\tsys_{\sf DTSO}$  if and only if ${\tt M}_{ target}$ is reachable in $\tsys_{\sf DTSO}$ (see Section~\ref{sec:prels} for the definition of a reachable set of configurations).
Then by Lemma~\ref{eq-reach-upw}, 
we have that the reachability problem of $\statemapping_{ target}$ in $\tsys_{\sf DTSO}$
 can be reduced to  the coverability problem of ${\tt M}_{ target}$ in   
  $(\tsys_{\sf DTSO},\genordering)$. 
  
\begin{lem}[DTSO reachability to coverability]
\label{eq-reach-upw}
${\tt M}_{ target} \!\!\uparrow$ is reachable in $\tsys_{\sf DTSO}$ iff ${\tt M}_{ target}$ is reachable in $\tsys_{\sf DTSO}$.
\end{lem}
\begin{proof}
 Let us assume  that ${\tt M}_{ target}\!\! \uparrow$ is reachable in $\tsys_{\sf DTSO}$. This means that there is a configuration $\conf \in {\tt M}_{ target}\!\! \uparrow$ which is reachable in $\tsys_{\sf DTSO}$. Let us assume that $\conf$ is of the form  $\tuple{\statemapping_{ target},\buffermapping,\mem}$. Then,    from  the configuration $\conf$, it is possible to reach   the configuration $\conf'=\tuple{\statemapping_{ target},\buffermapping',\mem}$, with  ${\buffermapping'(\proc)=\epsilon  \; \; \text{for all} \; \proc \in \procs}$, by performing a sequence of  delete transitions to empty the load buffer of each process. It is then easy to see that $\conf' \in {\tt M}_{ target}$ and so  ${\tt M}_{ target}$ is reachable in $\tsys_{\sf DTSO}$. The other direction of the  lemma is trivial since  ${\tt M}_{ target} \subseteq {{\tt M}_{ target} \!\!\uparrow}$.
\end{proof}



\section{Parameterized Concurrent Systems}
\label{parameterized:section}
In this section, we give the definitions for
{\it parameterized concurrent systems}, a model for representing unbounded number of communicating concurrent processes under the Dual TSO semantics, and its induced transition system.
Then, we define the DTSO reachability problem for the case of parameterized concurrent systems.

\subsection{Definitions for  Parameterized Concurrent Systems}
\label{parameterized:concurrent:system:definitions}
 Let $\valset$ be a finite data domain and   $\vars$ be a finite set of variables ranging over   $\valset$.  A {\it parameterized concurrent system} (or simply 
a {\it parameterized system}) consists  of an unbounded number of {\it identical}   processes running under the Dual TSO semantics.
Communication between processes is performed 
through a shared memory that consists of a finite number of the shared variables $\vars$ over the finite domain $\valset$.
Formally, a parameterized system $\parsys$ is  defined by  an extended finite-state automaton $\automaton=\tuple{\states,\initstate,\transitions}$  uniformly describing the behavior of each  process. 
 
An {\it instance} of $\parsys$ is a concurrent system
$\cprogram=\cstuple$, for some $n \in \mathbb{N}$, where for each $\proc: 1 \leq \proc \leq n$,  we have $\automaton_\proc=\automaton$. In other words, it consists of a finite set
of processes each running the same code defined
by $\automaton$. We use $\instance(\parsys)$ to denote all possible instances of $\parsys$. We use   
$\tsys_\cprogram:=\tuple{\allconfs_\cprogram,\initconfs_\cprogram,\actions_\cprogram,\by{}_\cprogram}$  to denote  the transition system induced by an instance  $\cprogram$ of $\parsys$ under the Dual TSO semantics.

A {\it parameterized configuration} $\alpha$  is a pair $\tuple{\procs,\conf}$ where $\procs=\{1,\ldots,n\}$, with $n \in \mathbb{N}$, is the set of process IDs and  $\conf$ is a DTSO-configuration of an  instance $\cprogram=\cstuple$ of $\parsys$.  
The parameterized configuration $\alpha=\tuple{\procs,\conf}$ is said to be {\it initial} if $\conf$ is an initial configuration of  $\cprogram$ (i.e., $\conf \in \initconfs_\cprogram$). We use     $\allconfs$ (resp. $\initconfs$) to denote the set of all  the  parameterized  configurations  (resp. all the initial configurations) of $\parsys$. 

Let $\actions$ denote   the set of  {\it actions} of all possible instances of $\parsys$ (i.e., $\actions=\cup_{\cprogram \in \instance(\parsys)} \; \actions_\cprogram$).
We define a {\it transition relation} $\by{}$ on the set $\allconfs$ of all parameterized configurations such that  
given two  configurations 
$\tuple{\procs,\conf}$ and $\tuple{\procs',\conf'}$,
we have 
$\tuple{\procs,\conf} \by{\transition} \tuple{\procs',\conf'}$ for some action $\transition \in \actions$ iff 
 $\procs'=\procs$ and there is an instance $\cprogram$ of $\parsys$ such that $\transition \in \actions_{\cprogram}$ and 
$\conf\by{\transition}_\cprogram \conf'$.
The {\it transition system}  induced by  $\parsys$ is given by $\tsys:=\tuple{\allconfs,\initconfs,\actions,\by{}}$.

\paragraph{\bf The Parameterized DTSO Reachability Problem}
A global state  $\statemapping_{ target}: \procs' \mapsto \states$ is said to be {\it reachable} in $\tsys$ if and only if there exists  a parameterized configuration $\alpha=(\procs,\tuple{\statemapping,\buffermapping,\mem})$, with ${\buffermapping(\proc)=\epsilon  \; \; \text{for all} \; \proc \in \procs}$, such that $\alpha$ is reachable in $\tsys$ and $\statemapping_{ target}(1) \cdots \statemapping_{ target}(|\procs'|) \preceq \statemapping(1) \cdots \statemapping(|\procs|)$. 

The parameterized DTSO reachability problem consists in checking whether $\statemapping_{ target}$ is reachable in  $\tsys$.  In other words, the DTSO reachability  problem for parameterized systems asks whether there is an instance of the parameterized system  that reaches   a configuration with a number of  processes in  certain given local states.

\subsection{Decidability of the Parameterized Reachability Problem}

We prove hereafter the following theorem:


\begin{thm}[Decidability of parameterized DTSO   reachability  problem]
\label{cover-decidability}
The parameterized DTSO  reachability  problem   is decidable.
\end{thm}

\begin{proof}
 Let $\parsys=\tuple{\states,\initstate,\transitions}$ be a parameterized system and $\tsys=\tuple{\allconfs,\initconfs,\actions,\by{}}$ be its induced transition system.
  The proof of the theorem 
  is done by
instantiating the framework of {\Wsts}. 
%
In more detail,
we will show that 
the parameterized transition system $\tsys$
is monotonic wrt.\ an ordering $\trianglelefteq$.
Then, we will  show the three sufficient conditions for the decidability of the coverability problem for  $(\tsys,\trianglelefteq)$
(as stated in Section~\ref{general:well-structured-transition-systems}). 
\begin{enumerate}
\item
We   first define the ordering $\trianglelefteq$ on the set $\allconfs$ of all parameterized  configurations (see Section~\ref{ordering:parameterized:case}).
\item
Then, we show that the transition system $\tsys$ is  monotonic wrt.\ the ordering $\trianglelefteq$ (see Lemma \ref{monotonicity-para}). 
\item
For the first sufficient condition, we show that  the ordering $\trianglelefteq$
is a well-quasi-ordering; and that
for every two parameterized configurations $\alpha$ and $\alpha'$, it is decidable whether $\alpha \genordering \alpha'$ (see Lemma \ref{par-wqo}). 
\item
The second  sufficient condition (i.e., checking  whether the upward closed set $\upclosure{\set{\alpha}}$, with $\alpha$ is a parameterized  configuration,  contains an initial configuration) for the decidability of the coverability problem is trivial. This check boils down to  verifying whether the configuration $\alpha$ is initial.   
\item
For the third sufficient condition,
we show that we can calculate the set of minimal parameterized configurations for the set of predecessors of any upward closed set (see Lemma~\ref{compute-pre-parsys}).
\item
Finally, we will show that the parameterized DTSO  reachability  problem for the parameterized system $\parsys$  can be reduced to  the  coverability   problem    in  
 the monotonic transition system $(\tsys,\trianglelefteq)$ (see Lemma \ref{eq-cover-upw}).
\end{enumerate}

This concludes the proof of Theorem~\ref{cover-decidability}.
\end{proof}


\subsubsection{\bf Ordering $\bm{\trianglelefteq}$.}  
\label{ordering:parameterized:case}
Let   $\alpha=\tuple{\procs,\tuple{\statemapping,\buffermapping,\mem}}$ and  $\alpha'=\tuple{\procs',\tuple{\statemapping',\buffermapping',\mem'}}$ be two parameterized configurations.  We define  the  ordering $\trianglelefteq$ 
on the set $\allconfs$ of parameterized  configurations  as follows:  We write $\alpha\trianglelefteq \alpha'$ if and only if the following conditions are satisfied: 
\begin{enumerate}
\item
$\mem=\mem'$.
\item
 There is an injection $h : \{1,\ldots,|\procs|\} \mapsto \{1,\ldots, |\procs'|\}$ such that 
 \begin{enumerate}[label=(\roman*)]
 \item
 for all $\proc,\proc'\in\procs$,
 $\proc<\proc'$ implies $h(\proc) < h(\proc')$; and
 \item
 for every $\proc \in \{1,\ldots, |\procs|\}$, $\statemapping(\proc)=\statemapping'(h(\proc))$ and $\buffermapping(\proc)\genordering \buffermapping'(h(\proc))$. 
 \end{enumerate}
  \end{enumerate}

 \subsubsection{\bf Monotonicity.}
We assume that three  configurations
 $\alpha_1=\tuple{\procs,\tuple{\statemapping_1,\buffermapping_1,\mem_1}}$, $\alpha_2=\tuple{\procs,\tuple{\statemapping_2,\buffermapping_2,\mem_2}}$ and $\alpha_3=\tuple{\procs',\tuple{\statemapping_3,\buffermapping_3,\mem_3}}$ are given.   Furthermore, we  assume that $\alpha_1 \trianglelefteq \alpha_3$ and $\alpha_1 \by{\transition} \alpha_2$ for some transition $\transition$. 
 We will show that it is possible to compute a parameterized configuration $\alpha_4$ and a run $\run$ such that $\alpha_3\by{\run}\alpha_4$ and $\alpha_2\trianglelefteq\alpha_4$.

 Since $\alpha_1 \trianglelefteq \alpha_3$, there is an injection function $h : \{1,\ldots,|\procs|\} \mapsto \{1,\ldots, |\procs'|\}$ such that: 
   \begin{enumerate}
   \item 
  For all $\proc,\proc'\in\procs$,
 $\proc<\proc'$ implies $h(\proc) < h(\proc')$. 
 \item
 For every $\proc \in \{1,\ldots, |\procs|\}$, $\statemapping_1(\proc)=\statemapping_3(h(\proc))$ and $\buffermapping_1(\proc)\genordering \buffermapping_3(h(\proc))$. 
 \end{enumerate}
 
 We define the parameterized configuration $\alpha'$  from $\alpha_3$ by only keeping  the local states and load buffers of processes in $h(\procs)$. 
 Formally,   $\alpha'=(\procs, \tuple{\statemapping',\buffermapping',\mem'})$  is defined  as follows: 
  \begin{enumerate}
 \item
 $\mem'=\mem_3$. 
 \item 
 For every $\proc \in \{1,\ldots, |\procs|\}$, $\statemapping'(\proc)=\statemapping_3(h(\proc))$ and $\buffermapping'(\proc)= \buffermapping_3(h(\proc))$. 
 \end{enumerate}
  
  We observe  that $\tuple{\statemapping_1,\buffermapping_1,\mem_1} \genordering \tuple{\statemapping',\buffermapping',\mem'}$. 
 Since   the transition relation $\by{}_{\sf DTSO}$ is 
monotonic wrt.\ the ordering $\genordering$ (see Lemma \ref{monotonicity-dtso}),  there is a DTSO-configuration $\tuple{\statemapping'',\buffermapping'',\mem''}$ such that $\tuple{\statemapping',\buffermapping',\mem'} \by{}_{\sf DTSO}^* \tuple{\statemapping'',\buffermapping'',\mem''}$ and $\tuple{\statemapping_2,\buffermapping_2,\mem_2} \genordering \tuple{\statemapping'',\buffermapping'',\mem''}$. 

Consider now the parameterized configuration $\alpha_4=\tuple{\procs',\tuple{\statemapping_4,\buffermapping_4,\mem_4}}$ such that: 
  \begin{enumerate}
  \item
$\mem''=\mem_4$. 
\item
For every $\proc \in \{1,\ldots, |\procs|\}$, $\statemapping''(\proc)=\statemapping_4(h(\proc))$ and $\buffermapping''(\proc)= \buffermapping_4(h(\proc))$.
\item 
For every $\proc \in (\{1,\ldots,|\procs'|\} \setminus \{h(1),\ldots,h(|\procs|)\})$, we have $\statemapping_4(\proc)=\statemapping_3(\proc)$ and $\buffermapping_4(\proc)= \buffermapping_3(\proc)$. 
\end{enumerate}
It is easy then to see that $\alpha_2\trianglelefteq\alpha_4$ and $\alpha_3 \by{}^* \alpha_4$. 

The following lemma shows that $\tuple{\tsys,\trianglelefteq}$ is a monotonic transition system.

\begin{lem}[Parameterized monotonic   transition system]
\label{monotonicity-para}
The transition relation $\by{}$ is  monotonic wrt.\  the ordering $\trianglelefteq$.
\end{lem}
\begin{proof}
The proof of the lemma is given in Appendix~\ref{proofs:monotonicity-para}.
\end{proof}

\subsubsection{\bf Conditions for Decidability} 
We show the first and the third conditions of the three conditions for the decidability of the coverability problem for $\tuple{\tsys, \trianglelefteq}$
 (as stated in Section~\ref{general:well-structured-transition-systems}). The second condition has been shown to be trivial in the main proof of Theorem~\ref{cover-decidability}.
  
The following lemma states that the ordering $\trianglelefteq$ is indeed a well-quasi-ordering: 
\begin{lem}[Parameterized well-quasi-ordering $\trianglelefteq$]
\label{par-wqo}
The ordering $\trianglelefteq$ is  a well-quasi-ordering over  $\allconfs$. Furthermore, for every two parameterized configurations $\alpha$ and $\alpha'$, it is decidable
whether $\alpha \trianglelefteq \alpha'$.
\end{lem}
\begin{proof}
The lemma follows a similar argument as in the proof of Lemma~\ref{well-quasi-order}.
\end{proof}

The following lemma shows that we can calculate the set of minimal parameterized configurations for the set of predecessors of any upward closed set.

\begin{lem}[Computable  minimal parameterized predecessor set]
\label{compute-pre-parsys}
For any  parameterized configuration $\alpha$,
we can compute $\minpre(\{\alpha\})$.
\end{lem}
\begin{proof}
The proof of the lemma is given in Appendix~\ref{proofs:compute-pre-parsys}.
\end{proof}

\subsubsection{\bf From Parameterized DTSO Reachability  to Coverability.} Let $\statemapping_{ target}: \procs' \mapsto \states$ be a global state.  Let  ${\tt M}_{ target}$ be the set of parameterized configurations of the form $\alpha=(\procs',{\tuple{\statemapping_{ target},\buffermapping,\mem}})$ with ${\buffermapping(\proc)=\epsilon  \; \; \text{for all} \; \proc \in \procs'}$. 
Lemma~\ref{eq-cover-upw} shows that
 the  parameterized  reachability problem of $\statemapping_{ target}$ in the transition system $\tsys$
 can be reduced to  the coverability problem of ${\tt M}_{ target}$ in   
  $(\tsys,\trianglelefteq)$. 
 
\begin{lem}[Parameterized DTSO reachability  to coverability]
\label{eq-cover-upw}
$\statemapping_{ target}$ is reachable in $\tsys$ iff 
${\tt M}_{ target} \!\!\uparrow$ is reachable in $\tsys$. 
\end{lem}
\begin{proof}
To prove the lemma, we first show  that 
${\tt M}_{ target} \!\uparrow$ is reachable in $\tsys$ if and only if there is a parameterized configuration $\alpha=(\procs,\tuple{\statemapping,\buffermapping,\mem})$, with ${\buffermapping(\proc)=\epsilon  \; \; \text{for all} \; \proc \in \procs}$, such that $\alpha$ is reachable in $\tsys$ and $\statemapping_{ target}(1) \cdots \statemapping_{ target}(|\procs'|) \preceq \statemapping(1) \cdots \statemapping(|\procs|)$.
Then  as a consequence,
 the lemma holds. 

 Let us assume that there is a parameterized configuration $\alpha=(\procs,\tuple{\statemapping,\buffermapping,\mem})$, with ${\buffermapping(\proc)=\epsilon  \; \; \text{for all} \; \proc \in \procs}$, such that $\alpha$ is reachable in $\tsys$ and $\statemapping_{ target}(1) \cdots \statemapping_{ target}(|\procs'|) \preceq \statemapping(1) \cdots \statemapping(|\procs|)$. It is then easy to show that $\alpha \in {\tt M}_{ target}\!\uparrow$.

 Now let us assume that there is  a parameterized configuration $\alpha'=(\procs'', \tuple{\statemapping',\buffermapping',\mem'}) \in  {\tt M}_{ target}\!\!\uparrow$ which  is reachable  in $\tsys$. From the configuration $\alpha'$, it is possible to reach  the  configuration $\alpha''=(\procs'',\tuple{\statemapping',\buffermapping'',\mem'})$, with  ${\buffermapping''(\proc)=\epsilon  \; \; \text{for all} \; \proc \in \procs''}$, by performing a sequence of  $\delete_\proc$ transitions to empty the load buffer of each process. Since $\alpha' \in  {\tt M}_{ target}\!\uparrow$, we have   $\statemapping_{ target}(1) \cdots \statemapping_{ target}(|\procs'|) \preceq \statemapping'(1) \cdots \statemapping'(|\procs''|)$. Hence,  $\alpha''$ is a witness of the  parameterized  reachability problem of $\statemapping_{ target}$ in the transition system $\tsys$. 
\end{proof}



\section{Experimental Results}
\label{experiments:section}

We have  implemented our techniques described in  Section~\ref{dual-reachability:section}  and Section \ref{parameterized:section} 
in an open-source tool called {\sf Dual-TSO}\footnote{\href{https://www.it.uu.se/katalog/tuang296/dual-tso}{Tool webpage: https://www.it.uu.se/katalog/tuang296/dual-tso}}.
The tool checks the state  reachability problems 
(c.f.\ Section~\ref{dual-tso} and Section~\ref{parameterized:concurrent:system:definitions})
for (parameterized) concurrent systems under the Dual TSO semantics. 
We emphasize
that 
besides checking the reachability for
a global state,
{\sf Dual-TSO}
can check  the reachability
for a set of global states.
Moreover,
{\sf Dual-TSO} accepts a more general 
input class of parameterized concurrent systems.
Instead of requiring that
the behavior of each process is described by a {\it unique} extended finite-state automaton as defined in Section~\ref{parameterized:section},
{\sf Dual-TSO} allows
that the  behavior of a process 
can be presented by
an extended finite-state automaton from a {\it fixed} set of predefined automata.
If the tool finds a witness for a given reachability problem,
we say that 
the  concurrent system is unsafe (wrt.\ the reachability problem). After finding the first witness for a given reachability problem, the tool terminates its execution. In the case that  no witness is encountered,  {\sf Dual-TSO} declares that the given concurrent program is safe (wrt.\ the reachability problem) after it reaches a fixpoint in calculation. 
 {\sf Dual-TSO}  always ends its execution   by reporting the running time (in seconds) and the total number of generated  configurations. Observe that the number of generated  configurations gives a rough estimation of the memory consumption of our tool.
%


We compare our tool with {\sf Memorax}~\cite{DBLP:conf/tacas/AbdullaACLR12,DBLP:conf/tacas/AbdullaACLR13} which  is the \emph{only precise and sound tool} for  deciding the state reachability problem of concurrent systems running under TSO. Observe  that {\sf Memorax} cannot handle the class of parameterized concurrent systems.
%
We use 
{\sf Dual-TSO}$(\genordering)$
and {\sf Dual-TSO}$(\trianglelefteq)$
to denote {\sf Dual-TSO} when applied to concurrent systems and parameterized concurrent systems, respectively.

In the following, we present two sets of results. The first set concerns the comparison of   {\sf Dual-TSO}$(\genordering)$ with {\sf Memorax} (see Table~\ref{dual-sb-comparison}). The second set shows the benefit of the parameterized verification compared to the  use of the  state reachability  when increasing the number of processes (see Table~\ref{pdual-table} and Figure~\ref{increase-processes-figs}). Our example programs are  from ~\cite{DBLP:conf/tacas/AbdullaACLR12,DBLP:journals/toplas/AlglaveMT14,DBLP:conf/esop/BouajjaniDM13,DBLP:conf/esop/AbdullaAP15,DBLP:conf/pldi/LiuNPVY12}.
%
In all experiments, we set up the time out to 600 seconds (10 minutes).
We perform all experiments  on an Intel x86-32 Core2 2.4 Ghz machine  and 4GB of RAM. 

\begin{table}[tb]
\centering
\small
\begin{tabular}{| l  | c | r r | r r | r r | }
\hline
\multirow{2}{*}{{\bf Program}} & \multirow{2}{*}{$\mathbf{\#P}$} &\multicolumn{2}{c|}{\bf Safe under} 	& \multicolumn{2}{c|}{{\sf Dual-TSO}($\genordering$)}	&\multicolumn{2}{c|}{\sf Memorax} 	 	\\ 
	& & {\bf SC} & {\bf TSO} & $\mathbf{\#T(s)}$ & $\mathbf{\#C}$ 	&  $\mathbf{\#T(s)}$ & $\mathbf{\#C}$  	 \\  \hline
 SB				& 	5	& yes & no & 	0.3 & 10 641			& 	559.7 &\;\; 10 515 914 								\\ 
 LB				& 	3	& yes &yes  &	0.0 & 2 048			&  	71.4 & 1 499 475								\\
 WRC			&	4 	& yes & yes &	0.0 & 1 507			& 	63.3 & 1 398 393								\\ 
 ISA2			& 	3	& yes &  yes &	0.0 & 509				&  	21.1 & 226 519									\\ 
RWC			& 	5	& yes & no &	0.1 &  4 277			& 	61.5 & 1 196 988									 \\ 
 W+RWC			& 	4	& yes &  no & 	0.0 & 1 713			&	 83.6 & 1 389 009								\\ 
 IRIW			& 	4	& yes & yes &	0.0 & 520				& 	34.4 & 358 057									\\ 
  MP				&	4	& yes & yes &	0.0 & 883				& 	$t/o$	& $\bullet$									\\ 
  Simple Dekker	&	2	& yes & no&	0.0 & 98				& 	0.0 & 595 					 					\\
 Dekker			&	2	& yes & no &	0.1 & 5 053			& 	1.1 & 19 788									\\
 Peterson			&	2	& yes &  no&	0.1 & 5 442			& 	5.2 & 90 301									\\ 
 Repeated Peterson	&	2 	& yes &no	&	0.2 & 7 632			& 	5.6 & 100 082									\\
  Bakery			&	2	& yes &  no& 	2.6 & 82 050			& 	$t/o$	& $\bullet$								\\
 Dijkstra			& 	2	& yes &no  & 	0.2 & 8 324			& 	$t/o$	& $\bullet$								\\
 
 Szymanski		&	2	& yes & no &	0.6 & 29 018			& 	1.0 & 26 003									\\
 Ticket Spin Lock	&	3	& yes & yes &	0.9 & 18 963			& 	$t/o$	& $\bullet$								\\
Lamport's Fast Mutex&	3	& yes &  no&	17.7 & 292 543			& 	$t/o$	& $\bullet$								\\
 Burns			&	4	& yes &no  &	124.3 &\;\; 2 762 578		& 	$t/o$	& $\bullet$								\\
 NBW-W-WR		& 	2	& yes & yes &	0.0 & 222				& 	10.7 & 200 844									\\ 
Sense Reversing Barrier&	2	& yes & yes &	0.1 & 1 704			& 	0.8 & 20 577									\\
 \hline
\end{tabular}
\caption{Comparison  between  {\sf Dual-TSO}$(\genordering)$ and {\sf Memorax}:
The columns {\it Safe under SC} and {\it Safe under TSO} indicate that whether the benchmark is safe under SC and TSO wrt.\ its reachability problem respectively.
The columns $\#P$,  $\#T$ and $\#C$ give the number of processes, the  running time in seconds and the number of generated configurations, respectively. If a tool runs out of time, we put {$t/o$} in the $\#T$ column and $\bullet$ in the $\#C$ column. }
\label{dual-sb-comparison}
\end{table}

\paragraph{\bf Verification of Concurrent Systems.}
Table~\ref{dual-sb-comparison} presents  a comparison   between
 {\sf Dual-TSO}$(\genordering)$ and {\sf Memorax} on 20 benchmarks. In all these benchmarks, {\sf Dual-TSO}$(\genordering)$ and {\sf Memorax} return  the same results for the state  reachability problems (except 6 examples where {\sf Memorax}  runs out of time).  In  the benchmarks where the two tools return, {\sf Dual-TSO}$(\genordering)$  out-performs {\sf Memorax}  and generates fewer  configurations (and so uses less memory). Indeed,  {\sf Dual-TSO}$(\genordering)$ is 600 times faster than {\sf Memorax}  and generates 277 times fewer minimal configurations  on average. 
 The experimental results  confirm  the correlation between the running time and the memory consumption (i.e.,  the tool who generates less configurations is  often the fastest).

\begin{figure}[tb]
\begin{subfigure}[b]{0.42\linewidth}
    \centering
    \includegraphics[width=0.99\textwidth]{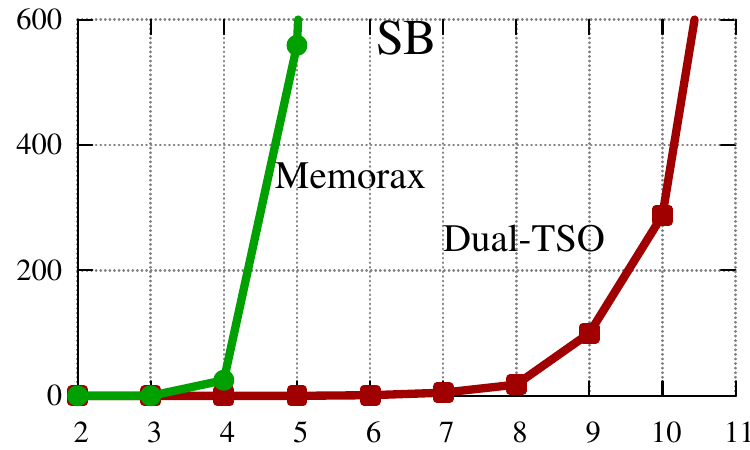}
  \end{subfigure}
  \begin{subfigure}[b]{0.42\linewidth}
    \centering
    \includegraphics[width=0.99\textwidth]{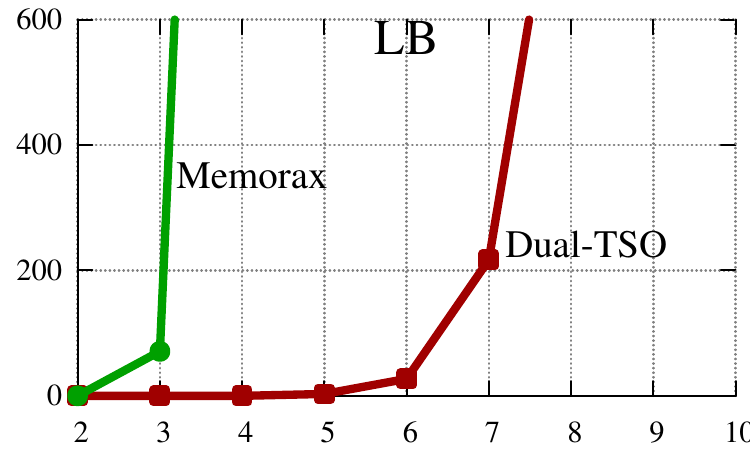}
  \end{subfigure}
  \\
  \begin{subfigure}[b]{0.42\linewidth}
    \centering
    \includegraphics[width=0.99\textwidth]{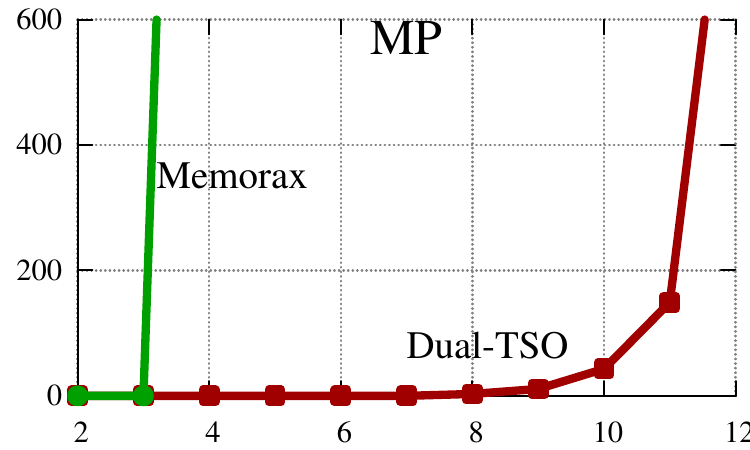}
  \end{subfigure}
  \begin{subfigure}[b]{0.42\linewidth}
    \centering
    \includegraphics[width=0.99\textwidth]{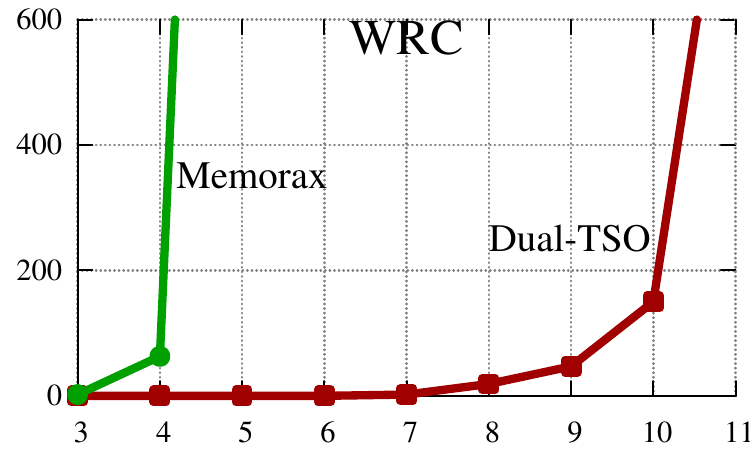} 
  \end{subfigure} \\
  \begin{subfigure}[b]{0.42\linewidth}
    \centering
    \includegraphics[width=0.99\textwidth]{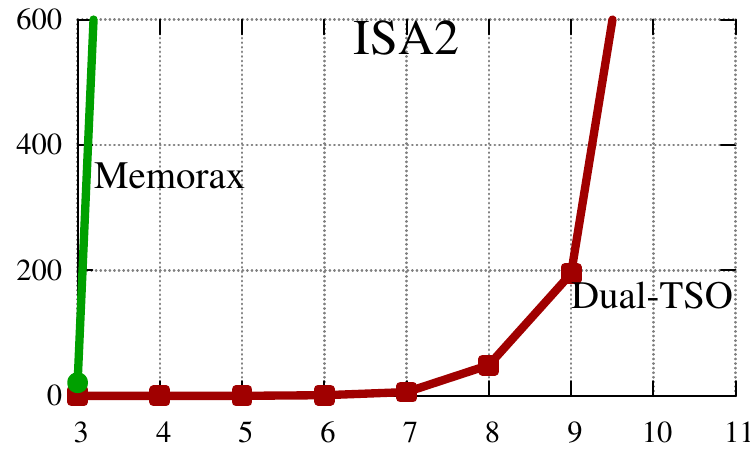}
  \end{subfigure}
  \begin{subfigure}[b]{0.42\linewidth}
    \centering
    \includegraphics[width=0.99\textwidth]{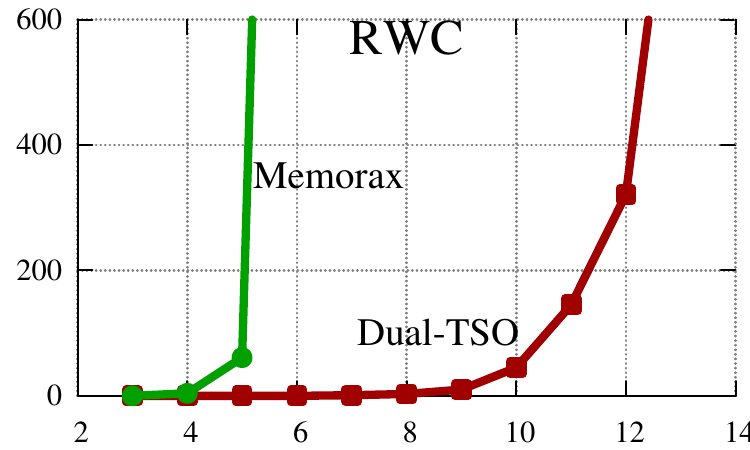}
  \end{subfigure}\\
  \begin{subfigure}[b]{0.42\linewidth}
    \centering
    \includegraphics[width=0.99\textwidth]{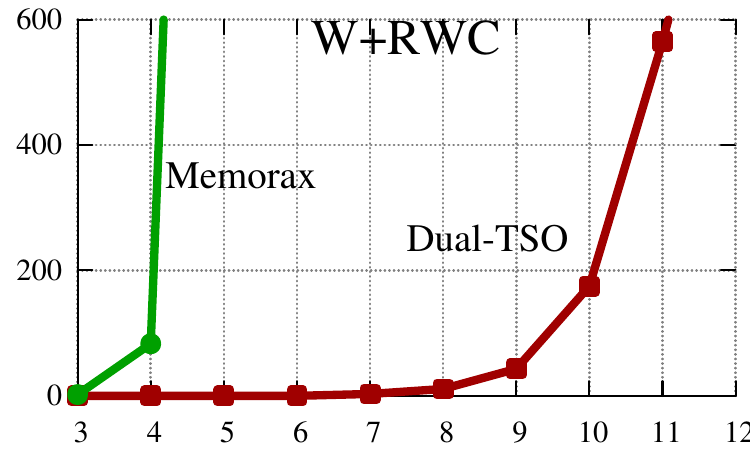}
  \end{subfigure}
  \begin{subfigure}[b]{0.42\linewidth}
    \centering
    \includegraphics[width=0.99\textwidth]{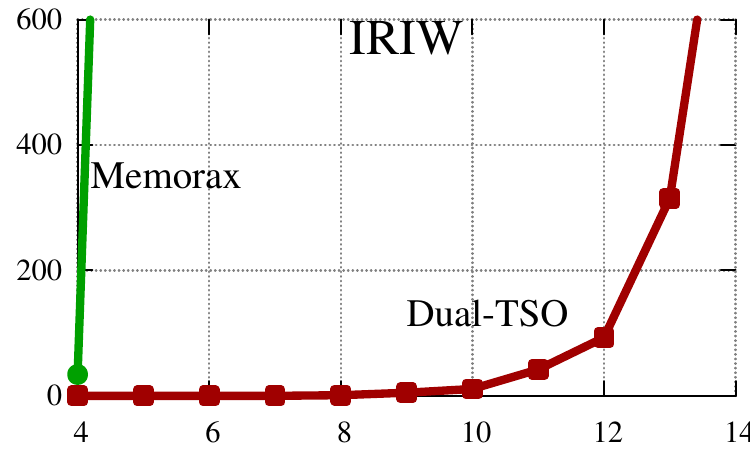} 
  \end{subfigure}
\caption{Running time of  {\sf Memorax} and  {\sf Dual-TSO}$(\genordering)$ by increasing number of processes. The x axis is number of processes and the y axis is running time in seconds. }
\label{increase-processes-figs}
\end{figure}

 \begin{table}[tb]
\begin{tabular}{| l | c | r  r |}
\hline
\multirow{2}{*}{{\bf Program}}	& \multirow{2}{*}{\bf Safe under TSO} & \multicolumn{2}{c|}{{\sf Dual-TSO}$(\trianglelefteq)$}	 	\\ 
	& & $\mathbf{\#T(s)}$ & $\mathbf{\#C}$  	  \\  \hline
SB			& no		& 0.0 & 147							\\ 
 LB			& yes	& 0.6 &\;\;\;\;\;1 028							\\ 
 MP			& yes	& 0.0 & 149							\\ 
 WRC		& yes	& 0.8 & 618							\\ 
ISA2			& yes 	& 4.3 & 1 539							\\ 
RWC		& no		& 0.2 & 293							\\ 
W+RWC\;\;\;\;\;		& no		& 1.5 & 828							\\ 
IRIW			& yes	& 4.6 & 648							\\ 
\hline
\end{tabular}
\caption{Parameterized verification with {\sf Dual-TSO}$(\trianglelefteq)$.}
\label{pdual-table}
\end{table}

\paragraph{\bf Verification of Parameterized Concurrent Systems.}
 The second set compares the scalability of {\sf Memorax} and {\sf Dual-TSO} while increasing the number of processes.  The results are given in 
 Figure~\ref{increase-processes-figs}.  
 We observe that although the algorithms implemented by 
 {\sf Dual-TSO}$(\genordering)$ and {\sf Memorax} 
 have the same (non-primitive recursive) lower bound (in theory), 
  {\sf Dual-TSO}$(\genordering)$
 scales better than {\sf Memorax} in all these benchmarks. 
 In fact, {\sf Memorax} can only handle  benchmarks  with {at most} 5 processes while
 {\sf Dual-TSO} can handle  benchmarks with more  processes.
We conjecture
that this is due to 
 the important advantages of the 
 Dual TSO semantics. In fact, the Dual TSO semantics   
 transforms the load buffers into lossy channels without adding the costly overhead 
 of memory snapshots
 that was necessary in the case of {\sf Memorax}.
The absence of this extra overhead means  that our tool generates less  configurations (due to the ordering)  and this results in a  better performance and scalability.

 Table \ref{pdual-table} presents the running time and the number of generated configurations when  checking the state reachability problem  for  the parameterized versions of the benchmarks in Figure~\ref{increase-processes-figs}
  with {\sf Dual-TSO}$(\trianglelefteq)$. 
It should be emphasized that
{\sf Dual-TSO}$(\trianglelefteq)$
and
 {\sf Dual-TSO}$(\genordering)$
 have the same results
 for the reachability problems in these benchmarks.
 We observe that the verification of these parameterized systems is  much more efficient than verification of bounded-size instances (starting from a number of processes of 3 or 4), especially concerning memory consumption (which is given in terms  of number of generated configurations).
 The reason behind is that  the size of the  generated  minor sets  in the analysis of a  parameterized system  are usually  smaller than the size of the  generated minor sets  during the analysis of  an instance of the system with a large number of processes. In fact, during the analysis of a parameterised concurrent system, the number of considered processes in the generated minimal configurations  is usually very small. Observe that, in the case of  concurrent systems, the  number of considered processes in the generated minimal configurations is equal to the number of processes in the given system. %

\section{Conclusion}
In this paper, we have presented an alternative (yet equivalent) semantics to the classical  one  for the TSO memory model that is more amenable for efficient algorithmic verification and for extension to parametric verification. This new semantics allows us to understand the TSO memory model in a totally different way compared to the classical semantics. Furthermore, the proposed  semantics offers several important advantages
from the point of view of formal reasoning and program verification.
First,  the Dual TSO semantics allows transforming the load buffers
to {\it lossy} channels (in the sense that 
the processes  can lose any  message situated at the head of any load buffer in non-deterministic manner) 
 without adding the 
costly overhead that was necessary in the case of  store buffers. 
This means that we can apply the theory of 
{\it well-structured systems} \cite{DBLP:journals/bsl/Abdulla10,abdulla-general-96,FinkelS01} in a straightforward manner leading 
to a much simpler proof of decidability of safety properties.
Second, the absence of extra overhead means  that we obtain 
more efficient algorithms and better scalability 
(as shown by our experimental results).
Finally, the Dual TSO semantics allows extending the framework to
perform {\it parameterized verification} which is an
important paradigm in concurrent program verification.

In the future, we plan to apply our techniques to other memory models
and to combine with predicate abstraction for handling programs with unbounded data domain.


  
\bibliographystyle{alpha} 
\bibliography{journal-abdulla}

\newcommand{\etalchar}[1]{$^{#1}$}
\begin{thebibliography}{AAC{\etalchar{+}}12b}

\bibitem[AAA{\etalchar{+}}15]{tacas15:tso}
P.~Abdulla, S.~Aronis, M.F. { Atig}, B.~Jonsson, C.~Leonardsson, and
  K.~Sagonas.
\newblock Stateless model checking for {TSO} and {PSO}.
\newblock In {\em TACAS}, volume 9035 of {\em LNCS}, pages 353--367. Springer,
  2015.

\bibitem[AABN16]{DBLP:conf/concur/AbdullaABN16}
Parosh~Aziz Abdulla, Mohamed~Faouzi Atig, Ahmed Bouajjani, and Tuan~Phong Ngo.
\newblock The benefits of duality in verifying concurrent programs under {TSO}.
\newblock In {\em {CONCUR}}, volume~59 of {\em LIPIcs}, pages 5:1--5:15.
  Schloss Dagstuhl - Leibniz-Zentrum fuer Informatik, 2016.

\bibitem[AABN17]{abdullaABN17}
Parosh~Aziz Abdulla, Mohamed~Faouzi Atig, Ahmed Bouajjani, and Tuan~Phong Ngo.
\newblock Context-bounded analysis for {POWER}.
\newblock In {\em {TACAS} 2017}, pages 56--74, 2017.

\bibitem[AAC{\etalchar{+}}12a]{DBLP:conf/tacas/AbdullaACLR12}
P.A. Abdulla, M.F. Atig, Y.F. Chen, C.~Leonardsson, and A.~Rezine.
\newblock Counter-example guided fence insertion under {TSO}.
\newblock In {\em {TACAS} 2012}, pages 204--219, 2012.

\bibitem[AAC{\etalchar{+}}12b]{DBLP:conf/sas/AbdullaACLR12}
Parosh~Aziz Abdulla, Mohamed~Faouzi Atig, Yu{-}Fang Chen, Carl Leonardsson, and
  Ahmed Rezine.
\newblock Automatic fence insertion in integer programs via predicate
  abstraction.
\newblock In {\em {SAS} 2012}, pages 164--180, 2012.

\bibitem[AAC{\etalchar{+}}13]{DBLP:conf/tacas/AbdullaACLR13}
P.A. Abdulla, M.F. Atig, Y.F. Chen, C.~Leonardsson, and A.~Rezine.
\newblock Memorax, a precise and sound tool for automatic fence insertion under
  {TSO}.
\newblock In {\em {TACAS}}, pages 530--536, 2013.

\bibitem[AAJL16]{DBLP:conf/cav/AbdullaAJL16}
Parosh~Aziz Abdulla, Mohamed~Faouzi Atig, Bengt Jonsson, and Carl Leonardsson.
\newblock Stateless model checking for {POWER}.
\newblock In {\em {CAV}}, volume 9780 of {\em LNCS}, pages 134--156. Springer,
  2016.

\bibitem[AALN15]{DBLP:conf/netys/AbdullaALN15}
Parosh~Aziz Abdulla, Mohamed~Faouzi Atig, Magnus L{\aa}ng, and Tuan~Phong Ngo.
\newblock Precise and sound automatic fence insertion procedure under {PSO}.
\newblock In {\em {NETYS} 2015}, pages 32--47, 2015.

\bibitem[AAP15]{DBLP:conf/esop/AbdullaAP15}
P.A. Abdulla, M.F. Atig, and N.T. Phong.
\newblock The best of both worlds: Trading efficiency and optimality in fence
  insertion for {TSO}.
\newblock In {\em {ESOP} 2015}, pages 308--332, 2015.

\bibitem[ABBM10]{ABBM10}
M.~F. Atig, A.~Bouajjani, S.~Burckhardt, and M.~Musuvathi.
\newblock On the verification problem for weak memory models.
\newblock In {\em POPL}, 2010.

\bibitem[ABBM12]{ABBM12}
M.F. Atig, A.~Bouajjani, S.~Burckhardt, and M.~Musuvathi.
\newblock What's decidable about weak memory models?
\newblock In {\em ESOP}, volume 7211 of {\em LNCS}, pages 26--46. Springer,
  2012.

\bibitem[Abd10]{DBLP:journals/bsl/Abdulla10}
Parosh~Aziz Abdulla.
\newblock Well (and better) quasi-ordered transition systems.
\newblock {\em Bulletin of Symbolic Logic}, 16(4):457--515, 2010.

\bibitem[ABP11]{AtigBP11}
M.F. Atig, A.~Bouajjani, and G.~Parlato.
\newblock Getting rid of store-buffers in {TSO} analysis.
\newblock In {\em CAV}, volume 6806 of {\em LNCS}, pages 99--115. Springer,
  2011.

\bibitem[ACJT96]{abdulla-general-96}
P.A. Abdulla, K.~Cerans, B.~Jonsson, and Y.K. Tsay.
\newblock General decidability theorems for infinite-state systems.
\newblock In {\em LICS'96}, pages 313--321. IEEE Computer Society, 1996.

\bibitem[AG96]{adve-gharachorloo-96}
S.~Adve and K.~Gharachorloo.
\newblock Shared memory consistency models: a tutorial.
\newblock {\em Computer}, 29(12), 1996.

\bibitem[AH90]{DBLP:conf/isca/AdveH90}
S.~Adve and M.~D. Hill.
\newblock Weak ordering - a new definition.
\newblock In {\em ISCA}, 1990.

\bibitem[AKNT13]{DBLP:conf/esop/AlglaveKNT13}
J.~Alglave, D.~Kroening, V.~Nimal, and M.~Tautschnig.
\newblock Software verification for weak memory via program transformation.
\newblock In {\em ESOP}, volume 7792 of {\em LNCS}, pages 512--532. Springer,
  2013.

\bibitem[AKT13]{AlglaveKT13}
J.~Alglave, D.~Kroening, and M.~Tautschnig.
\newblock Partial orders for efficient bounded model checking of concurrent
  software.
\newblock In {\em CAV}, volume 8044 of {\em LNCS}, pages 141--157, 2013.

\bibitem[AMT14]{DBLP:journals/toplas/AlglaveMT14}
Jade Alglave, Luc Maranget, and Michael Tautschnig.
\newblock Herding cats: Modelling, simulation, testing, and data mining for
  weak memory.
\newblock {\em {ACM} TOPLAS}, 36(2):7:1--7:74, 2014.

\bibitem[BAM07]{BAM07}
S.~Burckhardt, R.~Alur, and M.~M.~K. Martin.
\newblock {CheckFence}: checking consistency of concurrent data types on
  relaxed memory models.
\newblock In {\em PLDI}, pages 12--21. ACM, 2007.

\bibitem[BDM13]{DBLP:conf/esop/BouajjaniDM13}
Ahmed Bouajjani, Egor Derevenetc, and Roland Meyer.
\newblock Checking and enforcing robustness against {TSO}.
\newblock In {\em ESOP}, volume 7792 of {\em LNCS}, pages 533--553. Springer,
  2013.

\bibitem[BM08]{BM2008}
Sebastian Burckhardt and Madanlal Musuvathi.
\newblock Effective program verification for relaxed memory models.
\newblock In {\em CAV}, volume 5123 of {\em LNCS}, pages 107--120. Springer,
  2008.

\bibitem[BSS11]{BSS2011}
Jacob Burnim, Koushik Sen, and Christos Stergiou.
\newblock Testing concurrent programs on relaxed memory models.
\newblock In {\em ISSTA}, pages 122--132. ACM, 2011.

\bibitem[DL15]{DBLP:conf/oopsla/DemskyL15}
Brian Demsky and Patrick Lam.
\newblock Satcheck: Sat-directed stateless model checking for {SC} and {TSO}.
\newblock In {\em {OOPSLA} 2015}, pages 20--36. {ACM}, 2015.

\bibitem[DM14]{DM14}
Egor Derevenetc and Roland Meyer.
\newblock Robustness against {Power} is {PSpace}-complete.
\newblock In {\em ICALP (2)}, volume 8573 of {\em LNCS}, pages 158--170.
  Springer, 2014.

\bibitem[DMVY13]{DBLP:conf/sas/DanMVY13}
A.~Marian Dan, Y.~Meshman, M.~T. Vechev, and E.~Yahav.
\newblock Predicate abstraction for relaxed memory models.
\newblock In {\em SAS}, volume 7935 of {\em LNCS}, pages 84--104. Springer,
  2013.

\bibitem[DMVY17]{Dan201762}
Andrei Dan, Yuri Meshman, Martin Vechev, and Eran Yahav.
\newblock Effective abstractions for verification under relaxed memory models.
\newblock {\em Computer Languages, Systems and Structures}, 47, Part 1:62--76,
  2017.

\bibitem[DSB86]{DBLP:conf/isca/DuboisSB86}
M.~Dubois, C.~Scheurich, and F.~A. Briggs.
\newblock Memory access buffering in multiprocessors.
\newblock In {\em ISCA}, 1986.

\bibitem[FS01]{FinkelS01}
A.~Finkel and Ph. Schnoebelen.
\newblock Well-structured transition systems everywhere!
\newblock {\em Theor. Comput. Sci.}, 256(1-2):63--92, 2001.

\bibitem[HH16]{oopsla16}
Shiyou Huang and Jeff Huang.
\newblock Maximal causality reduction for {TSO} and {PSO}.
\newblock In {\em {OOPSLA} 2016}, pages 447--461, 2016.

\bibitem[Hig52]{higman:divisibility}
G.~Higman.
\newblock Ordering by divisibility in abstract algebras.
\newblock {\em Proc.\ London Math.\ Soc. (3)}, 2(7):326--336, 1952.

\bibitem[HVQF16]{DBLP:conf/pdp/HeVQF16}
Mengda He, Viktor Vafeiadis, Shengchao Qin, and Jo{\~{a}}o~F. Ferreira.
\newblock Reasoning about fences and relaxed atomics.
\newblock In {\em 24th Euromicro International Conference on Parallel,
  Distributed, and Network-Based Processing, {PDP} 2016, Heraklion, Crete,
  Greece, February 17-19, 2016}, pages 520--527, 2016.

\bibitem[KVY10]{KVY2010}
Michael Kuperstein, Martin~T. Vechev, and Eran Yahav.
\newblock Automatic inference of memory fences.
\newblock In {\em FMCAD}, pages 111--119. IEEE, 2010.

\bibitem[KVY11]{KVY2011}
Michael Kuperstein, Martin~T. Vechev, and Eran Yahav.
\newblock Partial-coherence abstractions for relaxed memory models.
\newblock In {\em PLDI}, pages 187--198. ACM, 2011.

\bibitem[Lam79]{lamport-79}
L.~Lamport.
\newblock How to make a multiprocessor computer that correctly executes
  multiprocess programs.
\newblock {\em IEEE Trans. Comp.}, C-28(9), 1979.

\bibitem[LNP{\etalchar{+}}12]{DBLP:conf/pldi/LiuNPVY12}
Feng Liu, Nayden Nedev, Nedyalko Prisadnikov, Martin~T. Vechev, and Eran Yahav.
\newblock Dynamic synthesis for relaxed memory models.
\newblock In {\em {PLDI} '12}, pages 429--440, 2012.

\bibitem[LV15]{DBLP:conf/icalp/LahavV15}
Ori Lahav and Viktor Vafeiadis.
\newblock Owicki-gries reasoning for weak memory models.
\newblock In {\em Automata, Languages, and Programming - 42nd International
  Colloquium, {ICALP} 2015, Kyoto, Japan, July 6-10, 2015, Proceedings, Part
  {II}}, pages 311--323, 2015.

\bibitem[LV16]{DBLP:conf/fm/LahavV16}
Ori Lahav and Viktor Vafeiadis.
\newblock Explaining relaxed memory models with program transformations.
\newblock In {\em {FM} 2016}, pages 479--495, 2016.

\bibitem[OSS09]{OSS2009}
S.~Owens, S.~Sarkar, and P.~Sewell.
\newblock A better x86 memory model: x86-tso.
\newblock In {\em TPHOL}, 2009.

\bibitem[SSO{\etalchar{+}}10]{SSONM2010}
P.~Sewell, S.~Sarkar, S.~Owens, F.~Z. Nardelli, and M.~O. Myreen.
\newblock x86-tso: A rigorous and usable programmer's model for x86
  multiprocessors.
\newblock {\em CACM}, 53, 2010.

\bibitem[TLF{\etalchar{+}}16]{eps402285}
Ermenegildo Tomasco, Truc~Nguyen Lam, Bernd Fischer, Salvatore~La Torre, and
  Gennaro Parlato.
\newblock Embedding weak memory models within eager sequentialization.
\newblock October 2016.

\bibitem[TLI{\etalchar{+}}16]{fmcad16}
Ermenegildo Tomasco, Truc~Nguyen Lam, Omar Inverso, Bernd Fischer, Salvatore~La
  Torre, and Gennaro Parlato.
\newblock Lazy sequentialization for tso and pso via shared memory
  abstractions.
\newblock In {\em {FMCAD}’16}, pages 193--200, 2016.

\bibitem[TW16]{DBLP:conf/ictac/TravkinW16}
Oleg Travkin and Heike Wehrheim.
\newblock Verification of concurrent programs on weak memory models.
\newblock In {\em {ICTAC} 2016}, pages 3--24, 2016.

\bibitem[Vaf15]{DBLP:conf/popl/Vafeiadis15}
Viktor Vafeiadis.
\newblock Separation logic for weak memory models.
\newblock In {\em Proceedings of the Programming Languages Mentoring Workshop,
  PLMW@POPL 2015, Mumbai, India, January 14, 2015}, page 11:1, 2015.

\bibitem[WG94]{sparc-v9-manual}
D.~Weaver and T.~Germond, editors.
\newblock {\em The SPARC Architecture Manual Version 9}.
\newblock PTR Prentice Hall, 1994.

\bibitem[YGLS04]{yang-gopalakrishnan-PDPS04}
Y.~Yang, G.~Gopalakrishnan, G.~Lindstrom, and K.~Slind.
\newblock Nemos: A framework for axiomatic and executable specifications of
  memory consistency models.
\newblock In {\em IPDPS}. IEEE, 2004.

\bibitem[ZKW15]{Zhang:pldi15}
N.~Zhang, M.~Kusano, and C.~Wang.
\newblock Dynamic partial order reduction for relaxed memory models.
\newblock In {\em PLDI}, pages 250--259. ACM, 2015.

\end{thebibliography}
%


\appendix

\section{Proof of Theorem~\ref{DTSO:TSO:equivalence:theorem}}
\label{proofs:DTSO:TSO:equivalence:theorem}
We prove the  theorem by showing its {\it if direction} and then {\it only if direction}. 
%
In the following, for a TSO (DTSO)-configuration 
$\conf=\tuple{\statemapping,\buffermapping,\mem}$, we
use $\onestatemappingof{\conf}$, $\onebuffermappingof{\conf}$, and
$\onememoryof{\conf}$ to denote
$\statemapping$, $\buffermapping$, and $\mem$ respectively.

\subsection
{\bf From Dual TSO to TSO}
We show the {\em if direction} of Theorem \ref{DTSO:TSO:equivalence:theorem}. Consider a DTSO-computation 
$$\comp_{\sf DTSO} = \conf_0\dtsomovesto{\transition_1} \conf_1 \dtsomovesto{\transition_2} \conf_2  \cdots \dtsomovesto{\transition_{n-1}} \conf_{n-1} \dtsomovesto{\transition_n} \conf_n$$
 where $c_0=\dinitconf$ and $\conf_i$ is of the form $\tuple{\statemapping_{i},\buffermapping_i,\mem_i}$ for all $i : 1 \leq i \leq n$ with $\statemapping_n=\statemapping_{target}$ and $\buffermapping_n(\proc)=\epsilon$ for all $\proc \in \procs$. 
 We will derive a TSO-computation $\comp_{\sf TSO}$ such that $\targetof{\comp_{\sf TSO}}$  is a configuration 
 of the form $\tuple{\onestatemappingof{\conf_n}, \buffermapping,\onememoryof{\conf_n}}$ where $\buffermapping(\proc)=\epsilon$ for all $\proc \in \procs$.
 
First, we  define some functions that we will use in the construction 
of the computation $\comp_{\sf TSO}$.
Then, we define a sequence of TSO-configurations that appear in $\comp_{\sf TSO}$.
Finally, we show that the  TSO-computation $\comp_{\sf TSO}$  exists. In particular, 
the target configuration $\targetof{\comp_{\sf TSO}}$ has the same local states as the target $\conf_n$ of the DTSO-computation $\comp_{\sf DTSO}$.

Let $1 \leq i_1 <  i_2 < \cdots <i_k \leq n$  be the sequence of indices such that $t_{i_1}t_{i_2} \ldots t_{i_k}$ is the sequence of write or atomic read-write operations occurring in the computation $\comp_{\sf DTSO}$. In the following, we assume  that $i_0=0$.

For each   $j: 0 \leq j \leq n$, we associate a mapping function $\mathsf{index}_j:  \procs \ra \{0,\ldots,k\}^* $ that associates for each process $\proc\in\procs$ and each message  
at the  position $\ell : 1  \leq \ell \leq \sizeof{\twobuffermappingof{\conf_j}\proc}$
in the load buffer $\twobuffermappingof{\conf_j}\proc$
  the index $\mathsf{index}_j(\proc)(\ell)$, i.e., the index of the last write or atomic write-read operations at the moment   this message has been added to the load buffer. 
Formally, we define $\mathsf{index}_j$ 
as follows: 
\begin{enumerate}
\item
$\mathsf{index}_0(\proc):=\epsilon$ for all $\proc \in \procs$.
\item
Consider j such that $0\leq j < n$.
Recall that
 $\conf_{j} \dtsomovesto{t_{j+1}} \conf_{j+1}$ with $\transition_{j+1} \in \transitions_\proc \cup \transitions'_{\proc}$. We define $\mathsf{index}_{j+1}$  based on $\mathsf{index}_{j}$:

\begin{itemize}
\item {\sf Nop, read, fence, arw:} If 
$\transition_{j+1}$ is  of  the following forms   $\tuple{\state,\nop,\state'}$, $\tuple{\state,\rop(\xvar,\data),\state'}$, $\tuple{\state,\fenceop,\state'}$, or $\tuple{\state,\arw(\xvar,\data,\data'),\state'}$, then $\mathsf{index}_{j+1}:=\mathsf{index_{j}}$.

\item {\sf Write:} If
$\transition_{j+1}$  is of the form $\tuple{\state,\wop(\xvar,\data),\state'}$, then  $\mathsf{index}_{j+1}:=\mathsf{index}_j\update{\proc}{r \app \mathsf{index}_j(\proc)}$ with $i_r=j+1$.

\item {\sf Propagate:}
If $\transition_{j+1}$ is of the form $\propagate_\proc^\xvar$, then $\mathsf{index}_{j+1}:=\mathsf{index}_j\update{\proc}{r \app \mathsf{index}_j(\proc)}$ where $r: 0 \leq r \leq k$  is the maximal index such that  $i_r \leq j+1$.

\item {\sf Delete:}
If $\transition_{j+1}$ is  of the form $\delete_\proc$, then  $\mathsf{index}_{j}:=\mathsf{index}_{j+1}\update{\proc}{ \mathsf{index}_{j+1}(\proc) \app  r}$ with $r= \mathsf{index}_{j}(\proc)(\sizeof{\mathsf{index}_j(\proc)})$.

\end{itemize}
\end{enumerate}

Next, we associate for each  process $\proc \in \procs$ and  $j : 0 \leq j \leq n$, the  memory view  ${\sf view}_\proc(\conf_j)$ of the process $\proc$ in the configuration $\conf_{j}$  as follows:

\begin{enumerate}

\item If $\twobuffermappingof{\conf_j}\proc=\epsilon$, then    ${\sf view}_\proc(\conf_j):=r$ where  $r: 0 \leq r \leq k$ is the maximal index such that $i_r \leq j$.

\item If $\twobuffermappingof{\conf_j}\proc \neq \epsilon$, then ${\sf view}_\proc(\conf_j):={\mathsf{index}_j(\proc)}(\sizeof{\mathsf{index_j}(\proc)})$.

\end{enumerate}

\begin{exa}
\label{dtso-tso-contruction:example1}
We give an example of how to calculate the functions {\tt index} and {\tt view} for a DTSO-computation. 
Let consider the following DTSO-computation 
$$
\comp_{\sf DTSO} = \conf_0\dtsomovesto{\transition_1} \conf_1 \dtsomovesto{\transition_2} \conf_2  
\dtsomovesto{\transition_3} \conf_3
\dtsomovesto{\transition_4} \conf_4
\dtsomovesto{\transition_5} \conf_5  
$$
containing only transitions of a process $\proc$
with two variables $\xvar$ and $\yvar$
where
$\conf_i=\tuple{\statemapping_i,\buffermapping_i,\mem_i}$ for all $i: 0\leq i \leq n=5$
such that:
\begin{align*} 
\statemapping_0(\proc) &= \state_0, & \buffermapping_0(\proc) &= \emptyword, &  \mem_0(\xvar) &= 0,  \mem_0(\yvar) = 0, &\transition_1&=\tuple{\state_0,\wop(\xvar,1),\state_1},  \\
\statemapping_1(\proc) &= \state_1, & \buffermapping_1(\proc) &= \tuple{\xvar,1,{\textsc{own}}}, &  \mem_1(\xvar) &= 1, \mem_1(\yvar) = 0, & \transition_2&=\propagate_\proc^\yvar,  \\
\statemapping_2(\proc) &= \state_1, & \buffermapping_2(\proc) &= (\yvar,0)\app\tuple{\xvar,1,{\textsc{own}}}, &  \mem_2(\xvar) &= 1,   \mem_2(\yvar) = 0, & \transition_3&=\delete_\proc,  \\
 \statemapping_3(\proc) &= \state_1, & \buffermapping_3(\proc) &= (\yvar,0), &  \mem_3(\xvar) &= 1,   \mem_3(\yvar) = 0, & \transition_4&=(\state_1,\rop(\yvar,0),\state_2),  \\
\statemapping_4(\proc) &= \state_2, & \buffermapping_4(\proc) &= (\yvar,0), &  \mem_4(\xvar) &= 1,   \mem_4(\yvar) = 0, & \transition_5&=  \delete_\proc,\\
 \statemapping_5(\proc) &= \state_2, & \buffermapping_5(\proc) &= \emptyword, &  \mem_5(\xvar) &= 1,   \mem_5(\yvar) = 0.
\end{align*}

We note that
$n=5$ and $\comp_{\sf DTSO}$ contains only transitions of the process $\proc$.
We also note that
$k=1$ and 
$i_1=1$ is the index of the only write transition
$\transition_1$ occurring in the computation
$\comp_{\sf DTSO}$.
Following the above definitions of  {\tt index} and {\tt view}, we define the functions {\tt index} and {\tt view}  as follows:
\begin{enumerate}
\item
For each $j:0\leq j \leq n=5$, we define the mapping function $\mathsf{index}_j(\proc)$: 
\begin{align*} 
\mathsf{index}_0(\proc)&=\emptyword, & \mathsf{index}_1(\proc)&=1, & \mathsf{index}_2(\proc) &=1.1,\\
\mathsf{index}_3(\proc)&=1, & \mathsf{index}_4(\proc)&=1, &\mathsf{index}_5(\proc)&=\emptyword.
\end{align*} 
\item
For each $j:0\leq j \leq n=5$, we define the  memory view  ${\sf view}_\proc(\conf_j)$:
\begin{align*} 
{\sf view}_\proc(\conf_0)&=0, & {\sf view}_\proc(\conf_1)&=1, & {\sf view}_\proc(\conf_2)&=1,\\
{\sf view}_\proc(\conf_3)&=1, & {\sf view}_\proc(\conf_4)&=1, & {\sf view}_\proc(\conf_5)&=1. \tag*{$\triangle$}
\end{align*}
\end{enumerate}
\end{exa}

Now, let $\prec$ be an arbitrary total order on the set of processes and let $\proc_{\it min}$ and $\proc_{\it max}$  be the smallest and largest elements of $\prec$ respectively.  For $\proc \neq \proc_{\it max}$, we define $\mathit{succ}(\proc)$ to be the successor of $\proc$ wrt.\ $\prec$, i.e., $\proc \prec \mathit{succ}(\proc)$ and there is no $\proc'$ with $\proc \prec \proc' \prec \mathit{succ}(\proc)$. We define $\mathit{prev}(\proc)$ for $\proc \neq \proc_{\it min}$  analogously.

The computation $\comp_{\sf TSO}$ will consist of $k+1$ phases (henceforth referred to as  the phases $0,1,2, \ldots, k$). In fact, $\comp_{\sf TSO}$ will have the {\it same} sequence of memory updates as $\comp_{\sf DTSO}$. At the phase $r$, the computation $\comp_{\sf TSO}$ {\it simulates} the movements of the processes where their memory view index  is $r$. The order in which the processes are simulated during phase $r$ is defined by the ordering $\prec$. First, process $\proc_{\it min}$ will perform a sequence of transitions. This sequence is derived from the sequence of transitions it performs in $\comp_{\sf DTSO}$ where its memory view index is $r$,
including 
``no'',
write,
read,
fence 
transitions.
Then, the next process performs its transitions. This continues until $\proc_{\it max}$ has made all its transitions. When all processes have performed their transitions in phase $r$, 
we execute exactly
one update transition (possibly with a write transition)
or one atomic read-write transition
in order to move to  phase 
$r+1$.
We start the
phase $r+1$  by letting $\proc_{\it min}$ execute its transitions, and so on. 

Formally, we define a {\em scheduling function}
$\alpha(r,\proc,\ell)$ 
that gives
for each
$r: 0 \leq r \leq k$, $\proc \in \procs$, and $\ell \geq 1$
 a natural number $j: 0 \leq j \leq n$
 such that 
 process $\proc$ executes the transition
 $\transition_j$ as its $\ell^{th}$ transition during phase $r$. 
The scheduling function $\alpha$ is defined as  follows where $r: 0 \leq r \leq k$, $\proc \in \procs$, and $\ell \geq 0$:

\begin{enumerate}

\item $\alpha(r,\proc,\ell+1)$ is defined to be the smallest $j$ such that $\alpha(r,\proc,\ell) < j$, $\transition_j \in \transitions_\proc$ and ${\sf view}_\proc(\conf_j)=r$. Intuitively, the $(\ell+1)^{th}$ transition of process $\proc$ during phase $r$ is defined by the next transition from $\transition_{\alpha(r,\proc,\ell)}$ that belongs to $\transitions_\proc$. Notice that ${\alpha(r,\proc,\ell+1)}$ is defined only for finitely many $\ell$.

\item 
If $\{{j}\,|\,{{\sf view}_\proc(\conf_j)=r}\}\neq \emptyset$, we define
$\alpha(r,\proc, 0):= \mathit{min}\{{j}\,|\,{{\sf view}_\proc(\conf_j)=r}\}$.
Otherwise,  we define $\alpha(r,\proc, 0):=\alpha(r-1,\proc,\sharp(r-1,\proc))$  where
\begin{align*}
\sharp(r,\proc):= \mathit{max}\{\ell\,|\, 
\ell \geq 0,
\alpha(r,\proc,\ell) \,\,\, \text{is defined}\}.
\end{align*}

Intuitively,
phase $r$ starts for process $\proc$ at the point where its memory view index  becomes equal to $r$. Notice that $\alpha(0,\proc,0)=0$ for all $\proc \in \procs$  since all processes are initially in phase $0$.


\end{enumerate}

\begin{exa}
\label{dtso-tso-contruction:example2}
In the following, we show how to calculate  the 
scheduling function $\alpha(r,\proc,\ell)$ and $\sharp(r,\proc)$
 where $r: 0 \leq r \leq k$, $\proc \in \procs$, and $\ell \geq 0$
for the DTSO-computation $\comp_{\sf DTSO}$
given in Example~\ref{dtso-tso-contruction:example1}.
We recall that $k=1$, $n=5$ and the definitions of the two functions {\tt index} and {\tt view} are given in Example~\ref{dtso-tso-contruction:example1}.
We also recall that $\comp_{\sf DTSO}$ contains only transitions of the process $\proc$.
The constructed TSO-computation $\tsocomp$ from $\comp_{\sf DTSO}$ will consist of $k+1=2$ phases, referred as the phase $0$ and the phase $1$.
In order to define the transitions of the process $\proc$ in different phases,
for each $r: 0 \leq r \leq k=1$ and $\ell \geq 0$,  the 
scheduling function $\alpha(r,\proc,\ell)$ and $\sharp(r,\proc)$ is defined as follows:
\begin{align*}
\alpha(0,\proc,0)&=0, & \alpha(1,\proc,0)&=1, & \alpha(1,\proc,1)&=4,\\
\sharp(0,\proc)&=0, &  \sharp(1,\proc)&=1.  \tag*{$\triangle$}
\end{align*}
\end{exa}

In order to define $\comp_{\sf TSO}$, we first define the set of configurations that will appear in $\comp_{\sf TSO}$. In more detail, for each $r: 0\leq r \leq k$, $\proc \in \procs$, and $\ell: 0 \leq \ell \leq \sharp(r,\proc)$, we define a TSO-configuration $\dconf_{r,\proc,\ell}$ based on the DTSO-configurations  that are appearing in $\comp_{\sf DTSO}$. We will define $\dconf_{r, \proc,\ell}$ by defining its local states, buffer contents, and memory state. 
\begin{enumerate}
\item
We define the {\it local states} of the processes as follows:
\begin{itemize}
\item
$\twostatemappingof{\dconf_{r,\proc,\ell}}\proc:=
\twostatemappingof{\conf_{\alpha(r,\proc,\ell)}}\proc$.
After process $\proc$ has performed its $\ell^{\it th}$ transition during phase $r$,
its local state is identical to its local state in the corresponding
DTSO-configuration $\conf_{\alpha(r,\proc,\ell)}$.
\item
If $\proc'\procordering\proc$ then
$\twostatemappingof{\dconf_{r,\proc,\ell}}{\proc'}:=
\twostatemappingof{\conf_{\alpha(r,\proc',\numberof{r}{\proc'})}}{\proc'}$,
i.e.\ the state of $\proc'$ will not change while
$\proc$ is making its moves.
This state is given by the local state of $\proc'$ after it made its last move
during phase $r$.
\item
If $\proc\procordering\proc'$ then
$\twostatemappingof{\dconf_{r,\proc,\ell}}{\proc'}:=
\twostatemappingof{\conf_{\alpha(r,\proc',0)}}{\proc'}$,
i.e.\ the local state of $\proc'$ will not change while
$\proc$ is making its moves.
This state is given by the local state of $\proc'$ when it entered phase $r$
(before it has made any moves during phase $r$).
\end{itemize}
\item
To define the {\it buffer contents},  we give more definitions.
For  a DTSO-message $a$  of the form $(\xvar,\data)$, we define
$\sbtotso{a}$ to be 
$\emptyword$. For a DTSO-message $a$  of the form $(\xvar,\data,{\textsc{own}})$, we define
$\sbtotso{a}$ to be 
$(\xvar,\data)$. 
From that, we define $\sbtotso{\epsilon}=\epsilon$ and 
$\sbtotso{a_1a_2\cdots a_n}:=
\sbtotso{a_1}\app\sbtotso{a_2}
\cdots\sbtotso{a_n}$, i.e., 
we concatenate the results of applying the operation
individually on each $a_i$ with $1 \leq i \leq n$.
We define $\dtsototso{w}$ for a word $w \in \left((\vars\times\dataset) \cup (\vars\times\dataset\times \{\textsc{own}\})\right)^*$ as follows: If $|w| = 0$ then $\dtsototso{w}:=\epsilon$, else $\dtsototso{w}:=\sbtotso{w(1) w(2)\cdots w(|w|-1)}$.
In the following, we give the definition of the buffer contents of $\dconf_{r,\proc,\ell}$:

\begin{itemize}
\item  $\twobuffermappingof{\dconf_{r,\proc,\ell}}\proc:=\dtsototso{\twobuffermappingof{\conf_{\alpha(r,\proc,\ell)}}\proc}$. After process $\proc$ has performed its $\ell^{th}$ transition during phase $r$, the content of its buffer is defined by considering
 the buffer of the corresponding DTSO-configuration $\conf_{\alpha(r,\proc,\ell)}$ and 
 only messages belong to $\proc$ (i.e., of the form $\tuple{\xvar,\data,{\textsc{own}}}$).

\item
If $\proc'\procordering\proc$ then
$\twobuffermappingof{\dconf_{r,\proc,\ell}}{\proc'}:=
\twobuffermappingof{\conf_{\alpha(r,\proc',\numberof{r}{\proc'})}}{\proc'}$.
In a similar manner to the case of states,
if $\proc'\procordering\proc$ then the buffer of $\proc'$ will not change while
$\proc$ is making its moves.
\item
If $\proc\procordering\proc'$ then
$\twobuffermappingof{\dconf_{r,\proc,\ell}}{\proc'}:=
\twobuffermappingof{\conf_{\alpha(r,\proc',0)}}{\proc'}$.
In a similar manner to the case of states,
if $\proc\procordering\proc'$ then the buffer of $\proc'$ will not be changed while
$\proc$ is making its moves.
\end{itemize}
\item
We define the {\it memory state} as follows:
\begin{itemize}
\item
$\onememoryof{\dconf_{r,\proc,\ell}}:=\onememoryof{\conf_{i_r}}$.
This definition is consistent with the fact that
all processes have identical views of the memory when they are in the same phase $r$.
This view is defined by the memory component of $\conf_{i_r}$.
\end{itemize}
\end{enumerate}

\begin{exa}
\label{dtso-tso-contruction:example3}
In the following, we give  the configurations $\dconf_{r,\proc,\ell}$
for all $r: 0\leq r \leq k$, $\proc \in \procs$, and $\ell: 0 \leq \ell \leq \sharp(r,\proc)$
that will appear in the constructed TSO-computation
$\tsocomp$ from $\comp_{\sf DTSO}$
given in Example~\ref{dtso-tso-contruction:example1}.
We call that $k=1$, $n=5$,
and $\comp_{\sf DTSO}$ contains only transitions of the process $\proc$.
We also recall that the
scheduling function $\alpha(r,\proc,\ell)$ and $\sharp(r,\proc)$ 
are given  in  Example~\ref{dtso-tso-contruction:example2}.

For each $r: 0\leq r \leq k=1$  and $\ell: 0 \leq \ell \leq \sharp(r,\proc)$, we define the TSO-configurations $\dconf_{r,\proc,\ell}=\tuple{\statemapping_{r,\proc,\ell}, \buffermapping_{r,\proc,\ell},\mem_{r,\proc,\ell}}$ based on the DTSO-configurations  that are appearing in $\comp_{\sf DTSO}$ as follows:
\begin{align*} 
\dconf_{0,\proc,0}&:  & \statemapping_{0,\proc,0}(\proc) &= \state_0, & \buffermapping_{0,\proc,0}(\proc) &= \emptyword, &  \mem_{0,\proc,0}(\xvar) &= 0,   \mem_{0,\proc,0}(\yvar) = 0, \\
\dconf_{1,\proc,0}&: &   \statemapping_{1,\proc,0}(\proc) &= \state_1, & \buffermapping_{1,\proc,0}(\proc) &= \emptyword, &  \mem_{1,\proc,0}(\xvar) &= 1, \mem_{1,\proc,0}(\yvar) = 0, \\
\dconf_{1,\proc,1}&: &   \statemapping_{1,\proc,0}(\proc) &= \state_2, & \buffermapping_{1,\proc,1}(\proc) &= \emptyword, &  \mem_{1,\proc,1}(\xvar) &= 1,   \mem_{1,\proc,1}(\yvar) = 0. 
\end{align*} 

Finally, we  construct the TSO-computation 
 $$
\tsocomp = \dconf_{0,\proc,0}\tsomovesto{\transition'_1} \dconf'_{0,\proc,0}
\tsomovesto{\transition'_2} \dconf_{1,\proc,0}
\tsomovesto{\transition'_3} \dconf_{1,\proc,1}
$$
where
\begin{align*}
 \dconf'_{0,\proc,0} &=\tuple{\statemapping'_{0,\proc,0}, \buffermapping'_{0,\proc,0},\mem'_{0,\proc,0}},\\
 \transition'_1&=\tuple{\state_0,\wop(\xvar,1),\state_1},\\
 \transition'_2&=\updateop_\proc,\\
  \transition'_3&=\tuple{\state_1,\rop(\yvar,0),\state_2},\\
\end{align*}
and:
\begin{align*} 
\dconf'_{0,\proc,0}&:  & \statemapping'_{0,\proc,0}(\proc) &= \state_1, & \buffermapping'_{0,\proc,0}(\proc) &= (\xvar,1), &  \mem'_{0,\proc,0}(\xvar) &= 0,   \mem'_{0,\proc,0}(\yvar) = 0. 
\end{align*} 
Since there is only one update transition in both two computations $\comp_{\sf DTSO}$ and $\tsocomp$,
it is easy to see that
$\tsocomp$ has the same sequence of memory updates as $\comp_{\sf DTSO}$.
It is also easy to see that
 $\dconf_{0,\proc,0}=\initconf$
 and
  $\dconf_{1,\proc,\numberof{1}{\proc}}= \dconf_{1,\proc,1}=
\tuple{\onestatemappingof{\conf_5},\buffermapping,\onememoryof{\conf_5}}$
where
$\buffermapping(\proc)=
\twobuffermappingof{\conf_5}\proc=\epsilon$.
Therefore, $\tsocomp$ is a witness of the construction.
\myend
\end{exa}

The following lemma  
shows the existence of a TSO-computation $\tsocomp$ that starts from the initial
TSO-configuration and whose target has the same local state definitions
as the target $\conf_n$ of the DTSO-computation $\sbcomp$.
This concludes the proof of the if direction of Theorem~\ref{DTSO:TSO:equivalence:theorem}.
\begin{lem}
\label{tso:comp:lemma}
$\dconf_{0,\minproc,0}\tsomovesto{\tsocomp}\dconf_{k,\maxproc,\numberof{k}{\maxproc}}$ for
some TSO-computation $\tsocomp$.
Furthermore,
$\dconf_{0,\minproc,0}$ is the initial TSO-configuration
and
\begin{align*}
\dconf_{k,\maxproc,\numberof{k}{\maxproc}}=
\tuple{\onestatemappingof{\conf_n},\buffermapping,\onememoryof{\conf_n}}
\end{align*}

where
$\onebuffermappingof\proc=
\twobuffermappingof{\conf_n}\proc=\epsilon$ for all $\proc \in \procs$.
\end{lem}

\begin{proof}
Lemmas~\ref{dtso:tso:singlestep:lemma}--\ref{dtso:tso:next:phase:lemma2}
show that the existence of the computation $\tsocomp$.
Lemma~\ref{final:tso:comp:lemma} and Lemma~\ref{init:tso:comp:lemma}
show the conditions on the initial and target configurations. 
\end{proof}

First, we start by establishing Lemma~\ref{prop-index}, Lemma~\ref{increasing}, and Lemma \ref{bufferdtso} that we will use later.

\begin{lem}
\label{prop-index}
For every $j: 0 \leq j \leq n$ and process $\proc \in \procs$, the following properties hold:

\begin{enumerate}

\item $\sizeof{\mathsf{index}_j(\proc)}=\sizeof{\twobuffermappingof{\conf_j}\proc}$.
\item  For every $\ell_1, \ell_2:1\leq \ell_1 \leq \ell_2 \leq \sizeof{\mathsf{index}_j(\proc)}$, $\mathsf{index}_j(\proc)(\ell_2) \leq \mathsf{index}_j(\proc)(\ell_1) \leq j$.

\item For every $\ell_1, \ell_2:1\leq \ell_1 < \ell_2 \leq \sizeof{\mathsf{index}_j(\proc)}$ such that $\threebuffermappingof{\conf_j}\proc{\ell_1}$ is of the form $\tuple{\xvar,\data,{\textsc{own}}}$, 
 $ \mathsf{index}_j(\proc)(\ell_2) < \mathsf{index}_j(\proc)(\ell_1)$. 
\item For every $\ell: 1 \leq \ell \leq \sizeof{\twobuffermappingof{\conf_j}\proc}$,  if  $r=\mathsf{index}_j(\proc)(\ell)$ and $\threebuffermappingof{\conf_j}\proc\ell$ is of the form $\tuple{\xvar,\data,{\textsc{own}}}$, then $t_{i_r} \in  \transitions_\proc$  and it is  of the form $\tuple{\state,\wop(\xvar,\data),\state'}$.

\item For every $\ell: 1 \leq \ell \leq \sizeof{\twobuffermappingof{\conf_j}\proc}$,  if  $r=\mathsf{index}_j(\proc)(\ell)$ and $\threebuffermappingof{\conf_j}\proc\ell$ is of the form $\tuple{\xvar,\data}$, then $\twomemoryof{\conf_{i_r}}\xvar=\data$.

\item For every $r_1,r_2: 0 \leq r_1 \leq r_2 \leq k$ such that   $r_1= {\textsc{min}}\{\mathsf{index_j(\proc)(\ell)} \,|\, \ell: 1 \leq \ell \leq \sizeof{\mathsf{index}_j(\proc)}\}$, $t_{r_2} \in \transitions_\proc$, and $t_{r_2}$ is of the form  $\tuple{\state,\wop(\xvar,\data),\state'}$, then there is an index $\ell: 1 \leq \ell \leq \sizeof{\mathsf{index}_j(\conf_j)(\proc)} $ such that $\mathsf{index}_j(\proc)(\ell)=r_2$ and  $\threebuffermappingof{\conf_j}\proc\ell=\tuple{\xvar,\data,{\textsc{own}}}$.
\end{enumerate}
\end{lem}

\begin{proof}
The lemma holds following an immediate consequence of the definition of $\mathsf{index}_j$.
\end{proof}

\begin{lem}
\label{increasing}
For every  process $\proc \in \procs$ and index $j : 0 \leq j < n$,  ${\sf view}_\proc(\conf_j) \leq {\sf view}_\proc(\conf_{j+1})$. Furthermore, $i_{{\sf view}_\proc(\conf_j)} \leq j $ and $i_{{\sf view}_\proc(\conf_{j+1})} \leq j+1$.

\end{lem}

\begin{proof}
The lemma holds following an  immediate consequence of the definitions of ${\sf view}_\proc$ and $\mathsf{index}_j$.
\end{proof}

\begin{lem}
\label{bufferdtso}
For every natural number $j$ such that $\alpha(r,\proc,\ell)\leq j <\alpha(r,\proc,\ell+1)-1$, 
$\dtsototso{\twobuffermappingof{\conf_{j}}\proc}=\dtsototso{\twobuffermappingof{\conf_{j+1}}\proc}$. \end{lem}

\begin{proof}
The proof is done by contradiction. Let us assume that there is  some $j: \alpha(r,\proc,\ell)\leq j <\alpha(r,\proc,\ell+1)-1$ such that
\begin{align*}
  \dtsototso{\twobuffermappingof{\conf_{j}}\proc} \neq \dtsototso{\twobuffermappingof{\conf_{j+1}}\proc}.
\end{align*}
Observe that the only three operations that can change the content of the load buffer of the process $\proc$  are write, delete and propagation operations. Since $\transition_{j} \notin \transitions_\proc$ (and so no write operation has been performed) and propagation will append messages of the form $(\xvar,\data)$,  this implies that $\transition_j$ is a delete transition of the process $\proc$ (i.e., $\transition_j=\delete_\proc$). Now, the only case  when  $\dtsototso{\twobuffermappingof{\conf_{j}}\proc} \neq \dtsototso{\twobuffermappingof{\conf_{j+1}}\proc}$ is where  $\twobuffermappingof{\conf_{j}}\proc$ is of the form $w \cdot (\yvar,\data',{\textsc{own}}) \cdot m$ with $m \in  \{(\xvar,\data), (\xvar,\data,{\textsc{own}})\,|\, \xvar \in \vars, \data \in \valset\}$. This implies that  $\twobuffermappingof{\conf_{j+1}}\proc=w \cdot (\yvar,\data',{\textsc{own}})$. Now we can use the third case  of Lemma \ref{prop-index} to prove that ${\sf view}_\proc(\conf_{j+1})>{\sf view}_\proc(\conf_{j})$. This contradicts the fact that ${\sf view}_\proc(\conf_{j+1}) \leq {\sf view}_\proc(\conf_{\alpha(r,\proc,\ell+1)}) $ (see Lemma \ref{increasing}) since  we have ${\sf view}_\proc(\conf_{\alpha(r,\proc,\ell)})={\sf view}_\proc(\conf_{\alpha(r,\proc,\ell+1)})=r$ (by definition),  ${\sf view}_\proc(\conf_{j}) \geq {\sf view}_\proc(\conf_{\alpha(r,\proc,\ell)}) $  (see Lemma \ref{increasing}) and ${\sf view}_\proc(\conf_{j+1})>{\sf view}_\proc(\conf_{j})$.
\end{proof}

Now we can start proving the existence of the computation $\tsocomp$ by showing that we can move from the configuration $\dconf_{r,\proc,\ell}$ to $\dconf_{r,\proc,\ell+1}$ using the transition $\transition_{\alpha(r,\proc,\ell+1)}$.

\begin{lem}
\label{dtso:tso:singlestep:lemma}
If $\ell<\numberof{r}{\proc}$ then
$\dconf_{r,\proc,\ell}\tsomovesto{\transition_{\alpha(r,\proc,\ell+1)}}\dconf_{r,\proc,\ell+1}$.
\end{lem}

\begin{proof}
We recall that $\transition_{\alpha(r,\proc,\ell+1)} \in \transitions_\proc$ by definition.
Therefore, $\transition_{\alpha(r,\proc,\ell+1)}$ is not
a propagation  transition  nor a delete transition. Furthermore,  suppose that $\transition_{\alpha(r,\proc,\ell+1)}$ is  an atomic read-write transition. 
It leads to the fact that  ${\sf view}_\proc(\conf_{\alpha(r,\proc,\ell+1)})  >  {\sf view}_\proc(\conf_{\alpha(r,\proc,\ell)})$, contradicting to the assumption that we are in phase $r$.
Hence, $\transition_{\alpha(r,\proc,\ell+1)}$ is not an atomic read-write transition.

Let $\transition_{\alpha(r,\proc,\ell+1)}\in \transitions_\proc$ be of the form 
$\tuple{\state,\op,\state'}$.
To prove the lemma, we will prove the following properties: 
\begin{enumerate}
\item
$\twostatemappingof{\dconf_{\alpha(r,\proc,\ell)}}\proc=\state$ and
$\onestatemappingof{\dconf_{r,\proc,\ell+1}}=
\onestatemappingof{\dconf_{\alpha(r,\proc,\ell)}}\update{\proc}{\state'}$,
\item
$\twostatemappingof{\dconf_{r,\proc,\ell+1}}{\proc'}=\twostatemappingof{\dconf_{r,\proc,\ell}}{\proc'}$ 
for $\proc'\neq\proc$,
\item
$\twobuffermappingof{\dconf_{r,\proc,\ell+1}}{\proc'}=
\twobuffermappingof{\dconf_{r,\proc,\ell}}{\proc'}$
for $\proc'\neq\proc$,
\item
$\onememoryof{\dconf_{r,\proc,\ell}}=\onememoryof{\dconf_{r,\proc,\ell+1}}=\onememoryof{\conf_{i_r}}$,
\item
The contents of $\twobuffermappingof{\dconf_{r,\proc,\ell}}\proc$
and $\twobuffermappingof{\dconf_{r,\proc,\ell+1}}\proc$ are compatible with the transition  $\transition_{\alpha(r,\proc,\ell+1)}$.
In means that with the properties (1)--(4), the property (5) allows that $\dconf_{r,\proc,\ell}\tsomovesto{\transition_{\alpha(r,\proc,\ell+1)}}\dconf_{r,\proc,\ell+1}$.
\end{enumerate}

We prove the property (1).
We see from definition of $\alpha$
 that $\transition_j\not\in\transitions_\proc$ for all
$j:\alpha(r,\proc,\ell)<j<\alpha(r,\proc,\ell+1)$.
It follows that
$\twostatemappingof{\conf_j}\proc=
\twostatemappingof{\conf_{\alpha(r,\proc,\ell)}}\proc$ for all
$j:\alpha(r,\proc,\ell)<j<\alpha(r,\proc,\ell+1)$.
In particular, we have
$\twostatemappingof{\conf_{\alpha(r,\proc,\ell+1)-1}}\proc=
\twostatemappingof{\conf_{\alpha(r,\proc,\ell)}}\proc$.
%
Then, from the fact that
$\conf_{\alpha(r,\proc,\ell+1)-1}
\dtsomovesto{\transition_\alpha(r,\proc,\ell+1)}
\conf_{\alpha(r,\proc,\ell+1)}$
and the definitions of $\dconf_{r,\proc,\ell}$ and $\dconf_{r,\proc,\ell+1}$, 
we know that $\twostatemappingof{\dconf_{r,\proc,\ell}}\proc=\twostatemappingof{\conf_{\alpha(r,\proc,\ell)}}\proc=
\twostatemappingof{\conf_{\alpha(r,\proc,\ell+1)-1}}\proc=\state$.
It  follows that 
\begin{align*}
  \twostatemappingof{\dconf_{r,\proc,\ell+1}}\proc=\twostatemappingof{\conf_{\alpha(r,\proc,\ell+1)}}\proc=\state'.
\end{align*}
This concludes the property (1).

We prove the property (2).
We see from the definitions of $\dconf_{\alpha(r,\proc,\ell)}$ and $\dconf_{\alpha(r,\proc,\ell+1)}$
that
if $\proc'\procordering\proc$ 
then  
$\twostatemappingof{\dconf_{r,\proc,\ell+1}}{\proc'}=
\twostatemappingof{\conf_{\alpha(r,\proc',\numberof{k}{\proc'})}}{\proc'}=
\twostatemappingof{\dconf_{r,\proc,\ell}}{\proc'}$.
Moreover, we have
if  $\proc\procordering\proc'$ 
then 
\begin{align*}
  \twostatemappingof{\dconf_{r,\proc,\ell+1}}{\proc'}=
  \twostatemappingof{\conf_{\alpha(r,\proc',0)}}{\proc'}=
  \twostatemappingof{\dconf_{r,\proc,\ell}}{\proc'}.
\end{align*}
This concludes the property (2).

We prove the properties (3) and (4). 
In a similar manner to the case of states, we can show 
the property (3). 
By the definitions of $\dconf_{\alpha(r,\proc,\ell)}$ and $\dconf_{\alpha(r,\proc,\ell+1)}$ and the fact that
$\ell < \ell+1\leq\numberof{r}{\proc}$, we have
$\onememoryof{\dconf_{r,\proc,\ell}}=\onememoryof{\conf_{i_r}}=\onememoryof{\dconf_{r,\proc,\ell+1}}$.
This concludes the property (4).

Now, it remains to prove the property (5).
We consider the cases where $\op$ is a write or a read operation.
The other cases can be treated in a similar way.
\begin{itemize}
\item
$\op=\wop(\xvar,\data)$:
We see from  Lemma \ref{bufferdtso} that  for all:
$j:\alpha(r,\proc,\ell)<j<\alpha(r,\proc,\ell+1)$
 $$\dtsototso{\twobuffermappingof{\conf_{\alpha(r,\proc,\ell)}}\proc} =\dtsototso{\twobuffermappingof{\conf_{j}}\proc}$$
 In particular, we have  $$\dtsototso{\twobuffermappingof{\conf_{\alpha(r,\proc,\ell)}}\proc} =\dtsototso{\twobuffermappingof{\conf_{\alpha(r,\proc,\ell+1)-1}}\proc}$$
Then,  
since
$\conf_{\alpha(r,\proc,\ell+1)-1}
\dtsomovesto{\transition_\alpha(r,\proc,\ell+1)}
\conf_{\alpha(r,\proc,\ell+1)}$,
we have
$\twobuffermappingof{\conf_{\alpha(r,\proc,\ell+1)}}\proc= \tuple{\xvar,\data,{\textsc{own}}}\app  \twobuffermappingof{\conf_{\alpha(r,\proc,\ell+1)-1}}\proc$.


We  will show that $\twobuffermappingof{\conf_{\alpha(r,\proc,\ell+1)-1}}\proc\neq \epsilon$ by contradiction. Let us suppose that $\twobuffermappingof{\conf_{\alpha(r,\proc,\ell+1)-1}}\proc= \epsilon$. By definition, we have  ${\sf view}_\proc(\conf_{\alpha(r,\proc,\ell+1)})=r'$ such that $i_{r'}={\alpha(r,\proc,\ell+1)}$. Furthermore, by applying Lemma \ref{increasing} to $\conf_{\alpha(r,\proc,\ell)}$, we know that $i_r \leq {\alpha(r,\proc,\ell)}$.
 Then, since $\alpha(r,\proc,\ell) <\alpha(r,\proc,\ell+1)$ by definition, we have $i_r < i_{r'}$. This contradicts to the fact that ${\sf view}_\proc(\conf_{\alpha(r,\proc,\ell+1)})=r$ by definition.
 Therefore,  we have $\twobuffermappingof{\conf_{\alpha(r,\proc,\ell+1)-1}}\proc\neq \epsilon$.

As a consequence of the fact that $\twobuffermappingof{\conf_{\alpha(r,\proc,\ell+1)-1}}\proc\neq \epsilon$, we know that 
$$\dtsototso{\twobuffermappingof{\conf_{\alpha(r,\proc,\ell+1)}}\proc}= (\xvar,\data) \cdot  \dtsototso{\twobuffermappingof{\conf_{\alpha(r,\proc,\ell)-1}}\proc}$$
Then, since
\begin{align*}
\dtsototso{\twobuffermappingof{\conf_{\alpha(r,\proc,\ell)}}\proc} =\dtsototso{\twobuffermappingof{\conf_{\alpha(r,\proc,\ell+1)-1}}\proc},
\end{align*}
it follows that
$\onebuffermappingof{\dconf_{r,\proc,\ell+1}}=
(\xvar,\data) \cdot \onebuffermappingof{\dconf_{r,\proc,\ell}}$.  Hence this implies that $\dconf_{r,\proc,\ell}\tsomovesto{\transition_{\alpha(r,\proc,\ell+1)}}\dconf_{r,\proc,\ell+1}$.

\item
$\op=\rop(\xvar,\data)$:
We see from
Lemma \ref{bufferdtso} that for all
$j:\alpha(r,\proc,\ell)<j<\alpha(r,\proc,\ell+1)$
  $$\dtsototso{\twobuffermappingof{\conf_{\alpha(r,\proc,\ell)}}\proc} =\dtsototso{\twobuffermappingof{\conf_{j}}\proc}$$  
  In particular,  we have 
  $$\dtsototso{\twobuffermappingof{\conf_{\alpha(r,\proc,\ell)}}\proc} =\dtsototso{\twobuffermappingof{\conf_{\alpha(r,\proc,\ell+1)-1}}\proc}$$
Then, 
since 
$\conf_{\alpha(r,\proc,\ell+1)-1}
\dtsomovesto{\transition_\alpha(r,\proc,\ell+1)}
\conf_{\alpha(r,\proc,\ell+1)}$,
we have
$\twobuffermappingof{\conf_{\alpha(r,\proc,\ell+1)}}\proc=   \twobuffermappingof{\conf_{\alpha(r,\proc,\ell+1)-1}}\proc$.  Therefore, $\twobuffermappingof{\dconf_{r,\proc,\ell+1}}\proc=\twobuffermappingof{\dconf_{r,\proc,\ell}}\proc$.
We consider two cases about the type of the operation $\op$:

\begin{itemize}
\item    {\sf Read own write}:
We see that there is an $i:1\leq i<\sizeof{\twobuffermappingof{\conf_{\alpha(r,\proc,\ell+1)-1}}\proc}$ such that
$\threebuffermappingof{\conf_{\alpha(r,\proc,\ell+1)-1}}\proc{i}=\tuple{\xvar,\data,{\textsc{own}}}$, and that
there are no $j: 1\leq j <i$ and $\data'\in\dataset$  such that $\threebuffermappingof{\conf_{\alpha(r,\proc,\ell+1)-1}}\proc{j}=\tuple{\xvar,\data',{\textsc{own}}}$.  As a consequence, this implies that there is an $i': 1 \leq i' \leq \sizeof{\dtsototso{\twobuffermappingof{\conf_{\alpha(r,\proc,\ell+1)-1}}\proc}}$ such that $\dtsototso{\threebuffermappingof{\conf_{\alpha(r,\proc,\ell+1)-1}}\proc{i'}}=\tuple{\xvar,\data}$ and there are no  $j':1\leq j' <i'$ and $\data'\in\dataset$  such that $\dtsototso{\threebuffermappingof{\conf_{\alpha(r,\proc,\ell+1)-1}}\proc{j}}=\tuple{\xvar,\data'}$. From the fact that
\begin{align*}
  \twobuffermappingof{\dconf_{r,\proc,\ell+1}}\proc=\twobuffermappingof{\dconf_{r,\proc,\ell}}\proc=\dtsototso{\onebuffermappingof{\conf_{\alpha(r,\proc,\ell+1)-1}}},
\end{align*}
we have $\dconf_{r,\proc,\ell}\tsomovesto{\transition_{\alpha(r,\proc,\ell+1)}}\dconf_{r,\proc,\ell+1}.$


\item {\sf Read memory}: 
We consider two cases:

\begin{enumerate}[label=$\triangleright$]

\item 
$\threebuffermappingof{\conf_{\alpha(r,\proc,\ell+1)-1}}\proc{i}=\tuple{\xvar,\data,{\textsc{own}}}$ where  $i=\sizeof{\twobuffermappingof{\conf_{\alpha(r,\proc,\ell+1)-1}}\proc}$ and
there are no $j: 1\leq j <i$ and $\data'\in\dataset$  such that $\threebuffermappingof{\conf_{\alpha(r,\proc,\ell+1)-1}}\proc{j}=\tuple{\xvar,\data',{\textsc{own}}}$:  Since ${\sf view}_\proc(\conf_{\alpha(r,\proc,\ell+1)-1})= r$, this implies from Lemma \ref{prop-index} that $t_{i_r}\in\transitions_\proc$  and  it is of the form $\tuple{\state,\wop(\xvar,\data),\state'}$. Hence, we see that $\twomemoryof{\conf_{i_r}}\xvar=\data$ and thus $\twomemoryof{\dconf_{r,\proc,\ell}}\xvar=\twomemoryof{\dconf_{r,\proc,\ell+1}}\xvar=\data$.
Therefore, we have $\dconf_{r,\proc,\ell}\tsomovesto{\transition_{\alpha(r,\proc,\ell+1)}}\dconf_{r,\proc,\ell+1}$.


\item
 $\tuple{\xvar,\data',{\textsc{own}}}\not\in\twobuffermappingof{\conf_{\alpha(r,\proc,\ell+1)-1}}\proc$
for all 
$\data'\in\dataset$: Thus $\onebuffermappingof{\conf_{\alpha(r,\proc,\ell+1)-1}}$ is of the form $w \cdot (\xvar,\data)$.
 Since ${\sf view}_\proc(\conf_{\alpha(r,\proc,\ell+1)-1})= r$, this implies from Lemma \ref{prop-index} that  $\twomemoryof{\conf_{i_r}}\xvar=\data$ and   thus $\twomemoryof{\dconf_{r,\proc,\ell}}\xvar=\twomemoryof{\dconf_{r,\proc,\ell+1}}\xvar=\data$.
 Therefore, we have
 $\dconf_{r,\proc,\ell}\tsomovesto{\transition_{\alpha(r,\proc,\ell+1)}}\dconf_{r,\proc,\ell+1}$.

\end{enumerate}
\end{itemize}
\end{itemize}

This concludes the proof of Lemma~\ref{dtso:tso:singlestep:lemma}.
\end{proof}

\begin{lem}
\label{dtso:tso:singlestep:succ:lemma}
If $\proc\procordering\maxproc$ then
$\dconf_{r,\proc,\numberof{r}{\proc}}=\dconf_{r,\procsuccof\proc,0}$.
\end{lem}

\begin{proof}

To prove the lemma, we will prove the following properties:
\begin{enumerate}
\item
$\twostatemappingof{\dconf_{r,\proc,\numberof{r}{\proc}}}{\proc'}=
\twostatemappingof{\dconf_{r,\procsuccof\proc,0}}{\proc'}$
for all $\proc'\in\procs$,
\item
$\twobuffermappingof{\dconf_{r,\proc,\numberof{r}{\proc}}}{\proc'}=
\twobuffermappingof{\dconf_{r,\procsuccof\proc,0}}{\proc'}$
for all $\proc'\in\procs$,
\item
$\onememoryof{\dconf_{r,\proc,\numberof{r}{\proc}}}=\onememoryof{\dconf_{r,\procsuccof\proc,0}}$.
\end{enumerate}

We prove the property (1) by considering four cases:
\begin{itemize}
\item
 $\proc'=\proc$:
From the definitions of $\dconf_{r,\procsuccof\proc,\ell}$ and $\dconf_{r,\proc,\numberof{r}{\proc}}$, we have 
$\twostatemappingof{\dconf_{r,\procsuccof\proc,\ell}}\proc=
 \twostatemappingof{\dconf_{r,\proc,\numberof{r}{\proc}}}\proc$
for all $\ell:0\leq\ell \leq \numberof{r}{\procsuccof\proc}$.
In particular, we see that
$\twostatemappingof{\dconf_{r,\procsuccof\proc,0}}\proc=
\twostatemappingof{\dconf_{r,\proc,\numberof{r}{\proc}}}\proc$.

\item
$\proc'=\procsuccof\proc$: 
From the definitions of $\dconf_{r,\proc,\ell}$ and $\dconf_{r,\procsuccof\proc,0}$,
$\twostatemappingof{\dconf_{r,\proc,\ell}}{\procsuccof\proc}
=
\twostatemappingof{\dconf_{r,\procsuccof\proc,0}}{\procsuccof\proc}
$
for all $\ell:0\leq\ell \leq \numberof{r}{\proc}$.
In particular, 
we see that
$\twostatemappingof{\dconf_{r,\proc,\numberof{r}{\proc}}}{\procsuccof\proc}
=
\twostatemappingof{\dconf_{r,\procsuccof\proc,0}}{\procsuccof\proc}
$.
It follows from $\proc'=\procsuccof\proc$
that 
$\twostatemappingof{\dconf_{r,\proc,\numberof{r}{\proc}}}{\proc'}
=
\twostatemappingof{\dconf_{r,\procsuccof\proc,0}}{\proc'}
$.

\item
$\proc'\procordering\proc\procordering\procsuccof\proc$:
From the definitions of $\dconf_{r,\procsuccof\proc,\ell}$ and $\dconf_{r,\proc',\numberof{r}{\proc'}}$, we know that 
$\twostatemappingof{\dconf_{r,\procsuccof\proc,\ell}}{\proc'}=
\twostatemappingof{\dconf_{r,\proc',\numberof{r}{\proc'}}}{\proc'}
$
for all
$\ell:0\leq\ell \leq \numberof{r}{\procsuccof\proc}$.
In particular, we see that
$\twostatemappingof{\dconf_{r,\procsuccof\proc,0}}{\proc'}=
\twostatemappingof{\dconf_{r,\proc',\numberof{r}{\proc'}}}{\proc'}
$.
Also, by a similar argument, we have
$\twostatemappingof{\dconf_{r,\proc,\ell}}{\proc'}=
\twostatemappingof{\dconf_{r,\proc',\numberof{r}{\proc'}}}{\proc'}
$
for all
$\ell:0\leq \ell \leq \numberof{r}{\proc}$.
In particular, we see that
$\twostatemappingof{\dconf_{r,\proc,\numberof{r}{\proc}}}{\proc'}=
\twostatemappingof{\dconf_{r,\proc',\numberof{r}{\proc'}}}{\proc'}$.
Hence, we have
$\twostatemappingof{\dconf_{r,\procsuccof\proc,0}}{\proc'}=
\twostatemappingof{\dconf_{r,\proc,\numberof{r}{\proc}}}{\proc'}$.

\item
$\proc\procordering\procsuccof\proc\procordering\proc'$: 
From the definitions of $\dconf_{r,\procsuccof\proc,\ell}$ and
$\dconf_{r,\proc',0}$,
we can see that
$\twostatemappingof{\dconf_{r,\procsuccof\proc,\ell}}{\proc'}=
\twostatemappingof{\dconf_{r,\proc',0}}{\proc'}
$
for all
$\ell:0\leq\ell \leq \numberof{r}{\proc}$.
In particular, we see that
$\twostatemappingof{\dconf_{r,\procsuccof\proc,0}}{\proc'}=
\twostatemappingof{\dconf_{r,\proc',0}}{\proc'}
$.
Also, by a similar argument, we have
$\twostatemappingof{\dconf_{r,\proc,\ell}}{\proc'}=
\twostatemappingof{\dconf_{r,\proc',0}}{\proc'}
$
for all
$\ell:0\leq\ell \leq \numberof{r}{\proc}$.
In particular, we see that
$\twostatemappingof{\dconf_{r,\proc,\numberof{r}{\proc}}}{\proc'}=
\twostatemappingof{\dconf_{r,\proc',0}}{\proc'}$.
Hence, we have
$\twostatemappingof{\dconf_{r,\procsuccof\proc,0}}{\proc'}=
\twostatemappingof{\dconf_{r,\proc,\numberof{r}{\proc}}}{\proc'}$.
\end{itemize}

We prove the properties (2) and (3).
By a similar manner to the case of states, we can show the property (2).
Finally, to show the property (3), by the definition of $\dconf_{r,\proc,\numberof{r}{\proc}}$, it follows that 
$\onememoryof{\dconf_{r,\proc,\numberof{r}{\proc}}}=\onememoryof{\conf_{i_r}}$. 
Also, by a similar argument, we have
$\onememoryof{\dconf_{r,\procsuccof\proc,0}}=\onememoryof{\conf_{i_r}}$.
Hence, we have
$\onememoryof{\dconf_{r,\proc,\numberof{r}{\proc}}}=\onememoryof{\dconf_{r,\procsuccof\proc,0}}$.

This concludes the proof of Lemma~\ref{dtso:tso:singlestep:succ:lemma}.
\end{proof}

\begin{lem}
\label{dtso:tso:next:phase:lemma}
If $r<k$ and  $\transition_{i_{r+1}} \in \transitions_{\proc^u}$ such that $\transition_{i_{r+1}}$ is of the form $\tuple{\state,\arw(\xvar,\data,\data'),\state'}$,
 then
$\dconf_{r,\maxproc,\numberof{r}{\maxproc}}\tsomovesto{\transition_{{i}_{r+1}}}
\dconf_{r+1,\minproc,0}$.
\end{lem}

\begin{proof}

To prove the lemma, we will prove the following properties:
\begin{enumerate}
\item
$\twostatemappingof{\dconf_{r,\maxproc,\numberof{r}{\maxproc}}}{\proc}=
\twostatemappingof{\dconf_{r+1,\minproc,0}}{\proc}$ for all $\proc \neq \proc^u$,
\item
$\twobuffermappingof{\dconf_{r,\maxproc,\numberof{r}{\maxproc}}}\proc=
\twobuffermappingof{\dconf_{r+1,\minproc,0}}\proc$
for all $\proc \neq \proc^{u}$,
\item
$\twostatemappingof{\dconf_{r,\maxproc,\numberof{r}{\maxproc}}}{\proc^u}=\state$, and $\twostatemappingof{\dconf_{r+1,\minproc,0}}{\proc^u}=\state'$,
\item
 $\twobuffermappingof{\dconf_{r,\maxproc,\numberof{r}{\maxproc}}}{\proc^u}=
\twobuffermappingof{\dconf_{r+1,\minproc,0}}{\proc^u}=\epsilon$,
\item
$\twomemoryof{\dconf_{r,\maxproc,\numberof{r}{\maxproc}}}\xvar=\data$ and
$\twomemoryof{\dconf_{r+1,\minproc,0}}\xvar=\data'$.
\end{enumerate}

We show the property (1).
Let 
$\proc\in\procs\setminus \set{\proc^u}$.
From the definition of $\dconf_{r,\maxproc,\numberof{r}{\maxproc}}$ and $\dconf_{r,\proc,\numberof{r}{\proc}}$,
$\twostatemappingof{\dconf_{r,\maxproc,\numberof{r}{\maxproc}}}{\proc}=
\twostatemappingof{\dconf_{r,\proc,\numberof{r}{\proc}}}{\proc}=
\twostatemappingof{\conf_{\alpha(r,\proc,\numberof{r}{\proc})}}{\proc}$
and 
$\twostatemappingof{\dconf_{r+1,\minproc,0}}{\proc}=
\twostatemappingof{\dconf_{r+1,\proc,0}}\proc=
\twostatemappingof{\conf_{\alpha(r+1,\proc,0)}}\proc$.
From the definition of $\alpha$, it follows that
$\transition_j\not\in\transitions_\proc$ for all
$j:\alpha(r,\proc,\numberof{r}{\proc})\leq j<\alpha(r+1,\proc,0)$. This implies that $\twostatemappingof{\conf_{\alpha(r+1,\proc,0)-1}}\proc= \twostatemappingof{\conf_{\alpha(r,\proc,\numberof{r}{\proc})}}\proc$. Now we have two cases: 

\begin{itemize}
\item  $\{{j}\,|\,{{\sf view}_\proc(\conf_j)=r+1}\}= \emptyset$: We see that $\alpha(r+1,\proc,0)=\alpha(r,\proc,\numberof{r}{\proc})$,  and hence that $\twostatemappingof{\dconf_{r,\maxproc,\numberof{r}{\maxproc}}}\proc=
\twostatemappingof{\dconf_{r+1,\minproc,0}}\proc$. 
\item  $\{{j}\,|\,{{\sf view}_\proc(\conf_j)=r+1}\} \neq \emptyset$:   Since ${\sf view}_\proc(\conf_{\alpha(r+1,\proc,0)-1})=r$, we can show that $\transition_{\alpha(r+1,\proc,0)} \notin \transitions_\proc$. This is done by contradiction as follows. In fact if $\transition_{\alpha(r+1,\proc,0)} \in \transitions_\proc$, then  it is either a write transition or an atomic read-write transition. This implies that in both cases that $\twobuffermappingof{\conf_{\alpha(r+1,\proc,0)-1}}\proc=\epsilon$ and that ${\sf view}_\proc(\conf_{\alpha(r+1,\proc,0)})=\alpha(r+1,\proc,0)$. Hence, we have $\alpha(r+1,\proc,0)=r+1$, and this leads to a contradiction since $\transition_{i_{r+1}} \in \transitions_{\proc^u}$. Thus, we have $\twostatemappingof{\dconf_{r,\maxproc,\numberof{r}{\maxproc}}}\proc=
\twostatemappingof{\dconf_{r+1,\minproc,0}}\proc$. 
\end{itemize}

In a similar manner to the case of states, we can show the property (2).
Now we show the properties (3) and (4).
Using a similar reasoning as for the process $\proc$,
we know that $\twostatemappingof{\conf_{\alpha(r+1,\proc^u,0)-1}}{\proc^u}= \twostatemappingof{\conf_{\alpha(r,\proc^u,\numberof{r}{\proc^u})}}{\proc^u}$. From the definition of $\comp_{\sf DTSO}$, it follows that
$\conf_{\alpha(r+1,\proc^u,0)-1}
\dtsomovesto{\transition_{\alpha(r+1,\proc^u,0)}} \conf_{\alpha(r+1,\proc^u,0)}$. 
Furthermore, since ${\sf view}_\proc(\conf_{\alpha(r+1,\proc^u,0)})=r+1$ and ${\sf view}_\proc(\conf_{\alpha(r+1,\proc^u,0)-1})<r+1$, we know that   $\transition_{\alpha(r+1,\proc^u,0)}= t_{i_{r+1}}$. This implies that $\twobuffermappingof{\conf_{\alpha(r+1,\proc^u,0)-1}}{\proc^u}=\twobuffermappingof{\conf_{\alpha(r+1,\proc^u,0)}}{\proc^u}=\epsilon$ and  $\twostatemappingof{\conf_{\alpha(r+1,\proc^u,0)-1}}{\proc^u}=\state$ and $\twostatemappingof{\conf_{\alpha(r+1,\proc^u,0)}}{\proc^u}=\state'$.
Now since 
$$\dtsototso{\twobuffermappingof{\conf_{\alpha(r,\proc^u,\numberof{r}{\proc^u})}}{\proc^u}}=\dtsototso{\twobuffermappingof{\conf_{\alpha(r+1,\proc^u,0)-1}}{\proc^u}},$$
we see that $$\twobuffermappingof{\dconf_{r,\maxproc,\numberof{r}{\maxproc}}}{\proc^u}=
\twobuffermappingof{\dconf_{r+1,\minproc,0}}{\proc^u}=\epsilon$$ and that $\twostatemappingof{\dconf_{r,\maxproc,\numberof{r}{\maxproc}}}{\proc^u}=\state$ and $\twostatemappingof{\dconf_{r+1,\minproc,0}}{\proc^u}=\state'$.
This concludes the properties (3) and (4).

We show the property (5). From the definition of $\comp_{\sf DTSO}$, it follows that
 $\onememoryof{\conf_{i_{r+1}}}=\onememoryof{\conf_{i_{r}}}\update{\xvar}{\data'}$ with $\twomemoryof{\conf_{i_r}}\xvar=\data$. 
Then from the definitions of $\dconf_{r,\maxproc,\numberof{r}{\maxproc}}$ and $\dconf_{r+1,\minproc,0}$,
we have the property (5).
 
 This concludes the proof of Lemma~\ref{dtso:tso:next:phase:lemma}.
 \end{proof}

\begin{lem}
\label{dtso:tso:next:phase:lemma2}
If $r<k$ and  $\transition_{i_{r+1}} \in \transitions_{\proc^u}$ such that $\transition_{i_{r+1}}$ is of the form $\tuple{\state,\wop(\xvar,\data),\state'}$,
 then
$\dconf_{r,\maxproc,\numberof{r}{\maxproc}}\tsomovesto{*}
\dconf_{r+1,\minproc,0}$.
\end{lem}
\begin{proof}

To prove the lemma, we will prove the following properties:
\begin{enumerate}
\item
$\twostatemappingof{\dconf_{r,\maxproc,\numberof{r}{\maxproc}}}\proc=
\twostatemappingof{\dconf_{r+1,\minproc,0}}\proc$ for all $\proc \neq \proc^u$,
\item
$\twobuffermappingof{\dconf_{r,\maxproc,\numberof{r}{\maxproc}}}\proc=
\twobuffermappingof{\dconf_{r+1,\minproc,0}}\proc$
for all $\proc \neq \proc^{u}$,
\item
The contents of buffers
$\twostatemappingof{\dconf_{r,\maxproc,\numberof{r}{\maxproc}}}{\proc^u}$ and $\twostatemappingof{\dconf_{r+1,\minproc,0}}{\proc^u}$
 are compatible, i.e.\ with the properties (1)--(2), the property (3) allows
$\dconf_{r,\maxproc,\numberof{r}{\maxproc}}\tsomovesto{*}
\dconf_{r+1,\minproc,0}$.
\end{enumerate}

We show the property (1).
Let  $\proc\in\procs\setminus\set{\proc^u}$.
From the definitions of $\dconf_{r,\maxproc,\numberof{r}{\maxproc}}$ and $\dconf_{r,\proc,\numberof{r}{\proc}}$,
$\twostatemappingof{\dconf_{r,\maxproc,\numberof{r}{\maxproc}}}\proc=
\twostatemappingof{\dconf_{r,\proc,\numberof{r}{\proc}}}\proc=
\twostatemappingof{\conf_{\alpha(r,\proc,\numberof{r}{\proc})}}\proc$
and that
$\twostatemappingof{\dconf_{r+1,\minproc,0}}\proc=
\twostatemappingof{\dconf_{r+1,\proc,0}}\proc=
\twostatemappingof{\conf_{\alpha(r+1,\proc,0)}}\proc$.
From the definition of $\alpha$, it follows that
$\transition_j\not\in\transitions_\proc$ for all
$j:\alpha(r,\proc,\numberof{r}{\proc})\leq j<\alpha(r+1,\proc,0)$. This implies that $\twostatemappingof{\conf_{\alpha(r+1,\proc,0)-1}}\proc= \twostatemappingof{\conf_{\alpha(r,\proc,\numberof{r}{\proc})}}\proc$. Now we have two cases: 

\begin{itemize}
\item   $\{{j}\,|\,{{\sf view}_\proc(\conf_j)=r+1}\}= \emptyset$: We see that $\alpha(r+1,\proc,0)=\alpha(r,\proc,\numberof{r}{\proc})$,  and hence that $\twostatemappingof{\dconf_{r,\maxproc,\numberof{r}{\maxproc}}}\proc=
\twostatemappingof{\dconf_{r+1,\minproc,0}}\proc$. 
\item  $\{{j}\,|\,{{\sf view}_\proc(\conf_j)=r+1}\} \neq \emptyset$: Since ${\sf view}_\proc(\conf_{\alpha(r+1,\proc,0)-1})=r$,  we can show that $\transition_{\alpha(r+1,\proc,0)} \notin \transitions_\proc$ . This is done by contradiction as follows. In fact if $\transition_{\alpha(r+1,\proc,0)} \in \transitions_\proc$, then  it is either a write transition or an atomic read-write transition. This implies that in both cases that $\twobuffermappingof{\conf_{\alpha(r+1,\proc,0)-1}}\proc=\epsilon$ and that ${\sf view}_\proc(\conf_{\alpha(r+1,\proc,0)})=\alpha(r+1,\proc,0)$. Hence, we have $\alpha(r+1,\proc,0)=r+1$, and this leads to a contradiction since $\transition_{i_{r+1}} \in \transitions_{\proc^u}$. Thus, we have $\twostatemappingof{\dconf_{r,\maxproc,\numberof{r}{\maxproc}}}\proc=
\twostatemappingof{\dconf_{r+1,\minproc,0}}\proc$. 
\end{itemize}

In a similar manner to the case of states, we can show the property (2).
Now we show the property (3).
Using a similar reasoning as for the process $\proc$, we know that $\twostatemappingof{\conf_{\alpha(r+1,\proc^u,0)-1}}{\proc^u}= \twostatemappingof{\conf_{\alpha(r,\proc^u,\numberof{r}{\proc^u})}}{\proc^u}$. From the definition of $\comp_{\sf DTSO}$, it follows that
$\conf_{\alpha(r+1,\proc^u,0)-1}
\dtsomovesto{\transition_{\alpha(r+1,\proc^u,0)}} \conf_{\alpha(r+1,\proc^u,0)}$. 
Furthermore, from the fact that ${\sf view}_\proc(\conf_{\alpha(r+1,\proc^u,0)})=r+1$ and ${\sf view}_\proc(\conf_{\alpha(r+1,\proc^u,0)-1})<r+1$, we  have two cases to consider:

\begin{itemize}
\item   $\twobuffermappingof{\conf_{\alpha(r+1,\proc^u,0)-1}}{\proc^u}=\epsilon$: It follows from the conditions 
for ${\sf view}_\proc(\conf_{\alpha(r+1,\proc^u,0)})$ and ${\sf view}_\proc(\conf_{\alpha(r+1,\proc^u,0)-1})$
 that  $\transition_{\alpha(r+1,\proc^u,0)}= t_{i_{r+1}}$, 
 $\twobuffermappingof{\conf_{\alpha(r+1,\proc^u,0)}}{\proc^u}=(\xvar,\data,{\textsc{own}})$, and that $\twostatemappingof{\conf_{\alpha(r+1,\proc^u,0)-1}}{\proc^u}=\state$ and $\twostatemappingof{\conf_{\alpha(r+1,\proc^u,0)}}{\proc^u}=\state'$. From
 \begin{align*}
 \dtsototso{\twobuffermappingof{\conf_{\alpha(r,\proc^u,\numberof{r}{\proc^u})}}{\proc^u}}=\dtsototso{\twobuffermappingof{\conf_{\alpha(r+1,\proc^u,0)-1}}{\proc^u}}
 \end{align*}
  we have $\twobuffermappingof{\dconf_{r,\maxproc,\numberof{r}{\maxproc}}}{\proc^u}=
\twobuffermappingof{\dconf_{r+1,\minproc,0}}{\proc^u}=\epsilon$.
Moreover, we have $\twostatemappingof{\dconf_{r,\maxproc,\numberof{r}{\maxproc}}}{\proc^u}=\state$, and $\twostatemappingof{\dconf_{r+1,\minproc,0}}{\proc^u}=\state'$.
Then, it is easy to see that $\onememoryof{\conf_{i_{r+1}}}=\onememoryof{\conf_{i_{r}}}\update{\xvar}{\data}$. Hence, we have $\dconf_{r,\maxproc,\numberof{r}{\maxproc}}\tsomovesto{\transition_{{i}_{r+1}}} \dconf' \tsomovesto{{\updateop_{\proc^{u}}}}
\dconf_{r+1,\minproc,0}$ for some configuration $\dconf'$.

\item $\twobuffermappingof{\conf_{\alpha(r+1,\proc^u,0)-1}}{\proc^u} \neq\epsilon$:  It follows from the conditions 
for ${\sf view}_\proc(\conf_{\alpha(r+1,\proc^u,0)})$ and ${\sf view}_\proc(\conf_{\alpha(r+1,\proc^u,0)-1})$
 that $\transition_{\alpha(r+1,\proc^u,0)}$ is a delete transition of the process $\proc^u$. As a consequence, $\twobuffermappingof{\conf_{\alpha(r+1,\proc^u,0)-1}}{\proc^u}= w \cdot (\xvar,\data,{\textsc{own}}) \cdot m$  and $\twobuffermappingof{\conf_{\alpha(r+1,\proc^u,0)}}{\proc^u}=w \cdot (\xvar,\data,{\textsc{own}})$. Hence, we see  that $$\dtsototso{\twobuffermappingof{\conf_{\alpha(r+1,\proc^u,0)}}{\proc^u}}= \dtsototso{\twobuffermappingof{\conf_{\alpha(r+1,\proc^u,0)-1}}{\proc^u}} \cdot (\xvar,\data)$$ and therefore $
\twobuffermappingof{\dconf_{r+1,\minproc,0}}{\proc^u}= \twobuffermappingof{\dconf_{r,\maxproc,\numberof{r}{\maxproc}}}{\proc^u} (\xvar,\data)$. Furthermore, we have $\twostatemappingof{\conf_{\alpha(r+1,\proc^u,0)}}{\proc^u}=\twostatemappingof{\conf_{\alpha(r+1,\proc^u,0)-1}}{\proc^u}$ and this implies that 
$\twostatemappingof{\dconf_{r,\maxproc,\numberof{r}{\maxproc}}}{\proc^u}=\twostatemappingof{\dconf_{r+1,\minproc,0}}{\proc^u}$.
Then, it is easy to see that $\onememoryof{\conf_{i_{r+1}}}=\onememoryof{\conf_{i_{r}}}\update{\xvar}{\data}$. Hence, we have $\dconf_{r,\maxproc,\numberof{r}{\maxproc}} \tsomovesto{{\updateop_{\proc^{u}}}}
\dconf_{r+1,\minproc,0}$.
\end{itemize}

This concludes the proof of Lemma~\ref{dtso:tso:next:phase:lemma2}.
\end{proof}

The following lemma shows that the  TSO-computation $\tsocomp$ starts from the initial TSO-configuration.
\begin{lem}
\label{init:tso:comp:lemma}
$\dconf_{0,\minproc,0}$ is the initial TSO-configuration.
\end{lem}
\begin{proof}
Let us take any $\proc\in\procs$.
By the definitions of $\dconf_{0,\minproc,0}$, $\dconf_{0,\proc,0}$, and $\alpha(0,\proc,0)$, it follows that $\twostatemappingof{\dconf_{0,\minproc,0}}\proc=
\twostatemappingof{\dconf_{0,\proc,0}}\proc=
\twostatemappingof{\conf_{\alpha(0,\proc,0)}}\proc=
\twostatemappingof{\conf_0}\proc=
\initstatemapping
$.
Also,
$\twobuffermappingof{\dconf_{0,\minproc,0}}\proc=
\twobuffermappingof{\dconf_{0,\proc,0}}\proc=
\dtsototso{\twobuffermappingof{\conf_{0}}\proc}=
\emptyword
$.
Finally, we have $\onememoryof{\dconf_{0,\minproc,0}}=\onememoryof{\conf_{i_0}}=\onememoryof{\conf_0}$.
The result follows immediately for the 
definition of the initial TSO-configuration.
This concludes the proof of Lemma~\ref{init:tso:comp:lemma}.
\end{proof}

The following lemma shows that the target of the TSO-computation $\tsocomp$ has the same local process states as
the target $\conf_n$ of the DTSO-computation $\sbcomp$.
\begin{lem}
\label{final:tso:comp:lemma}
$\onestatemappingof{\dconf_{k,\maxproc,\numberof{k}{\maxproc}}}=
\onestatemappingof{\conf_n}$.
\end{lem}
\begin{proof}
Let us take any $\proc\in\procs$.
By the definitions of $\dconf_{k,\maxproc,\numberof{k}{\maxproc}}$ and $\dconf_{k,\proc,\numberof{k}{\proc}}$, it follows that
$
\twostatemappingof{\dconf_{k,\maxproc,\numberof{k}{\maxproc}}}\proc=
\twostatemappingof{\dconf_{k,\proc,\numberof{k}{\proc}}}\proc=
\twostatemappingof{\conf_{\alpha(k,\proc,\numberof{k}{\proc})}}\proc
$.
By definition of $\alpha(k,\proc,\numberof{k}{\proc})$, we know that
$\transition_j\not\in\transitions_\proc$ for all
$j:\alpha(k,\proc,\numberof{k}{\proc})<j \leq n$.
Therefore, we have
$\twostatemappingof{\conf_j}\proc=
\twostatemappingof{\conf_n}\proc$ for all
$j:\alpha(k,\proc,\numberof{k}{\proc})\leq j<n$.
In particular, we have
$\twostatemappingof{\conf_{\alpha(k,\proc,\numberof{k}{\proc})}}\proc=
\twostatemappingof{\conf_n}\proc$.
Hence, we have
$\twostatemappingof{\dconf_{k,\maxproc,\numberof{r}{\maxproc}}}\proc=
\twostatemappingof{\conf_n}\proc$.
This concludes the proof of Lemma~\ref{final:tso:comp:lemma}.
\end{proof}


\subsection
{\bf From TSO to Dual TSO}
We show the {\em only if direction} of Theorem \ref{DTSO:TSO:equivalence:theorem}. Consider a TSO-computation 
$$\tsocomp=\conf_0\tsomovesto{\transition_1} \conf_1\tsomovesto{\transition_2} \conf_2\;\cdots\;
\tsomovesto{\transition_{n-1}}\conf_{n-1}\tsomovesto{\transition_{n}}\conf_n.$$
 where $c_0=\initconf$ and $\conf_i$ is of the form $\tuple{\statemapping_{i},\buffermapping_i,\mem_i}$ for all $i : 1 \leq i \leq n$ with $\statemapping_n=\statemapping_{target}$.
 In the following,  we will derive a  DTSO-computation $\comp_{\sf DTSO}$
such that $\onestatemappingof{\targetof{\comp_{\sf DTSO}}}=\onestatemappingof{\conf_n}$, i.e.\
the runs $\tsocomp$ and $\comp_{\sf DTSO}$ reach  the same set of local  states at the end of the runs.

Similar to the previous case, we  will first define some functions that we will use in the construction 
of the computation $\comp_{\sf DTSO}$.
Then, we define a sequence of DTSO-configurations that appear in $\comp_{\sf DTSO}$.
Finally, we show that the  DTSO-computation $\comp_{\sf DTSO}$  exists. In particular, 
the target configuration  $\targetof{\comp_{\sf DTSO}}$ has the same local states as the target $\conf_n$ of the TSO-computation $\comp_{\sf TSO}$.

For every $\proc \in \procs$,  let $\transitions_{\proc}^{{\sf w, arw}} \subseteq \transitions_{\proc}$ (resp. $\transitions_{\proc}^{{\sf u, arw}} \subseteq \transitions_{\proc} \cup \{\updateop_\proc\}$) be the set of write (resp. update) and atomic read-write transitions that can be performed by  process $\proc$. 
Let $\transitions^{\sf r}_{\proc}$ be the set of read transitions that can be performed by  $\proc$.

Let $I=i_1 \ldots i_m$ be the maximal sequence of indices  such that $1\leq i_1< i_2 <\cdots<i_m \leq n$ and for every $j : 1 \leq j \leq m$, we have $\transition_{i_j}$ is an {\it update} transition or an {atomic read-write transition} (i.e., $\transition_{i_j} \in {\bigcup_{\proc \in \procs } \, \transitions_\proc^{\sf u, arw}}$).
In the following, we assume that $i_0=0$.
Let $I_{\proc}$  be the maximal subsequence of $I$ such that
all transitions with indices in $I_{\proc}$ belong to process $\proc$.

Let $I'=i'_1 \ldots i'_m$ be the maximal sequence of indices  such that $1\leq i'_1< i'_2 <\cdots<i'_m \leq n$ and for every $j : 1 \leq j \leq m$, we have $\transition_{i'_j}$ is a {\it write} transition or an {atomic read-write transition} (i.e., $\transition_{i'_j} \in {\bigcup_{\proc \in \procs } \, \transitions_\proc^{\sf w, arw}}$).
Let $I'_{\proc}$ be the maximal subsequence of $I'$ such that
all transitions with indices in $I'_{\proc}$ belong to process $\proc$.
 Observe that $|I_{\proc}|=|I'_{\proc}|$.
 
 For every  $j: 1\leq j \leq m$, let $\processof{j}$ be  the process that has the 
  update  or atomic read-write transition $t_{i_j}$ where $i_j \in I$.
 We define $\matchingof{i_j}$ to be the index of the write (resp. atomic read-write) transition $t_{\matchingof{i_j}}$
 that corresponds to  the update (resp. atomic read-write) transition  $t_{i_j}$.
 Formally,
 $\matchingof{i_j}$:=$l$ where $\exists k: 1 \leq k \leq |I_{\proc}|, I_\proc(k)=i_j$, $I'_\proc(k) = l$ and $1 \leq l \leq n$.
Observe that  if $t_{i_j}$ is an atomic  read-write operation, then $\matchingof{i_j}=i_j$.

\begin{exa}
\label{tso:dualtso:construction:exampl1}
We give an example of how to calculate 
the function {\tt match}
for a TSO-computation.
Let us consider the following TSO-computation
$$
\tsocomp = \conf_0\tsomovesto{\transition_1} \conf_1 \tsomovesto{\transition_2} \conf_2  
\tsomovesto{\transition_3} \conf_3  
$$
containing only transitions of a process $\proc$
with two variables $\xvar$ and $\yvar$
where
$\conf_i=\tuple{\statemapping_i,\buffermapping_i,\mem_i}$ for all $i:0\leq i \leq n=3$
such that:
\begin{align*} 
\statemapping_0(\proc) &= \state_0, & \buffermapping_0(\proc) &= \emptyword, &  \mem_0(\xvar) &= 0,   \mem_0(\yvar) = 0, &\transition_1&=\tuple{\state_0,\wop(\xvar,1),\state_1},  \\
\statemapping_1(\proc) &= \state_1, & \buffermapping_1(\proc) &= (\xvar,1), &  \mem_1(\xvar) &= 0,   \mem_1(\yvar) = 0, &\transition_2&=\updateop_\proc, \\
\statemapping_2(\proc) &= \state_1, & \buffermapping_2(\proc) &= \emptyword, &  \mem_2(\xvar) &= 1,   \mem_2(\yvar) = 0, &\transition_3&=\tuple{\state_1,\rop(\yvar,0),\state_2},\\
\statemapping_3(\proc) &= \state_2, & \buffermapping_3(\proc) &= \emptyword, &  \mem_3(\xvar) &= 1,   \mem_3(\yvar) = 0.
\end{align*}

Following the above definitions of $I$ and $I'$, $I=i_1=2$ (hence, $m=1$) is the maximal sequence of indices of all update or atomic read-write transitions in $\tsocomp$.
In a similar way, $I'=i'_1=1$ is the maximal sequence of indices of all write or atomic read-write transitions in $\tsocomp$.
We note that
$t_{i_1}=t_2$ is an update transition,
and 
$t'_{i'_1}=t_1$ is a write transition.
Since the TSO-computation contains only transition of the process $\proc$,
it follows that
$I=I_\proc$ and $I'=I'_\proc$.
Following the above definition of {\tt match}, with $m=1$ and $n=3$, we have
$\matchingof{i_1}=\matchingof{2}=1$.
\myend
\end{exa}

For every $j : 1 \leq j \leq n$ such that $\transition_j  \in \transitions_\proc^{\sf r}$ is a read transition 
of  process $\proc$, we define  ${\sf fromMem}(\transition_j)$ as a predicate such that  ${\sf fromMem}(\transition_j)$ holds if and only if $(\xvar,\data') \notin\onebuffermappingof{\conf_{j-1}}$ for all $\data'\in\dataset$.

For every $j : 1 \leq j \leq n$ and  $\proc\in\procs$, we define the function ${\sf label}_\proc(j)$ as follows:
\begin{enumerate}
\item
${\sf label}_{\proc}(j):=(\xvar,\data)$ if $\transition_j \in \transitions_\proc^r$ is of the form $(\state, \rop(\xvar,\data),\state')$ and   ${\sf fromMem}(\transition_j)$ holds.  
\item
${\sf label}_{\proc}(j):=(\xvar,\data,{\textsc{own}})$ if $\transition_j =\updateop_\proc$  and   $\matchingof{j}=l$ with $\transition_l$  of the form $ (\state, \wop(\xvar,\data),\state')$.
\item
${\sf label}_{\proc}(j):=\epsilon$ otherwise. 
\end{enumerate}
Given a sequence  $\ell_1 \cdots \ell_k$ with $k\geq1$ and $1\leq \ell_i \leq n$ for all $i: 1\leq i \leq k$, we define ${\sf label}_\proc(\ell_1 \cdots \ell_k ):= {\sf label}_\proc(\ell_1)  \cdots  {\sf label}_\proc(\ell_{k-1}) \app {\sf label}_\proc(\ell_k)$. 
Let ${\sf label}^{\sf rev}_\proc(\ell_1 \cdots \ell_k )$ with $k\geq1$
and $1\leq \ell_i \leq n$ for all $i: 1\leq i \leq k$
 be the reversed string of 
${\sf label}_\proc(\ell_1 \cdots \ell_k )$, i.e.\
${\sf label}^{\sf rev}_\proc(\ell_1 \cdots \ell_k ):= {\sf label}_\proc(\ell_k) \app  {\sf label}_\proc(\ell_{k-1})  \cdots {\sf label}_\proc(\ell_1)$.

\begin{exa}
\label{tso:dualtso:construction:exampl2}
In the following, we give an example of how to calculate the functions {\tt fromMem} and {\tt label} for the TSO-computation $\tsocomp$ 
given in Example~\ref{tso:dualtso:construction:exampl1}.
We recall that $n=3$ 
and the function {\tt match} is given  in Example~\ref{tso:dualtso:construction:exampl1}.
We also note that $\transition_3$ is the only read transition in $\tsocomp$.
Following the above definition of {\tt fromMem}, we have that ${\tt fromMem}(\transition_3)$ holds.
Then following the  definition of {\tt match}, for every $j: 1\leq j \leq n=3$, we define the function ${\tt label}_\proc(j)$ as follows:
\begin{align*} 
{\tt label}_\proc(1)&= \emptyword, & 
{\tt label}_\proc(2)&= (\xvar,1,{\it own}), & 
{\tt label}_\proc(3)&= (\yvar,0).  \tag*{$\triangle$}
\end{align*}
\end{exa}

Below we show how to simulate all transitions of the TSO-computation $\comp_{\sf TSO}$ by a set of corresponding transitions in the DTSO-computation $\comp_{\sf DTSO}$.  
 The idea is  to divide the DTSO-computation to $m+1$ phases. 
 For 
$0\leq r< m$, each phase $r$  will end at the configuration $\dconf_{r+1}$ by the simulation of the transition $\transition_{\matchingof{i_{r+1}}}$ in $\comp_{\sf TSO}$.
Moreover, in phase $r: 0\leq r < m$,
we call the process $\processof{r+1}$ as the \emph{active } process, and other processes as the \emph{inactive } ones. We  execute only the DTSO-transitions of the active process $\proc=\processof{r+1}$ in its active phases. For other processes $\proc'\neq\proc$, we only change the content of their buffers in the active phases of $\proc$.
In the final phase  $r=m$, 
all processes will be considered to be active because the index ${i_{m+1}}$ is not defined in the definition of the sequence $I$.
The DTSO-computation  $\comp_{\sf DTSO}$ will end at the configuration $\dconf_{m+1}$.

%
For every $r : -1 \leq r < m$ and  $\proc\in\procs$, 
we define the function  $\posof{r}{\proc}$
 in an inductive way on $r$:
 \begin{enumerate}
 \item
$\posof{-1}{\proc}:=0$ for all $\proc\in\procs$.
\item
$\posof{r}{\proc}:=\posof{r-1}{\proc}$ for all $\proc\neq\processof{r+1}$ and $0\leq r < m$.
\item
$\posof{r}{\proc}:=\matchingof{i_{r+1}}$ for $\proc=\processof{r+1}$ and $0\leq r < m$.
\end{enumerate}
In other words, 
the function  $\posof{r}{\proc}$ is the index of the last simulated transition  by process $\proc$   at the end of phase $r$ in the computation  $\comp_{\sf TSO}$. 
Moreover, we use  $\posof{-1}{\proc}$ to be the index of the starting transition of process $\proc$ before phase $0$.

\begin{exa}
\label{tso:dualtso:construction:exampl3}
In the following, we give an example of how to calculate the function {\tt pos} for the TSO-computation $\tsocomp$ 
given in Example~\ref{tso:dualtso:construction:exampl1}.
We recall that
$m=1$ and 
$\tsocomp$ contains only transitions of the process $\proc$.
We also recall that the function {\tt match} is given in Example~\ref{tso:dualtso:construction:exampl2}.
Following the above definition of {\tt pos}, for every $r: -1\leq r < m=1$, we define the function 
$\posof{r}{\proc}$ as follows:
\begin{align*} 
\posof{-1}{\proc} &= 0, & 
\posof{0}{\proc} &= 1.  \tag*{$\triangle$}
\end{align*}
\end{exa}

Let $\dconf_0=\dinitconf=\tuple{\initstatemapping,\initbuffermapping,\initmem}$. We define 
the sequence of DTSO-configurations $\dconf_1,\ldots,\dconf_m,\dconf_{m+1}$   by defining their local states, buffer contents, and memory states as follows: 
\begin{enumerate}
\item For every configuration $\dconf_{r+1}$ where $ 0 \leq r < m$: 
\begin{itemize}
\item $\twostatemappingof{\dconf_{r+1}}\proc:=\twostatemappingof{\conf_{\posof{r}{\proc}}}\proc$,
\item $\onememoryof{\dconf_{r+1}}:=\onememoryof{\conf_{i_{r+1}}}$,
\item $\twobuffermappingof{\dconf_{r+1}}\proc:={\sf label}^{\sf rev}_\proc({\posof{r}{\proc}+1} \cdots {i_{r+1}} )$.
\end{itemize}
\item
For the final configuration $\dconf_{m+1}$:
\begin{itemize}
\item $\twostatemappingof{\dconf_{m+1}}\proc:=\twostatemappingof{\conf_n}\proc$,
\item $\onememoryof{\dconf_{m+1}}:=\onememoryof{\conf_n}$,
\item $\twobuffermappingof{\dconf_{m+1}}\proc:=\emptyword$.
\end{itemize}
\end{enumerate}
\begin{exa}
\label{tso:dualtso:construction:exampl4}
In the following, we give an example of how to calculate the sequence of configurations $\dconf_1,\ldots,\dconf_m,\dconf_{m+1}$
that will appear in the constructed DTSO-computation $\comp_{\sf DTSO}$ 
from the TSO-computation  $\tsocomp$ given in Figure~\ref{tso:dualtso:construction:exampl1}. 
We recall that $m=1$, $n=3$, and the TSO-computation  $\tsocomp$ contains only transitions of the process $\proc$.
We also recall that 
the functions $\tt label$ and {\tt pos} are given in Example~\ref{tso:dualtso:construction:exampl2}
and Example~\ref{tso:dualtso:construction:exampl3}, respectively. 

The DTSO-computation $\comp_{\sf DTSO}$ will consist of $m+1=2$ phases, referred as the phase $0$ and  the phase $1$. 
For each $r: 0 \leq r \leq m+1=2$, we define the DTSO-configuration $\dconf_r=(\statemapping'_r,\buffermapping'_r,\mem'_r)$ based on the TSO-configurations that are appearing in $\tsocomp$ as follows:
\begin{align*} 
\dconf_{0} &: & \statemapping'_0(\proc) &= \state_0, & \buffermapping'_0(\proc) &= \emptyword, &  \mem'_0(\xvar) &= 0,   \mem'_0(\yvar) = 0, \\
\dconf_{1} &: & \statemapping'_1(\proc) &= \state_1, & \buffermapping'_1(\proc) &= (\xvar,1,{\it own}), &  \mem'_1(\xvar) &= 1,   \mem'_1(\yvar) = 0, \\
\dconf_{2} &: & \statemapping'_2(\proc) &= \state_2, & \buffermapping'_1(\proc) &= \emptyword, &  \mem'_2(\xvar) &= 1,   \mem'_2(\yvar) = 0. 
\end{align*}

Finally, we construct the DTSO-computation as follows:
$$
\comp_{\sf DTSO} = \dconf_0\dtsomovesto{\transition'_1} \dconf_1 \dtsomovesto{\transition'_2} \dconf_{12}
\dtsomovesto{\transition'_3} \dconf_{13} 
\dtsomovesto{\transition'_4} \dconf_{14} 
\dtsomovesto{\transition'_5} \dconf_2
$$  
where
 $\dconf_{12}=(\statemapping'_{12},\buffermapping'_{12},\mem_{12})$,
 $\dconf_{13}=(\statemapping_{13},\buffermapping'_{13},\mem_{13})$,
 $\dconf_{14}=(\statemapping'_{14},\buffermapping'_{14},\mem_{14})$,
 $\transition'_1=(\state_0,\wop(\xvar,1),\state_1)$,
  $\transition'_2=\propagate_\proc^{\yvar}$,
  $\transition'_3=\delete_\proc$
   $\transition'_4=(\state_1,\rop(\yvar,0),\state_2)$, 
   $\transition'_5=\delete_\proc$,
   and:
\begin{align*} 
\dconf_{12} &: & \statemapping'_{12}(\proc) &= \state_1, & \buffermapping'_{12}(\proc) &= (\yvar,0)\app(\xvar,1,{\it own}), &  \mem'_{12}(\xvar) &= 1,   \mem'_{12}(\yvar) = 0, \\
\dconf_{13} &: & \statemapping'_{13}(\proc) &= \state_1, & \buffermapping'_{13}(\proc) &= (\yvar,0), &  \mem'_{13}(\xvar) &= 1,   \mem'_{13}(\yvar) = 0, \\
\dconf_{14} &: & \statemapping'_{14}(\proc) &= \state_2, & \buffermapping'_{14}(\proc) &= (\yvar,0), &  \mem'_{14}(\xvar) &= 1,   \mem'_{14}(\yvar) = 0.
\end{align*}
Since there is only one update transition in   $\comp_{\sf DTSO}$ and $\tsocomp$, it is easy to see that $\tsocomp$ has the same sequence of memory updates as $\comp_{\sf DTSO}$.
It is also easy to see that 
 $\dconf_0=\dinitconf$
 and
 $\dconf_3=\tuple{\onestatemappingof{\conf_3},\buffermapping,\onememoryof{\conf_3}}$
 where $\buffermapping(\proc):=\emptyword$.
Therefore $\comp_{\sf DTSO}$ is a witness of the construction.
\myend
\end{exa}

Lemma~\ref{proof:if-direction:full}
shows the existence of a DTSO-computation $\comp_{\sf DTSO}$ that starts from the initial
TSO-configuration and whose target has the same local state definitions
as the target $\conf_n$ of the TSO-computation $\tsocomp$.
The only if direction of  Theorem~\ref{DTSO:TSO:equivalence:theorem} will follow
directly from Lemma~\ref{proof:if-direction:full}. 
This concludes the proof of the only if direction of  Theorem~\ref{DTSO:TSO:equivalence:theorem}.

\begin{lem}
\label{proof:if-direction:full}
The following properties hold for the constructed sequence $\dconf_1,\ldots,\dconf_m,\dconf_{m+1}$:
\begin{itemize}
\item
For every $r: 0\leq r < m$, 
$\dconf_r\dtsomovesto{*}\dconf_{r+1}$,
\item
$\dconf_m\dtsomovesto{*}\dconf_{m+1}$.
\end{itemize}
\end{lem}

\begin{proof}
We show the proof of the lemma
follows directly Lemma~\ref{tso:dtso:middelstep:lemma} and Lemma~\ref{tso:dtso:finalstep:lemma}.
To make the proof understandable,
below we consider a $\fenceop$ transition $\transition=\tuple{\state,\fenceop,\state'}$  
such that
$\conf\tsomovesto{\transition}\conf'$ for some $\conf,\conf'$
 as an atomic read-write transition of the form
$\tuple{\state,\arw(\xvar,v,v),\state'}$ where $v\in\dataset$ is the memory value of variable $\xvar\in \vars$ in  $\conf$. For a given
TSO-computation $\tsocomp$, we can calculate such value $v$ for each $\fenceop$ transition $\tsocomp$. 
\end{proof}

%
%
%
%
%
%
%
%


\begin{lem}
\label{tso:dtso:middelstep:lemma}
If $0\leq r < m$, then
$\dconf_r\dtsomovesto{*}\dconf_{r+1}$.
\end{lem}
\begin{proof}

We are in phase $r$.
Because from the  configuration $\dconf_r$, the memory has not been changed until the transition $\transition_{i_{r+1}}$, we observe that all memory-read transitions of the process $\proc$ between transitions  $\transition_{i_r}$ and $\transition_{i_{r+1}}$ will get values from $\onememoryof{\dconf_r}$  where $\proc\in\procs$.
Therefore, we can execute a sequence of propagation transitions to propagate from the memory to the buffer of process $\proc$ to
full fill it by all messages that will satisfy all memory-read transitions of $\proc$ between  $\transition_{i_r}$ and  $\transition_{i_{r+1}}$.
We propagate to processes according to the order $\prec$: first to the process $\proc_{\it min}$ and last to the process $\proc_{\it max}$. We have the following sequence:
$\dconf_r \dtsomovesto{({\transitions^\propagate})^*}{\dconf_r^{\proc_{\it min}}}
\cdots\dtsomovesto{(\transitions^\propagate)^*} {\dconf_r^{\proc_{\it max}}}$.
The shape of the configuration $\dconf_r^{\proc_{\it max}}$ is:
\begin{itemize}
\item $\twostatemappingof{\dconf_r^{\proc_{\it max}}}\proc=\twostatemappingof{\conf_{\posof{r-1}{\proc}}}\proc$,
\item $\onememoryof{\dconf_r^{\proc_{\it max}}}=\onememoryof{\conf_{i_{r}}}$,
\item $\twobuffermappingof{\dconf_r^{\proc_{\it max}}}\proc={\sf label}^{\sf rev}_\proc({\posof{r-1}{\proc}+1} \cdots {i_{r+1}-1} )$.
\end{itemize}

Below let $\proc=\processof{r+1}$ be the active process in phase $r$ of the DTSO-computation.
For each transition $\transition$ in the sequence of transitions (including updates) of the active process, $seq$=$(\transition_{\posof{r-1}{\proc}+1}\cdots\transition_{\matchingof{i_{r+1}}})|_{\transitions_{\processof{r+1}} \cup\set{\updateop_{\processof{r+1}}}}$,
 we execute a set of transitions in the DTSO-computation as follows:

\begin{itemize}
\item To simulate a memory-read transition, we execute the same read transition. 
And then we execute a delete transition to delete the oldest 
message 
in the buffer of  $\processof{r+1}$.
\item To simulate a read-own-write transition, we execute the same read transition. 
\item
To simulate a write transition, we execute the same write transition. 
This transition must be the transition $\transition_{\matchingof{i_{r+1}}}$. 
According to Dual TSO semantics, we add an own-message to the buffer of  $\processof{r+1}$.
\item
To simulate an $\arw$ transition, we execute the same atomic read-write transition.
This transition must be the transition $\transition_{\matchingof{i_{r+1}}}$ and  $\matchingof{i_{r+1}}=i_{r+1}$.
\item
To simulate an update transition, we execute a delete transition to delete the oldest 
message in the buffer of  $\processof{r+1}$.
\item To simulate a $\nop$ transition, we execute the same transitions in the DTSO-computation.
\end{itemize}

Let $\beta(r,l)$ indicate the index in the TSO-computation of the $l^{th}$ transition   in the sequence 
$seq$ where $1 \leq l \leq |seq|$.
Formally, we define $\beta(r,l):=j$ where $1 \leq j \leq n$,
$\transition_j\in (\transitions_\proc \cup\set{\updateop_{\proc}})$
 and
$seq(l)=\transition_j$.
Let configuration $\dconf_{r,l}$ where $0 \leq r < m$ be the DTSO-configuration
\emph{before} simulating the transition with the index $\beta(r,l)$.
We define $\dconf_{r,l}$ by defining its local states,
buffer contents,
and memory state:
\begin{itemize}
\item $\twostatemappingof{\dconf_{r,l}}\proc$=$\twostatemappingof{\conf_{\posof{r-1}{\proc}}}\proc$ for all inactive process $\proc$ and all $l: 1\leq l \leq |seq|$,
\item $\twostatemappingof{\dconf_{r,l}}\proc=\twostatemappingof{\conf_{\beta(r,l)-1}}\proc$ for the active process $\proc$ and all $l: 1\leq l \leq |seq|$,
\item  $\onememoryof{\dconf_{r,l}}=\onememoryof{\conf_{i_r}}$   for the active process $\proc$ and all $l: 1\leq l \leq |seq|$,
\item $\twobuffermappingof{\dconf_{r,l}}\proc$=${\sf label}^{\sf rev}_\proc({\posof{r-1}{\proc}+1} \cdots {i_{r+1}-1} )$ for all inactive process $\proc$ and all $l: 1\leq l \leq |seq|$,
\item
$\twobuffermappingof{\dconf_{r,l}}\proc$=${\sf label}^{\sf rev}_\proc({\beta(r,l)} \cdots {i_{r+1}-1} )$  for the active process $\proc$  and all $l: 1\leq l \leq |seq|$.
\end{itemize}

The Lemma~\ref{tso:dtso:firststep:lemma}, Lemma~\ref{tso:dtso:midlestep:lemma},
and Lemma~\ref{tso:dtso:endstep:lemma} imply the result.
More precisely, it shows
the existence of a DTSO-computation
that starts from the DTSO-configuration $\dconf^{pmax}_r$ and
whose target is the configuration $\dconf_{r+1}$.
This concludes the proof of Lemma~\ref{tso:dtso:middelstep:lemma}.
\end{proof}

\begin{lem}
\label{tso:dtso:firststep:lemma}
$\dconf_{r,1}$= $\dconf_r^{\proc_{\it max}}$ for $0\leq r < m$.
\end{lem}
\begin{proof}
We show that $\dconf_{r,1}$ and $\dconf_r^{\proc_{\it max}}$ have the same local states,
memory, and buffer contents. We consider two cases for the active and inactive processes.

\begin{itemize}
\item
For inactive process $\proc\neq\processof{r+1}$,
it is easy to see that:
\begin{itemize}
\item
$\twostatemappingof{\dconf_{r,1}}\proc=\twostatemappingof{\conf_{\posof{r-1}{\proc}}}\proc=\twostatemappingof{\dconf^{pmax}_{r}}\proc$ by the definitions of configurations $\dconf_{r,1}$ and $\dconf^{pmax}_{r}$.
\item
$\twobuffermappingof{\dconf_{r,1}}\proc$=${\sf label}^{\sf rev}_\proc({\posof{r-1}{\proc}+1} \cdots {i_{r+1}-1} )
= \twobuffermappingof{\dconf^{pmax}_{r}}\proc$ by the definitions of configurations $\dconf_{r,1}$ and $\dconf^{pmax}_{r}$.
\end{itemize}
\item
For the active process $\proc=\processof{r+1}$:
\begin{itemize}
\item
$\twostatemappingof{\dconf_{r,1}}\proc=\twostatemappingof{\conf_{\beta(r,1)-1}}\proc=
\twostatemappingof{\conf_{\posof{r-1}{\proc}}}\proc$ by the definition of $\beta(r,1)$.
Therefore  $\twostatemappingof{\dconf_{r,1}}\proc=\twostatemappingof{\dconf^{pmax}_{r}}\proc$.
\item
$\twobuffermappingof{\dconf_{r,1}}\proc$=${\sf label}^{\sf rev}_\proc({\beta(r,1)} \cdots {i_{r+1}-1} )$
= ${\sf label}^{\sf rev}_\proc({\posof{r-1}{\proc}+1} \cdots {i_{r+1}-1} )$ by the definition of $\beta(r,1)$.
Therefore $\twobuffermappingof{\dconf_{r,1}}\proc$ = $\twobuffermappingof{\dconf^{pmax}_{r}}\proc$.
\end{itemize}
\end{itemize}
In both cases, 
for the memory,
$\onememoryof{\dconf_{r,1}}=\onememoryof{\conf_{i_r}}=\onememoryof{\dconf^{pmax}_{r}}$ by the definitions of configurations $\dconf_{r,1}$ and $\dconf^{pmax}_{r}$.

This concludes the proof of Lemma~\ref{tso:dtso:firststep:lemma}.
\end{proof}

\begin{lem}
\label{tso:dtso:endstep:lemma}
$\dconf_{r,|seq|}  \dtsomovesto{\transition_{\matchingof{i_{r+1}}}}  \dconf_{r+1}$ for $0\leq r < m$.
\end{lem}

\begin{proof}
To prove the lemma, we will show the following properties:
\begin{enumerate}
\item
$\exists \dconf'_{r+1}: \dconf_{r,|seq|}  \dtsomovesto{\transition_{\matchingof{i_{r+1}}}}  \dconf'_{r+1}$, i.e.\
the transition $\transition_{\matchingof{i_{r+1}}}$
 is feasible from
the configuration $\dconf_{r,|seq|}$.
 \item 
 Moreover, $\dconf'_{r+1}$=$\dconf_{r+1}$.
\end{enumerate}

%

Let $\proc=\processof{r+1}$ be the active process.
We show the property (1) by considering two cases:
\begin{itemize}
\item
$\matchingof{i_{r+1}}$ is a write transition: 
By simulation, we execute the same transition in the DTSO-computation.
It is feasible since 
$\twostatemappingof{\dconf_{r,|seq|}}\proc=\twostatemappingof{\conf_{\beta(r,|seq|)-1}}\proc
= \twostatemappingof{\conf_{\matchingof{i_{r+1}}-1}}\proc$
by the definitions of $\beta(r,|seq|)$ and $\dconf_{r,|seq|}$.
This concludes the property (1).
\item

 $\matchingof{i_{r+1}}$ is an atomic read-write transition: 
We notice that $\matchingof{i_{r+1}} = i_{r+1}$.
It is feasible since 
$\twostatemappingof{\dconf_{r,|seq|}}\proc=\twostatemappingof{\conf_{\beta(r,l)-1}}\proc
= \twostatemappingof{\conf_{\matchingof{i_{r+1}}-1}}\proc$,
$\onememoryof{\dconf_{r,|seq|}}=\onememoryof{\conf_{i_r}}$,
and
\begin{displaymath}
\twobuffermappingof{\dconf_{r,|seq|}}\proc = \twobuffermappingof{\conf_{\beta(r,l)-1}}\proc
= \twobuffermappingof{\conf_{\matchingof{i_{r+1}}-1}}\proc =\emptyword
\end{displaymath}
by the definitions of $\beta(r,|seq|)$ and $\dconf_{r,|seq|}$.
This concludes the property (1).
\end{itemize}

We show the property (2)
 by showing that $\dconf'_{r+1}$ and $\dconf_{r+1}$ have the same local states,
 memory,
 and buffer contents.
  Recall that  
 the $\transition_{\matchingof{i_{r+1}}}$ can be a write transition or an atomic read-write transition.
 
 We consider  inactive processes.
For an inactive process $\proc\neq\processof{r+1}$,
we have:
\begin{itemize}
\item
Since the transition $\transition_{\matchingof{i_{r+1}}}$ is of the active process,
we have $\twostatemappingof{\dconf'_{r+1}}\proc=\twostatemappingof{\dconf^{pmax}_{r}}\proc$.
 Moreover,
 by
the definition
of $\dconf^{pmax}_{r}$, we see that $\twostatemappingof{\dconf^{pmax}_{r}}\proc=\twostatemappingof{\conf_{\posof{r-1}{\proc}}}\proc$ . 
Hence,  by the definition of $\dconf_{r+1}$,
 $$\twostatemappingof{\dconf'_{r+1}}\proc=\twostatemappingof{\dconf_{r+1}}\proc.$$
\item
Since the transition $\transition_{\matchingof{i_{r+1}}}$ is of the active process,
we have
$\twobuffermappingof{\dconf'_{r+1}}\proc=\twobuffermappingof{\dconf^{pmax}_{r}}\proc$.
Moreover, by
the definition
of $\dconf^{pmax}_{r}$, we have $\twobuffermappingof{\dconf^{pmax}_{r}}\proc={\sf label}^{\sf rev}_\proc({\posof{r-1}{\proc}+1} \cdots {i_{r+1}-1} )$.
Hence, by the definition of 
$\dconf_{r+1}$,  $$\twobuffermappingof{\dconf'_{r+1}}\proc)=\twobuffermappingof{\dconf_{r+1}}\proc.$$
\end{itemize}

We consider the active process $\proc=\processof{r+1}$ for the case that the transition $\transition_{\matchingof{i_{r+1}}}$ is a write one.
By executing the same transition,
we add an owing message to the buffer of process $\proc$ and change the memory.
\begin{itemize}
\item
Since the transition $\transition_{\matchingof{i_{r+1}}}$ is of the active process,
we have 
\begin{align*}
\twostatemappingof{\dconf'_{r+1}}\proc=\twostatemappingof{\conf_{\beta(r,l)}}\proc.
\end{align*}
Moreover, 
it follows from the fact 
$\beta(r,l)=\matchingof{i_{r+1}}$ and the definition of $\posof{r}{\proc}$
that 
$\twostatemappingof{\conf_{\beta(r,l)}}\proc
=\twostatemappingof{\conf_{\matchingof{i_{r+1}}}}\proc
=\twostatemappingof{\conf_{\posof{r}{\proc}}}\proc$.
Hence, it follows by the definition of 
$\dconf_{r+1}$
that  $\twostatemappingof{\dconf'_{r+1}}\proc=
\twostatemappingof{\dconf_{r+1}}\proc$.
\item
\sloppypar{Since the transition $\transition_{\matchingof{i_{r+1}}}$ is of the active process,
we have 
$\twobuffermappingof{\dconf'_{r+1}}\proc={\sf label}^{\sf rev}_\proc(i_{r+1}) \app \twobuffermappingof{\dconf_{r, |seq|}}\proc$.
Then, it follows from the definition of $\dconf_{r, |seq|}$
that
$
\twobuffermappingof{\dconf'_{r+1}}\proc=
{\sf label}^{\sf rev}_\proc(i_{r+1}) \app {\sf label}^{\sf rev}_\proc({\beta(r, |seq|)} \cdots {i_{r+1}-1} )
= {\sf label}^{\sf rev}_\proc(i_{r+1}) \app {\sf label}^{\sf rev}_\proc({\matchingof{i_{r+1}}} \cdots {i_{r+1}-1} )
= {\sf label}^{\sf rev}_\proc({\matchingof{i_{r+1}}} \cdots {i_{r+1}})
= {\sf label}^{\sf rev}_\proc({\posof{r}{\proc}} \cdots {i_{r+1}})
={\sf label}^{\sf rev}_\proc({\posof{r}{\proc}+1} \cdots {i_{r+1}})$.
Hence, it follows by the definition of 
$\dconf_{r+1}$
that $\twobuffermappingof{\dconf'_{r+1}}\proc=\twobuffermappingof{\dconf_{r+1}}\proc$.}
\end{itemize}
 
We consider the active process $\proc=\processof{r+1}$ for the case that the transition $\transition_{\matchingof{i_{r+1}}}$ is an atomic read-write one.
By simulation, we execute the same transition and change the memory.
\begin{itemize}
\item
Since the transition $\transition_{\matchingof{i_{r+1}}}$ is of the active process,
we have 
$\twostatemappingof{\dconf'_{r+1}}\proc=\twostatemappingof{\conf_{\beta(r,l)}}\proc$.
Moreover, 
it follows from the fact 
 $\beta(r,l)=\matchingof{i_{r+1}}$ and the definition of $\posof{r}{\proc}$
 that
$\twostatemappingof{\conf_{\beta(r,l)}}\proc= \twostatemappingof{\conf_{\matchingof{i_{r+1}}}}\proc
= \twostatemappingof{\conf_{\posof{r}{\proc}}}\proc$. 
Hence, it follows by
the definition of 
$\dconf_{r+1}$
that $\twostatemappingof{\dconf'_{r+1}}\proc=
\twostatemappingof{\dconf_{r+1}}\proc$.
\item
Since the transition $\transition_{\matchingof{i_{r+1}}}$ is of the active process,
we have 
$\twobuffermappingof{\dconf'_{r+1}}\proc=\emptyword$.
From the definitions of 
$\dconf_{r+1}$ and $\posof{r}{\proc}$ and the fact   $\matchingof{i_{r+1}} = i_{r+1}$, we have
$\twobuffermappingof{\dconf_{r+1}}\proc= {\sf label}^{\sf rev}_\proc({\posof{r}{\proc}+1} \cdots {i_{r+1}})={\sf label}^{\sf rev}_\proc({\matchingof{i_{r+1}}+1} \cdots {i_{r+1}})={\sf label}^{\sf rev}_\proc({i_{r+1}+1} \cdots {i_{r+1}})=\emptyword$.
Hene, it follows that $\twobuffermappingof{\dconf'_{r+1}}\proc$=$\twobuffermappingof{\dconf_{r+1}}\proc$.
\end{itemize}

For both cases, for the memory, we have
$\onememoryof{\dconf'_{r+1}}(\proc)=\onememoryof{\conf_{i_{r+1}}}=\onememoryof{\dconf_{r+1}}$ from the fact that we change the memory by transition $\transition_{\matchingof{i_{r+1}}}$
and by the definition of $\dconf_{r+1}$.
Finally, we have $\dconf'_{r+1}$=$\dconf_{r+1}$.

This concludes the proof of Lemma~\ref{tso:dtso:endstep:lemma}.
\end{proof}

\begin{lem}
\label{tso:dtso:midlestep:lemma}
$\dconf_{r,l} \dtsomovesto{*} \dconf_{r,l+1}$ for $0\leq r < m$, $1\leq l < |seq|$.
\end{lem}
\begin{proof}
The transition $\transition_{\beta(r,l)}$ can be a read-from-memory, read-own-write, nop, update one. 
First, we give our simulation
of the transition $\transition_{\beta(r,l)}$
from the configuration $\dconf_{r,l}$ and show that this simulation is feasible.
We consider different types of the transition  $\transition_{\beta(r,l)}$.
Let process $\proc=\processof{r+1}$ is the active process.
\begin{itemize}
\item
$\transition_{\beta(r,l)}$ is a read-from-memory transition:
By simulation, we execute the same transition in the DTSO-computation. 
Note that under the DTSO semantics, this transition will read an element in the buffers.
Then we delete the 
oldest 
element in the buffer of the active process.
The  transition $\transition_{\beta(r,l)}$ is feasible because 
by  the definition of $\dconf_{r,l}$, we have
$\twostatemappingof{\dconf_{r,l}}\proc=\twostatemappingof{\conf_{\beta(r,l)-1}}\proc$
and 
$\twobuffermappingof{\dconf_{r,l}}\proc={\sf label}^{\sf rev}_\proc({\beta(r,l)} \cdots {i_{r+1}-1} )$.
\item
$\transition_{\beta(r,l)}$ is a nop transition:
By simulation, we execute the same transition in the DTSO-computation. 
The nop transition is feasible because by the definition of $\dconf_{r,l}$,
we have
$\twostatemappingof{\dconf_{r,l}}\proc=\twostatemappingof{\conf_{\beta(r,l)-1}}\proc$.
\item
 $\transition_{\beta(r,l)}$ is a read-own-write read transition:
By simulation, we execute the same transition in the DTSO-computation. 
Observe that 
$\twostatemappingof{\dconf_{r,l}}\proc=\twostatemappingof{\conf_{\beta(r,l)-1}}\proc$.
We show the read-own-write transition is feasible in the DTSO-computation.
In the TSO-computation, this read must get its value from
a write transition $\transition'_1\in\transitions_\proc^{\wop}$ that
has the corresponding update transition $\transition'_2\in\transitions_\proc^{\updateop}$.
According to the TSO semantics,
the write comes and goes out the buffer in FIFO order.
We have the order of these transitions in the TSO-computation:
(i) transition $\transition_{\beta(r,l)}$ is between transitions $\transition'_1$ and $\transition_{\matchingof{i_{r+1}}}$, and
(ii) transition $\transition'_2$ is between transitions $\transition_{\beta(r,l)}$ and $\transition_{\matchingof{i_{r+1}}}$.
Moreover, (iii) there is no other write transition of the same process and the same variable between transitions 
$\transition'_1$ and $\transition_{\beta(r,l)}$.
In the simulation of the DTSO-computation,
when we meet the transition $\transition'_1$ we put an own-message $m$ to the buffer of the active process.
From that we do not meet any write transition to the same variable of the active process until the simulation of transition 
$\transition_{\beta(r,l)}$.
Moreover, the  message $m$ exists in the buffer until the simulation of transition 
$\transition_{\beta(r,l)}$ because the update transition $\transition'_2$ is after the transition 
$\transition_{\beta(r,l)}$.
Therefore the message $m$ is the 
newest 
own-message in the buffer that can match to the read $\transition_{\beta(r,l)}$.
In other words, the  read transition $\transition_{\beta(r,l)}$ is feasible.
\item
$\transition_{\beta(r,l)}$ is an update transition:
By simulation, we delete the 
oldest 
own-message in the buffer of the active process in the DTSO-computation. 
This transition is feasible because
by the definition of $\dconf_{r,l}$, we have
$\twostatemappingof{\dconf_{r,l}}\proc=\twostatemappingof{\conf_{\beta(r,l)-1}}\proc$
and 
$\twobuffermappingof{\dconf_{r,l}}\proc={\sf label}^{\sf rev}_\proc({\beta(r,l)} \cdots {i_{r+1}-1} )$.
\end{itemize}

We have show our simulation of the transition
$\transition_{\beta(r,l)}$ in the DTSO-computation is feasible.
Let $\dconf'_{r,l+1}$ be the configuration in the DTSO-computation after the simulation.
We proceed the proof of the lemma by 
proving that $\dconf'_{r,l+1}$=$\dconf_{r,l+1}$.
To do this, we will show that 
$\dconf'_{r,l+1}$ and $\dconf_{r,l+1}$ have the same local states, memory,
and buffer contents.

We consider  inactive processes. 
For an inactive process $\proc\neq\processof{r+1}$,
we have:
\begin{itemize}
\item
Since in the simulation, we only execute the transition of the active process,
we have
$\twostatemappingof{\dconf'_{r,l+1}}\proc=\twostatemappingof{\dconf^{pmax}_{r}}\proc$. Moreover, by the definition of 
$\dconf^{pmax}_{r}$, we see that $\twostatemappingof{\dconf^{pmax}_{r}}\proc=\twostatemappingof{\conf_{\posof{r-1}{\proc}}}\proc$ . 
Hence, it follows by the definition of $\dconf_{r,l+1}$
that $\twostatemappingof{\dconf'_{r,l+1}}\proc=\twostatemappingof{\dconf_{r,l+1}}\proc$.
\item
Since in the simulation, we only execute the transition of the active process, we have
$\twobuffermappingof{\dconf'_{r,l+1}}\proc=\twobuffermappingof{\dconf^{pmax}_{r}}\proc$.
Moreover, by the  definition of 
$\dconf^{pmax}_{r}$,
we see that $\twobuffermappingof{\dconf^{pmax}_{r}}\proc= {\sf label}^{\sf rev}_\proc({\posof{r-1}{\proc}+1} \cdots {i_{r+1}-1} )$.
Hence, it follows by the definition of $\dconf_{r,l+1}$ that $\twobuffermappingof{\dconf'_{r+1}}\proc=\twobuffermappingof{\dconf_{r,l+1}}\proc$.
\end{itemize}

We consider the active process $\proc\neq\processof{r+1}$ for the case that 
the transition  $\transition_{\beta(r,l)}$ is a read-from-memory one.
From the simulation of $\transition_{\beta(r,l)}$, 
$\twostatemappingof{\dconf'_{r,l+1}}\proc=\twostatemappingof{\conf_{\beta(r,l+1)-1}}\proc$.
We have $\twostatemappingof{\dconf_{r,l+1}}\proc=\twostatemappingof{\conf_{\beta(r,l+1)-1}}\proc$
from the definition of $\dconf_{r,l+1}$.
Furthermore, because we delete the 
oldest 
message in the buffer of the process $\proc$ after we execute the read transition, it follows by the definition of $\dconf_{r,l+1}$ that
$\twobuffermappingof{\dconf'_{r,l+1}}\proc={\sf label}^{\sf rev}_\proc({\beta(r,l+1)} \cdots {i_{r+1}-1} )=\twobuffermappingof{\dconf_{r,l+1}}\proc$.
Finally, by the definitions of $\dconf_{r,l}$ and $\dconf_{r,l+1}$,
we have
$\onememoryof{\dconf'_{r,l+1}}=\onememoryof{\dconf_{r,l}}=\onememoryof{\conf_{i_r}}= \onememoryof{\dconf_{r,l+1}}$.
Hence, it follows that $\dconf'_{r,l+1}$=$\dconf_{r,l+1}$.

We consider the active process $\proc\neq\processof{r+1}$ for the case that  the transition  $\transition_{\beta(r,l)}$ is a nop one.
From the simulation of $\transition_{\beta(r,l)}$, 
$\twostatemappingof{\dconf'_{r,l+1}}\proc=\twostatemappingof{\conf_{\beta(r,l+1)-1}}\proc$.
From the definition of $\dconf_{r,l+1}$,
$\twostatemappingof{\dconf_{r,l+1}}\proc=\twostatemappingof{\conf_{\beta(r,l+1)-1}}\proc$.
From the definitions of $\dconf_{r,l}$ and $\dconf_{r,l+1}$,
$\twobuffermappingof{\dconf'_{r,l+1}}\proc=\twobuffermappingof{\dconf_{r,l}}\proc = {\sf label}^{\sf rev}_\proc({\beta(r,l)} \cdots {i_{r+1}-1} )={\sf label}^{\sf rev}_\proc({\beta(r,l+1)} \cdots {i_{r+1}-1} )=\twobuffermappingof{\dconf_{r,l+1}}\proc$.
Finally, by the definitions of $\dconf_{r,l}$ and $\dconf_{r,l+1}$, we have 
$\onememoryof{\dconf'_{r,l+1}}=\onememoryof{\dconf_{r,l}}=\onememoryof{\conf_{i_r}}=\onememoryof{\dconf_{r,l+1}}$.
Hence, it follows that $\dconf'_{r,l+1}$=$\dconf_{r,l+1}$.

\sloppypar{
We consider the active process $\proc\neq\processof{r+1}$ for the case that  the transition  $\transition_{\beta(r,l)}$ is a read-own-write one.
From the simulation of the transition, we have
$\twostatemappingof{\dconf'_{r,l+1}}\proc=\twostatemappingof{\conf_{\beta(r,l+1)-1}}\proc$.
Then from the definition of  $\dconf_{r,l+1}$,
we have 
$\twostatemappingof{\dconf_{r,l+1}}\proc=\twostatemappingof{\conf_{\beta(r,l+1)-1}}\proc$.
It follows from the definitions of $\dconf_{r,l}$ and $\dconf_{r,l+1}$
that
$\twobuffermappingof{\dconf'_{r,l+1}}\proc=\twobuffermappingof{\dconf_{r,l}}\proc= {\sf label}^{\sf rev}_\proc({\beta(r,l)} \cdots {i_{r+1}-1} )= {\sf label}^{\sf rev}_\proc({\beta(r,l+1)} \cdots {i_{r+1}-1} )=\twobuffermappingof{\dconf_{r,l+1}}\proc$.
Finally, 
by the definitions of $\dconf_{r,l}$ and $\dconf_{r,l+1}$, we have 
$\onememoryof{\dconf'_{r,l+1}}=\onememoryof{\dconf_{r,l}}=\onememoryof{\conf_{i_r}}=\onememoryof{\dconf_{r,l+1}}$ .
Hence, it follows that $\dconf'_{r,l+1}$=$\dconf_{r,l+1}$.
}

We consider the active process $\proc\neq\processof{r+1}$ for the case that  the transition  $\transition_{\beta(r,l)}$ is an update one.
From the simulation of the transition, 
$\twostatemappingof{\dconf'_{r,l+1}}\proc=\twostatemappingof{\conf_{\beta(r,l+1)-1}}\proc$.
We have 
$\twostatemappingof{\conf_{\beta(r,l+1)-1}}\proc=\twostatemappingof{\dconf_{r,l+1}}\proc$
from the definition of $\dconf_{r,l+1}$.
Moreover, we have 
$\twobuffermappingof{\dconf'_{r,l+1}}\proc
= {\sf label}^{\sf rev}_\proc({\beta(r,l+1)} \cdots {i_{r+1}-1} ) = \twobuffermappingof{\dconf_{r,l+1}}\proc$.
Futhermore,  by the definitions of $\dconf_{r,l}$ and $\dconf_{r,l+1}$,
we have 
$\onememoryof{\dconf'_{r,l+1}}=\onememoryof{\dconf_{r,l}}=\onememoryof{\conf_{i_r}}=\onememoryof{\dconf_{r,l+1}}$ .
Hence, it follows that $\dconf'_{r,l+1}$=$\dconf_{r,l+1}$.

This concludes the proof of Lemma~\ref{tso:dtso:midlestep:lemma}.
\end{proof}

\begin{lem}
\label{tso:dtso:finalstep:lemma}
$\dconf_m\dtsomovesto{*}\dconf_{m+1}$.
\end{lem}
\begin{proof}
We are in the final phase $r=m$.
Observe that in this phase we do not have any write and atomic read-write transitions.
 Because from the  configuration $\dconf_m$ until the end of the TSO-computation the memory has not been changed, we observe that all memory-read transitions of a process $\proc\in\procs$ after transitions  $\transition_{i_m}$ get their values from $\onememoryof{\dconf_m}$.
Therefore, we can execute a sequence of propagation transitions to propagate from the memory to buffer of the process $\proc$ to
full fill it by all messages that will satisfy all memory-read transitions of $\proc$ after  $\transition_{i_m}$.
We propagate to processes according to the order $\prec$: first to process $\proc_{\it min}$ and last to process $\proc_{\it max}$. We have the following sequence:
$\dconf_m\dtsomovesto{(\transitions^\propagate)^*}{\dconf_m^{\proc_{\it min}}}\cdots\dtsomovesto{(\transitions^\propagate)^*}{\dconf_m^{\proc_{\it max}}}$.

Next we simulate the remaining transitions 
 $(\transition_{\posof{m-1}{\proc}+1}\cdots\transition_{n})|_{\transitions_{\proc}\cup\set{\updateop_\proc}}$ 
 for each process $\proc$ of the TSO-computation $\comp_{\sf TSO}$ according to the order $\prec$: first process 
 $\proc_{min}$ and last process $\proc_{max}$.

 \begin{itemize}
\item To simulate a memory-read transition, we execute the same read transition. 
And then we execute a delete transition to delete the oldest 
message
in the buffer of the process $\proc$.
\item To simulate a read-own-write transition, we execute the same read transition. 
\item
To simulate an update transition, we execute a delete transition to delete the oldest 
message in the buffer of the process $\proc$.
\item To simulate a $\nop$ transition, we execute the same transitions in the DTSO-computation.
\end{itemize}

Following the same argument as in Lemma~\ref{tso:dtso:firststep:lemma}, Lemma~\ref{tso:dtso:endstep:lemma},
and Lemma~\ref{tso:dtso:midlestep:lemma}
we show that all simulations of transitions are feasible.
As a consequence,
from the configuration $\dconf_m$
we reach  the configuration $\dconf_{m+1}$ where for all $\proc\in\procs$: $\twostatemappingof{\dconf_{m+1}}\proc=\twostatemappingof{\conf_{n}}\proc$,
$\twobuffermappingof{\dconf_{m+1}}\proc$=$\emptyword$,
and 
$\onememoryof{\dconf_{m+1}}=\onememoryof{\conf_{n}}$.

This concludes the proof of Lemma~\ref{tso:dtso:finalstep:lemma}.
\end{proof}



\section{Proof of Lemma~\ref{monotonicity-dtso}}
\label{proofs:monotonicity-dtso}

Let $\conf_i=\tuple{\statemapping_i,\buffermapping_i,\mem_i}$ be DTSO-configurations
for $i:1\leq i\leq 3$.  Let us assume that 
$\conf_1\dtsomovesto{\transition}\conf_2$
for some
$\transition\in\transitions_{\proc}\cup\set{\propagate_{\proc}^{\xvar} , \delete_{\proc}}$
and $\proc\in\procs$.
We will define
$\conf_4=\tuple{\statemapping_4,\buffermapping_4,\mem_4}$
such that 
$\conf_3\dtsomovesto{*}\conf_4$ 
and $\conf_2\cordering\conf_4$.
We consider the following cases depending on $\transition$:
\begin{enumerate}
\item {\sf Nop:}
$\transition=\tuple{\state_1,\nop,\state_2}$.
Define 
$\statemapping_4:=\statemapping_2$,
$\buffermapping_4:=\buffermapping_3$, and
$\mem_4:=\mem_2=\mem_3=\mem_1$.
We have  
$\conf_3\dtsomovesto{\transition}\conf_4$.

\item {\sf Write to memory:}
$\transition=\tuple{\state,\wop(\xvar,\data),\state'}$.
Define 
$\statemapping_4:=\statemapping_2$,
$\buffermapping_4:=
\buffermapping_3\update{\proc}{(\xvar,\data,{\textsc{own}}) \app 
\buffermapping_3(\proc)}$,
and 
$\mem_4:=\mem_2$.
We have  
$\conf_3\dtsomovesto{\transition}\conf_4$.
\item {\sf Propagate:}
$\transition=\propagate_\proc^\xvar$.
Define 
$\statemapping_4:=\statemapping_2$,
$\mem_4:=\mem_2=\mem_3=\mem_1$, and
$\buffermapping_4:=
\buffermapping_3\update{\proc}{(\xvar,\data) \app \buffermapping_3(\proc)}$
where $\data=\mem_4(\xvar)$.
We have  
$\conf_3\dtsomovesto{\transition}\conf_4$.
\item {\sf Delete:}
$\transition=\delete_\proc$.
Define 
$\statemapping_4:=\statemapping_2$
and
$\mem_4:=\mem_2=\mem_3=\mem_1$.
Define $\buffermapping_4$
according to one of the following cases:
\begin{itemize}
\item
If 
$\buffermapping_1=
\buffermapping_2\update{\proc}{\buffermapping_2(\proc) \app (\xvar,\data) }$, 
then define
$\buffermapping_4:=\buffermapping_3$.
In other words, we define $\conf_4:=\conf_3$.
\item
If 
$\buffermapping_1=
\buffermapping_2\update{\proc}{ \buffermapping_2(\proc) \app (\xvar,\data,{\textsc{own}})}$ 
and $(\xvar,\data',{\textsc{own}})\in\buffermapping_2(\proc)$ for some
$\data'\in\valset$,
then define
$\buffermapping_4:=\buffermapping_3$.
In other words, we define $\conf_4:=\conf_3$.
%
%
%
\item
If 
$\buffermapping_1=
\buffermapping_2\update{\proc}{ \buffermapping_2(\proc) \app (\xvar,\data,{\textsc{own}})}$ 
and there is no $\data'\in\valset$ 
such that $(\xvar,\data',{\textsc{own}})\in\buffermapping_2(\proc)$.
Since 
$\buffermapping_1(\proc)\genordering\buffermapping_3(\proc)$,
we know that there is an $i$
and therefore a smallest $i$ such that
$\buffermapping_3(\proc)(i)=(\xvar,\data,{\textsc{own}})$.
Define 
$\buffermapping_4:=
\buffermapping_3\update{\proc}{\buffermapping_3(\proc)(1) 
\app
\buffermapping_3(\proc)(2)  \cdots \buffermapping_3(\proc)(i-1)  }$.
We can perform the following sequence of transitions
$\conf_3
\dtsomovesto{\delete_\proc}
\conf'_1
\dtsomovesto{\delete_\proc}
\conf'_2
\cdots
\dtsomovesto{\delete_\proc}
\conf'_{\sizeof{\buffermapping_3(\proc)}-i}
\dtsomovesto{\delete_\proc}
\conf_4$.
In other words, we reach  the configuration $\conf_4$
from $\conf_3$ by first deleting
$\sizeof{\buffermapping_3(\proc)}-i$ messages from the head of 
$\buffermapping_3(\proc)$.

\end{itemize}
\item {\sf Read:}
$\transition=\tuple{\state,\rop(\xvar,\data),\state'}$.
Define 
$\statemapping_4:=\statemapping_2$,
and
$\mem_4:=\mem_2=\mem_3=\mem_1$.
We define $\buffermapping_4$
according to one of the following cases:
\begin{itemize}
\item  {\sf Read-own-write:}
If there is an $i:1\leq i\leq\sizeof{\buffermapping_1(\proc)}$ such that
$\buffermapping_1(\proc)(i)=\tuple{\xvar,\data,{\textsc{own}}}$, and
there are no $j: 1\leq j <i$ and $\data'\in\dataset$  
such that $\buffermapping_1(\proc)(j)=\tuple{\xvar,\data',{\textsc{own}}}$.
Since $\buffermapping_1(\proc)\genordering\buffermapping_3(h(\proc))$,
there is an $i':1\leq i'\leq\sizeof{\buffermapping_3(\proc)}$ such that
$\buffermapping_3(\proc)(i')=\tuple{\xvar,\data,{\textsc{own}}}$, and
there are no $j: 1\leq j <i'$ and $\data'\in\dataset$  
such that $\buffermapping_3(\proc)(j)=\tuple{\xvar,\data',{\textsc{own}}}$.
Define $\buffermapping_4:=\buffermapping_3$.
In other words,
we define $\conf_4:=\conf_3$.
\item {\sf Read from  buffer:}
If
$\tuple{\xvar,\data',{\textsc{own}}}\not\in\buffermapping_1(\proc)$
for all $\data'\in\dataset$ and $\buffermapping_1(\proc)= w  \app \tuple{\xvar,\data}$.
%
Let $i$ be the largest $i:1\leq i \leq \sizeof{\buffermapping_3(\proc)}$
such that $\buffermapping_3(\proc)(i)=\tuple{\xvar,\data}$.
Since $\buffermapping_1(\proc)\genordering\buffermapping_3(\proc)$,
we know that such index $i$ exists.
Define 
$\buffermapping_4:=
\buffermapping_3\update{\proc}{\buffermapping_3(\proc)(1) 
\app
\buffermapping_3(\proc)(2)  \cdots \buffermapping_3(\proc)(i-1)  }$.
We can perform the following sequence of transitions
$\conf_3
\dtsomovesto{\delete_\proc}
\conf'_1
\dtsomovesto{\delete_\proc}
\conf'_2
\cdots
\dtsomovesto{\delete_\proc}
\conf'_{\sizeof{\buffermapping_3(\proc)}-i}
\dtsomovesto{\delete_\proc}
\conf_4$.
In other words, we reach  the configuration $\conf_4$
from $\conf_3$ by first deleting
$\sizeof{\buffermapping_3(\proc)}-i$ messages from the head of 
$\buffermapping_3(\proc)$.
\end{itemize}
\item {\sf Fence:}
$\transition=\tuple{\state,\fenceop,\state'}$.
Define 
$\statemapping_4:=\statemapping_2$,
$\buffermapping_4:=\emptyword$, and
$\mem_4:=\mem_2$.
We can perform the following sequence of transitions
\begin{align*}
\pconf_3
\dtsomovesto{\delete_\proc}
\pconf'_1
\dtsomovesto{\delete_\proc}
\pconf'_2
\cdots
\dtsomovesto{\delete_\proc}
\pconf'_{\sizeof{\buffermapping_3(\proc)}}
\dtsomovesto{\transition}
\pconf_4.
\end{align*}
In other words, we reach  the configuration $\conf_4$
from $\conf_3$ by first  emptying
the content of
$\buffermapping_3(\proc)$
and then performing $\transition$.

\item {\sf ARW:}
$\transition=\tuple{\state,\arw(\xvar,\data,\data'),\state'}$.
Define 
$\statemapping_4:=\statemapping_2$,
$\buffermapping_4:=\emptyword$, and
$\mem_4:=\mem_2$.
We can reach  the configuration $\conf_4$
from $\conf_3$ in a similar manner to the case of the fence transition.
\end{enumerate}

This concludes the proof of Lemma~\ref{monotonicity-dtso}.
\qed

\section{Proof of Lemma~\ref{well-quasi-order}}
\label{proofs:well-quasi-order}



First we show that the ordering 
$\word \genordering \word'$ is a well-quasi-ordering.
It is an immediate consequence of the fact that
(i) the sub-word relation is a well-quasi-ordering on finite words
\cite{higman:divisibility},
and that (ii)
the number of 
own-messages in the form $(x,v, {\it own})$ 
that should be equal, is finite.

Given
two DTSO-configurations  $\conf=\tuple{\statemapping,\buffermapping,\mem}$ and $\conf'=\tuple{\statemapping',\buffermapping',\mem'}$.
We define three orders $\genordering^{{\it state}}$, $\genordering^{{\it mem}}$,
and $\genordering^{{\it buffer}}$ over  configurations of $\dtsoconfs$:
$\conf \genordering^{{\it state}} \conf'$ iff $\statemapping=\statemapping'$,
$\conf \genordering^{{\it mem}} \conf'$ iff $\mem'=\mem$,
and $\conf \genordering^{{\it bufer}} \conf'$ iff $\buffermapping(\proc)\genordering\buffermapping'(\proc)$ for all process $\proc \in \procs$.

It is easy to see that each one of three orderings is a well-quasi-ordering.
It follows that
the ordering $\genordering$ on DTSO-configurations
based on $\genordering^{{\it state}}$, $\genordering^{{\it mem}}$,
and $\genordering^{{\it buffer}}$ is a well-quasi-ordering.

Since the number of processes, the number of local states, memory content, and the number of own-messages that should be equal are finite,
 it is decidable
whether $\conf_1 \genordering \conf_2$.

This concludes the proof of Lemma~\ref{well-quasi-order}.
\qed

\section{Proof of  Lemma~\ref{compute-pre-DTSO}}
\label{proofs:compute-pre-DTSO}


Consider a DTSO-configuration $\conf=\tuple{\statemapping,\buffermapping,\mem}$.
Let we recall the definition of $\minpre(\set{\conf})$:
$\minpre(\set{\conf}):=\minnof{\preof{\tsys}{ \upclosure{\set{c}} } \cup \upclosure{\set{c}}}$.
%
%
We observe that
$$\minpre(\set{\conf})=
\minn\left(\cup_{\transition\in\transitions\cup\transitions''}
\minn{\setcomp{\conf'}{\conf'\by{t} \conf}}
\cup \set{\conf} \right).$$
For  $t\in\transitions\cup\transitions''$,
we select $\minn{\setcomp{\conf'}{\conf'\by{t} \conf}}$ to be the minimal
set  of all {\it finite} DTSO-configurations of the form 
$\conf'=\tuple{\statemapping',\buffermapping',\mem'}$
such that one of the following properties is satisfied:
\begin{enumerate}
\item {\sf Nop:}
$\transition=\tuple{\state_1,\nop,\state_2}$,
$\statemapping(\proc)=\state_2$ for some $\proc\in\procs$,
$\statemapping'=\statemapping\update{\proc}{\state_1}$,
$\buffermapping'=\buffermapping$, and
$\mem'=\mem$.
\item {\sf Write:}
$\transition=\tuple{\state_1,\wop(\xvar,\data),\state_2}$,
$\statemapping(\proc)=\state_2$ for some $\proc\in\procs$,
$\buffermapping(\proc)=(\xvar,\data,{\textsc{own}})\app\word$  for some $\word$,
$\mem(\xvar)=\data$,
$\mem'(\yvar)=\mem(\yvar)$ if $\yvar\neq\xvar$,
$\statemapping'=\statemapping\update{\proc}{\state_1}$,
 and
one of the following properties is satisfied:
\begin{itemize}
\item
$\buffermapping'=\buffermapping\update{\proc}{\word}$.

\item
$\buffermapping'=\buffermapping\update{\proc}{\word_1\app(\xvar,\data',{\it own})\app\word_2}$ for some $\data'\in \dataset$ where $\word_1\app\word_2=\word$ and $(\xvar,\data'',{\it own}) \notin \word_1$ for all $\data''\in \dataset$.
\end{itemize}

\item {\sf Propagate:}
$\transition=\propagate_\proc^\xvar$ for some $\proc\in\procs$,
$\mem(\xvar)=\data$,
$\statemapping'=\statemapping$,
$\mem'=\mem$, 
$\buffermapping(\proc)=(\xvar,\data) \app \word$
for some $\word$, 
and
$\buffermapping'=\buffermapping\update{\proc}{\word}$.

\item {\sf Read:}
$\transition=\tuple{\state_1,\rop(\xvar,\data),\state_2}$,
$\statemapping(\proc)=\state_2$ for some $\proc\in\procs$,
$\statemapping'=\statemapping\update{\proc}{\state_1}$,
and
$\mem'=\mem$, and one
of the following  conditions is
 satisfied:
\begin{itemize}
\item  {\sf Read-own-write:}
there is an $i:1\leq i\leq\sizeof{\buffermapping(\proc)}$ such that
$\buffermapping(\proc)(i)=\tuple{\xvar,\data,{\it own}}$, and
there are no $j: 1\leq j <i$ and $\data'\in\dataset$  
such that $\buffermapping(\proc)(j)=\tuple{\xvar,\data',{\it own}}$, and
$\buffermapping'=\buffermapping$.

\item {\sf Read from  buffer:}
$\tuple{\xvar,\data',{\textsc{own}}}\not\in\buffermapping(\proc)$
for all $\data'\in\dataset$, 
$\buffermapping(\proc)=\word\app(\xvar,\data)$ for some $\word$,
and
$\buffermapping'=\buffermapping$.

\item {\sf Read from  buffer:}
$\tuple{\xvar,\data',{\it own}}\not\in\buffermapping(\proc)$
for all $\data'\in\dataset$, 
$\buffermapping(\proc)\neq\word\app(\xvar,\data)$ for all $\word$,
and
$\buffermapping'=\buffermapping\update{\proc}{\buffermapping(\proc)\app(\xvar,\data)}$.

\end{itemize}

\item {\sf Fence:}
$\transition=\tuple{\state_1,\fenceop,\state_2}$,
$\statemapping(\proc)=\state_2$ for some $\proc\in\procs$,
$\buffermapping(\proc)=\emptyword$,
$\statemapping'=\statemapping\update{\proc}{\state_1}$,
$\buffermapping'=\buffermapping$, and
$\mem'=\mem$.

\item {\sf ARW:}
$\transition=\tuple{\state_1,\arw(\xvar,\data,\data'),\state_2}$,
$\mem(\xvar)=\data'$,
$\mem'=\mem\update{\xvar}{\data}$,
$\statemapping(\proc)=\state_2$ for some $\proc\in\procs$,
$\buffermapping(\proc)=\emptyword$,
$\statemapping'=\statemapping\update{\proc}{\state_1}$,
$\buffermapping'=\buffermapping$.

\item {\sf Delete:}
$\transition=\delete_\proc$ for some $\proc\in\procs$,
$\statemapping'=\statemapping$,
$\mem'=\mem$.
Moreover,
$(\xvar,\data,{\it own})\notin \buffermapping(\proc)$ for some $\xvar\in\vars$ and all $\data\in\dataset$,
$\buffermapping'=\buffermapping\update{\proc}{\buffermapping(\proc)\app(\xvar,\data',{\it own})}$ for some $\data'\in\dataset$.
\end{enumerate}

This concludes the proof of Lemma~\ref{compute-pre-DTSO}.
\qed

\section{Proof of Lemma~\ref{monotonicity-para}}
\label{proofs:monotonicity-para}

Let 
$\pconf_i=\tuple{\procs_i,\conf_i}$
and
$\conf_i=\tuple{\statemapping_i,\buffermapping_i,\mem_i}$
for $i:1\leq i\leq 4$.
We show that
if $\pconf_1\by{\transition}\pconf_2$ and $\pconf_1\trianglelefteq\pconf_3$
for some
$\transition\in\transitions_{\proc}\cup\set{\propagate_{\proc}^{\xvar} , \delete_{\proc}}$
and
$\proc\in\procs_1$ (note that $\procs_1=\procs_2$)
then the configuration $\pconf_4$ exists
such that $\pconf_3\by{}^{*}\pconf_4$ 
and $\pconf_2\trianglelefteq\pconf_4$.
%
%
First we define $\procs_4$:=$\procs_3$.
Because of $\pconf_1\trianglelefteq\pconf_3$, 
there exists an
injection $h:\procs_1\mapsto\procs_3$
in the ordering 
$\pconf_1\trianglelefteq\pconf_3$.
We define an injection
$h':\procs_2\mapsto\procs_4$
in 
 the ordering 
$\pconf_2\trianglelefteq\pconf_4$
such that
$h=h'$. 
Moreover,
for 
$\proc\in\procs_4$,
let
$\statemapping_4(\proc):=\statemapping_2(h'(\proc))$ if the process $\proc\in\procs_2$,
otherwise
$\statemapping_4(\proc):=\statemapping_3(\proc)$.
We define $\conf_4$ depending on different cases of $\transition$:
\begin{enumerate}
\item {\sf Nop:}
$\transition=\tuple{\state_1,\nop,\state_2}$.
Define 
$\buffermapping_4:=\buffermapping_3$ and
$\mem_4:=\mem_2=\mem_3=\mem_1$.
We have  
$\pconf_3\dtsomovesto{\transition}\pconf_4$.

\item {\sf Write:}
$\transition=\tuple{\state,\wop(\xvar,\data),\state'}$.
Define 
$\buffermapping_4:=
\buffermapping_3\update{h(\proc)}{(\xvar,\data,{\textsc{own}}) \app 
\buffermapping_3(h(\proc))}$
and
$\mem_4:=\mem_2$.
We have  
$\pconf_3\dtsomovesto{\transition}\pconf_4$.
\item {\sf Propagate:}
$\transition=\propagate_\proc^\xvar$.
Define 
$\mem_4:=\mem_2=\mem_3=\mem_1$
and
$\buffermapping_4:=
\buffermapping_3\update{h(\proc)}{(\xvar,\data) \app \buffermapping_3(h(\proc))}$ where $\data=\mem_4(\xvar)$.
We have  
$\pconf_3\dtsomovesto{\transition}\pconf_4$.
\item {\sf Delete:}
$\transition=\delete_\proc$.
Define 
$\mem_4:=\mem_2=\mem_3=\mem_1$.
Define $\buffermapping_4$
according to one of the following cases:
\begin{itemize}
\item
If 
$\buffermapping_1=
\buffermapping_2\update{\proc}{\buffermapping_2(\proc) \app (\xvar,\data) }$, 
then define
$\buffermapping_4:=\buffermapping_3$.
In other words, we have $\pconf_4=\pconf_3$.
\item
If 
$\buffermapping_1=
\buffermapping_2\update{\proc}{ \buffermapping_2(\proc) \app (\xvar,\data,{\textsc{own}})}$ 
and $(\xvar,\data',{\textsc{own}})\in\buffermapping_2(\proc)$ for some
$\data'\in\valset$,
then define
$\buffermapping_4:=\buffermapping_3$.
In other words, we have $\pconf_4=\pconf_3$.
\item
If 
$\buffermapping_1=
\buffermapping_2\update{\proc}{ \buffermapping_2(\proc) \app (\xvar,\data,{\textsc{own}})}$ 
and there is no $\data'\in\valset$ 
with $(\xvar,\data',{\textsc{own}})\in\buffermapping_2(\proc)$,
then since 
$\buffermapping_1(\proc)\genordering\buffermapping_3(h(\proc))$
we know that there is an $i$
and therefore a smallest $i$ such that
$\buffermapping_3(h(\proc))(i)=(\xvar,\data,{\textsc{own}})$.
Define 
$$\buffermapping_4:=
\buffermapping_3\update{h(\proc)}
{\buffermapping_3(h(\proc))(1) 
\app
\buffermapping_3(h(\proc))(2)  \cdots \buffermapping_3(h(\proc))(i-1)  }$$
%
We can perform the following sequence of transitions
$\pconf_3
\dtsomovesto{\delete_\proc}
\pconf'_1
\dtsomovesto{\delete_\proc}
\pconf'_2
\cdots
\dtsomovesto{\delete_\proc}
\pconf'_{\sizeof{\buffermapping_3(h(\proc))}-i}
\dtsomovesto{\delete_\proc}
\pconf_4$.
In other words, we reach  the configuration $\pconf_4$
from $\pconf_3$ by first deleting
$\sizeof{\buffermapping_3(h(\proc))}-i$ messages from the head of 
$\buffermapping_3(h(\proc))$.

\end{itemize}
\item {\sf Read:}
$\transition=\tuple{\state,\rop(\xvar,\data),\state'}$.
Define 
$\mem_4:=\mem_2$.
We define $\buffermapping_4$
according to one of the following cases:
\begin{itemize}
\item  {\sf Read-own-write:}
If there is an $i:1\leq i\leq\sizeof{\buffermapping_1(\proc)}$ such that
$\buffermapping_1(\proc)(i)=\tuple{\xvar,\data,{\textsc{own}}}$, and
there are no $1\leq j <i$ and $\data'\in\dataset$  
such that $\buffermapping_1(\proc)(j)=\tuple{\xvar,\data',{\textsc{own}}}$.
Since $\buffermapping_1(\proc)\genordering\buffermapping_3(h(\proc))$,
there is an $i':1\leq i'\leq\sizeof{\buffermapping_1(\proc)}$ such that
$\buffermapping_1(\proc)(i')=\tuple{\xvar,\data,{\textsc{own}}}$, and
there are no $1\leq j <i'$ and $\data'\in\dataset$  
such that $\buffermapping_1(\proc)(j)=\tuple{\xvar,\data',{\textsc{own}}}$.
Define $\buffermapping_4:=\buffermapping_3$.
In other words, we have that $\pconf_4=\pconf_3$.
\item {\sf Read from  buffer:}
If
$\tuple{\xvar,\data',{\textsc{own}}}\not\in\buffermapping_1(\proc)$
for all $\data'\in\dataset$ and $\buffermapping_1=
\buffermapping_2\update{\proc}{ \buffermapping_2(\proc) \app \tuple{\xvar,\data}  }$, then
let $i$ be the largest $i:1\leq i \leq \sizeof{\buffermapping_3(h(\proc))}$
such that $\buffermapping_3(h(\proc))(i)=\tuple{\xvar,\data}$.
Since $\buffermapping_1(\proc)\genordering\buffermapping_3(h(\proc))$,
we know that such an $i$ exists.
Define 
$$\buffermapping_4:=
\buffermapping_3\update{h(\proc)}
{\buffermapping_3(h(\proc))(1) 
\app
\buffermapping_3(h(\proc))(2)  \cdots \buffermapping_3(h(\proc))(i-1)  }$$
%
We can reach  the configuration $\pconf_4$ from $\pconf_3$ in a similar manner to the last case
of the delete transition.
\end{itemize}
\item {\sf Fence:}
$\transition=\tuple{\state,\fenceop,\state'}$.
Define 
$\buffermapping_4:=\emptyword$ and
$\mem_4:=\mem_2$.
We can perform the following sequence of transitions
$\pconf_3
\dtsomovesto{\delete_\proc}
\pconf'_1
\dtsomovesto{\delete_\proc}
\pconf'_2
\cdots
\dtsomovesto{\delete_\proc}
\pconf'_{\sizeof{\buffermapping_3(h(\proc))}}
\dtsomovesto{\transition}
\pconf_4$.
In other words, we can reach  the configuration $\pconf_4$
from $\pconf_3$ by first  emptying
the contents of
$\buffermapping_3(h(\proc))$
and then performing $\transition$.

\item {\sf ARW:}
$\transition=\tuple{\state,\arw(\xvar,\data,\data'),\state'}$.
Define 
$\buffermapping_4:=\emptyword$ and
$\mem_4:=\mem_2$.
We can reach  the configuration $\pconf_4$
from $\pconf_3$ in a similar manner to the case of the fence transition.
\end{enumerate}

This concludes the proof of Lemma~\ref{monotonicity-para}.
\qed

\section{Proof of Lemma~\ref{compute-pre-parsys}}
\label{proofs:compute-pre-parsys}


Consider a parameterized configuration  $\pconf=\tuple{\procs,\conf}$
with $\conf=\tuple{\statemapping,\buffermapping,\mem}$.
We recall the definition of $\minpre(\set{\pconf})$:
$\minpre(\set{\pconf})$:=$\minnof{\preof{\tsys}{ \upclosure{\set{\pconf}} } \cup \upclosure{\set{\pconf}}}$.
%
We observe that
$$\minpre(\set{\pconf})=
\minn\left(\cup_{\transition\in\transitions\cup\transitions''}
\minn{\setcomp{\pconf'}{\pconf'\by{t} \pconf}}
\cup \set{\pconf} \right).$$
For  $t\in\transitions\cup\transitions''$,
we select $\minn{\setcomp{\pconf'}{\pconf'\by{t} \pconf}}$ 
 to be the minimal
set  of all {\it finite} parameterized configurations of the form 
$\pconf'=\tuple{\procs',\conf'}$ with
$\conf'=\tuple{\statemapping',\buffermapping',\mem'}$
such that one of the following properties is satisfied:
\begin{enumerate}
\item {\sf Nop:}
$\transition=\tuple{\state_1,\nop,\state_2}$,
$\statemapping(\proc)=\state_2$ for some $\proc\in\procs$,
$\procs'=\procs$,
$\statemapping'=\statemapping\update{\proc}{\state_1}$,
$\buffermapping'=\buffermapping$, and
$\mem'=\mem$.
\item {\sf Write:}
$\transition=\tuple{\state_1,\wop(\xvar,\data),\state_2}$,
$\mem(\xvar)=\data$ for some $\data\in\dataset$,
$\mem'(\yvar)=\mem(\yvar)$ if $\yvar\neq\xvar$,
 and
one of the following conditions is satisfied:
\begin{itemize}
\item
$\statemapping(\proc)=\state_2$ for some $\proc\in\procs$, 
$\procs'=\procs$,
$\statemapping'=\statemapping\update{\proc}{\state_1}$,
$\buffermapping'=\buffermapping\update{\proc}{\word}$,
$\buffermapping(\proc)=(\xvar,\data,{\textsc{own}})\app\word$  for some $\word\in\left((\vars\times\dataset) 
\cup (\vars\times\dataset\times \{\textsc{own}\})\right)^*$.

\item
$\statemapping(\proc)=\state_2$ for some $\proc\in\procs$, 
$\procs'=\procs$,
$\statemapping'=\statemapping\update{\proc}{\state_1}$,
$\buffermapping(\proc)=(\xvar,\data,{\textsc{own}})\app\word$  for some $\word\in\left((\vars\times\dataset) 
\cup (\vars\times\dataset\times \{\textsc{own}\})\right)^*$, 
$\buffermapping'=\buffermapping\update{\proc}{\word_1\app(\xvar,\data',{\it own})\app\word_2}$ for some $\data'\in \dataset$ where $\word_1,\word_2\in\left((\vars\times\dataset) 
\cup (\vars\times\dataset\times \{\textsc{own}\})\right)^*$, $\word_1\app\word_2=\word$ and $(\xvar,\data'',{\it own}) \notin \word_1$ for all $\data''\in \dataset$.
In other words, $(\xvar,\data',{\it own})$ is the most recent message to variable $\xvar$ belonging to $\proc$ in the buffer $\buffermapping'(\proc)$. This condition corresponds to the case when we have some messages $(\xvar,\data',{\it own})$ that are hidden by
  the message $(\xvar,\data,{\it own})$ in the buffer $\buffermapping(\proc)$.

\item
$\statemapping(\proc)\neq\state_2$ or $\buffermapping(\proc)\neq(\xvar,\data,{\textsc{own}})\app\word$  for any $\proc\in\procs$, $\word\in\left((\vars\times\dataset) 
\cup (\vars\times\dataset\times \{\textsc{own}\})\right)^*$,
$\procs'=\procs\cup\set{\proc'}$ for some $\proc'\not\in\procs$,
$\statemapping'(\proc')=\state_1$,
$\statemapping'(\proc'')=\statemapping(\proc'')$
if $\proc''\neq\proc'$,
$\buffermapping'(\proc')=
\langle(\xvar_1,\data_1,{\textsc{own}}) | \emptyword \rangle
\langle (\xvar_2,\data_2,{\textsc{own}}) | \emptyword  \rangle \cdots
\langle (\xvar_m,\data_m,{\textsc{own}})  | \emptyword \rangle$
where $\xvar_i\neq\xvar_j$,
$\data_i\in\dataset$, $1 \leq i,j \leq \sizeof{X}$
and
$\buffermapping'(\proc'')=\buffermapping(\proc'')$
if $\proc''\neq\proc'$.
In other words, we add one more process $\proc'$ to the configuration $\pconf'$.

\end{itemize}

\item {\sf Propagate:}
$\transition=\propagate_\proc^\xvar$ for some $\proc\in\procs$,
$\mem(\xvar)=\data$,
$\procs'=\procs$,
$\statemapping'=\statemapping$,
$\mem'=\mem$, 
$\buffermapping(\proc)=(\xvar,\data) \app \word$
for some $\word\in\left((\vars\times\dataset) 
\cup (\vars\times\dataset\times \{\textsc{own}\})\right)^*$, 
and
$\buffermapping'=\buffermapping\update{\proc}{\word}$.

\item {\sf Read:}
$\transition=\tuple{\state_1,\rop(\xvar,\data),\state_2}$,
$\statemapping(\proc)=\state_2$ for some $\proc\in\procs$,
$\procs'=\procs$,
$\statemapping'=\statemapping\update{\proc}{\state_1}$,
and
$\mem'=\mem$, and one
of the following conditions is
 satisfied:
\begin{itemize}
\item  {\sf Read-own-write:}
there is an $i:1\leq i\leq\sizeof{\buffermapping(\proc)}$ such that
$\buffermapping(\proc)(i)=\tuple{\xvar,\data,{\it own}}$, and
there are no $j: 1\leq j <i$ and $\data'\in\dataset$  
such that $\buffermapping(\proc)(j)=\tuple{\xvar,\data',{\it own}}$, and
$\buffermapping'=\buffermapping$.

\item {\sf Read from  buffer:}
$\tuple{\xvar,\data',{\textsc{own}}}\not\in\buffermapping(\proc)$
for all $\data'\in\dataset$, 
$\buffermapping(\proc)=\word\app(\xvar,\data)$ for some $\word\in\left((\vars\times\dataset) 
\cup (\vars\times\dataset\times \{\textsc{own}\})\right)^*$,
and
$\buffermapping'=\buffermapping$.

\item {\sf Read from  buffer:}
$\tuple{\xvar,\data',{\it own}}\not\in\buffermapping(\proc)$
for all $\data'\in\dataset$, 
$\buffermapping(\proc)\neq\word\app(\xvar,\data)$ for any $\word\in\left((\vars\times\dataset) 
\cup (\vars\times\dataset\times \{\textsc{own}\})\right)^*$,
and
$\buffermapping'=\buffermapping\update{\proc}{\buffermapping(\proc)\app(\xvar,\data)}$.
This condition corresponds to the case when we have some  messages $(\xvar,\data)$ that are not explicitly presented  at the head of the buffer $\buffermapping(\proc)$.

\end{itemize}

\item {\sf Fence:}
$\transition=\tuple{\state_1,\fenceop,\state_2}$,
$\statemapping(\proc)=\state_2$ for some $\proc\in\procs$,
$\buffermapping(\proc)=\emptyword$,
$\procs'=\procs$,
$\statemapping'=\statemapping\update{\proc}{\state_1}$,
$\buffermapping'=\buffermapping$, and
$\mem'=\mem$.

\item {\sf ARW:}
$\transition=\tuple{\state_1,\arw(\xvar,\data,\data'),\state_2}$,
$\mem(\xvar)=\data'$,
$\mem'=\mem\update{\xvar}{\data}$,
and one of the following conditions is satisfied:
\begin{itemize}
\item
$\statemapping(\proc)=\state_2$ for some $\proc\in\procs$,
$\buffermapping(\proc)=\emptyword$,
$\procs'=\procs$,
$\statemapping'=\statemapping\update{\proc}{\state_1}$,
$\buffermapping'=\buffermapping$.
\item
$\statemapping(\proc)\neq\state_2$ or
$\buffermapping(\proc)\neq\emptyword$ for any $\proc\in\procs$,
$\procs'=\procs\cup\set{\proc'}$ for some $\proc'\not\in\procs$,
$\statemapping'(\proc')=\state_1$,
$\statemapping'(\proc'')=\statemapping(\proc'')$
if $\proc''\neq\proc'$,
$\buffermapping'(\proc')=\emptyword$, and
$\buffermapping'(\proc'')=\buffermapping(\proc'')$ if $\proc''\neq\proc'$.
In other words, we add one more process $\proc'$ to the configuration $\pconf'$.
\end{itemize}

\item {\sf Delete:}
$\transition=\delete_\proc$ for some $\proc\in\procs$,
$\procs'=\procs$,
$\statemapping'=\statemapping$,
$\mem'=\mem$,
$(\xvar,\data,{\it own})\notin \buffermapping(\proc)$ for some $\xvar\in\vars$ and all $\data\in\dataset$,
$\buffermapping'=\buffermapping\update{\proc}{\buffermapping(\proc)\app(\xvar,\data',{\it own})}$ for some $\data'\in\dataset$.
\end{enumerate}

This concludes the proof of Lemma~\ref{compute-pre-parsys}.
\qed


%
%

%




\end{document}